%% file: arxiv_v2.tex
\documentclass[11pt]{article}

\input{preamble}

\date{}
    
\title{
	On the average-case complexity landscape for \\
    Tensor-Isomorphism-complete problems over finite fields
}
\author{
Tiange Li\thanks{Wuhan University. {\tt 2022302011072@whu.edu.cn}}
       \and 
       Yinan Li\thanks{Wuhan University and Wuhan Institute of Quantum Technology. {\tt Yinan.Li@whu.edu.cn}}
       \and
	Youming Qiao\thanks{University of Technology Sydney. {\tt jimmyqiao86@gmail.com}}
	\and
	Dacheng Tao\thanks{Nanyang Technological University. {\tt dacheng.tao@ntu.edu.sg}}
	\and 
	Yingjie Wang\thanks{Nanyang Technological University. {\tt yingjiewang1201@gmail.com}}
}

\begin{document}
	
\maketitle
\pagenumbering{gobble}
\begin{abstract}
In Grochow and Qiao (\emph{SIAM J. Comput.}, 2023), the complexity class Tensor Isomorphism (TI) was introduced and isomorphism problems for groups, algebras, and polynomials were shown to be TI-complete. In this paper, we study average-case algorithms for several TI-complete problems over finite fields, including algebra isomorphism, matrix code conjugacy, and $4$-tensor isomorphism.

Our main results reveal a rich landscape of average-case complexities of Tensor Isomorphism problems, despite their equivalence in the worst-case setting. 
Over the finite field of order $q$, we devise 
\begin{enumerate}
    \item average-case \emph{polynomial-time} algorithms for algebra isomorphism and matrix code conjugacy that succeed for $1/\Theta(q)$ fraction of inputs;
    \item an average-case \emph{polynomial-time} algorithm for the $4$-tensor isomorphism that succeeds for $1/q^{\Theta(1)}$ fraction of inputs; 
    \item an average-case \emph{polynomial-time} algorithm for the orthogonal $3$-tensor isomorphism that succeed for $1-1/\poly(n)$ fraction of $n\times n\times n$ tensors.
\end{enumerate}

Our main technical contribution is to introduce the spectral properties of random matrices into algorithms for TI-complete problems. This leads to not only new algorithms but also new questions in random matrix theory over finite fields. To settle these questions, we need to extend both the generating function approach as in Neumann and Praeger (\emph{J. London Math. Soc.}, 1998) and the character sum method of Gorodetsky and Rodgers (\emph{Trans. Amer. Math. Soc.}, 2021). 
\end{abstract}

\newpage
\tableofcontents
\newpage

\pagenumbering{arabic}
\setcounter{page}{1}
\section{Introduction}

\subsection{Background}

\subsubsection{Graph Isomorphism: average-case algorithms}

An isomorphism problem asks whether two combinatorial or algebraic structures are essentially the same. The family of isomorphism problems has a long history in theoretical computer science. A flagship isomorphism problem is the Graph Isomorphism (\GrI) problem, which asks if two graphs are the same up to permuting vertices. For a long time, \GrI has been known to be easy in practice thanks to software such as McKay's Nauty \cite{McK80,MP14}, while the best algorithm for \GrI with worst-case analyses ran in time $2^{\tilde O(\sqrt{n})}$ for $n$-vertex graphs due to Babai and Luks \cite{BL83}. Of course, the situation regarding the worst-case complexity of \GrI changed dramatically by Babai's quasipolynomial-time algorithm in 2016 \cite{Bab16}.

Before Babai's breakthrough, the best theoretical justification for \GrI being potentially easy was through average-case algorithms. An average-case algorithm for \GrI requires that one graph is drawn from a probability distribution of graphs, such as the Erd\H{o}s--R\'enyi model, and the other graph is arbitrary. Babai, Erd\H{o}s and Selkow \cite{BES80} presented a polynomial-time algorithm that works for almost all graphs drawn from the Erd\H{o}s--R\'enyi $\mathbb{G}(n, 1/2)$ model, with improvements by Babai and Ku\v{c}era \cite{BK79}. More recently, average-case algorithms for \GrI received renewed attention: in a breakthrough \cite{AKM25}, Anastos, Kwan, and Moore conducted a smoothed analysis of \GrI and  an average-case analysis for canonical labeling for random graphs from $\mathbb{G}(n, p)$ covering all $p$, concluding a line of research \cite{BES80,BK79,bollobas1982distinguishing,czajka2008improved,linial2017rigidity}.

In the 1980's, there were some speculations that \GrI for some graphs could be hard, and hence amenable for use in cryptography \cite{BC86,BY90}. Furthermore, a digital signature can be designed based on \GrI following the Goldreich--Micali--Wigderson zero-knowledge protocol for \GrI \cite{GMW91} and the Fiat--Shamir transformation \cite{FS86}. However, such speculations could not be realized because of the difficulty in identifying such hard-for-isomorphism graphs. Furthermore, Babai's algorithm \cite{Bab16} implies that they are at most quasipolynomial-hard.

\subsubsection{Tensor Isomorphism: complexity and algorithms}\label{subsec:TI-algo}

While cryptography based on graph isomorphism seems unlikely, many cryptographic constructions based on graph isomorphism, such as the Goldreich--Micali--Wigderson zero-knowledge protocol for \GrI, can be adapted to other isomorphism problems. This is an important motivation for the recent active research on tensor isomorphism. In this subsection, we provide an overview of the complexity theoretic and algorithmic aspects of tensor isomorphism, and then turn to cryptography based on tensor isomorphism in Section~\ref{subsec:phenomenon}.

Tensors are multilinear maps, and in algorithms, they are stored as multidimensional arrays. For now, let us focus on $3$-tensors (i.e.~trilinear forms or bilinear maps).
Let $\T(\ell\times m\times n, \F)$ be the linear space of $\ell\times m\times n$ $3$-way arrays over a field $\F$. That is, $\tA=(a_{i,j,k})\in \T(\ell\times m\times n, \F)$ is an $\ell\times m\times n$ cuboid of numbers $a_{i,j,k}\in \F$. 

Tensor isomorphism can be defined by resorting to elementary operations on tensors, in analogy with elementary operations on matrices. First, note that we can slice $\tA$ along any of the three directions to obtain a tuple of matrices. For $n\in \N$, we set $[n]:=\{1,\dots,n\}$. The horizontal slices of $\tA$ are $A^h_i=(a_{i,j,k})_{j\in[m], k\in[n]}\in\M(m\times n, \F)$, the vertical slices of $\tA$ are $A^v_j=(a_{i,j,k})_{i\in[\ell],k\in[n]}\in\M(\ell\times n, \F)$, and the frontal slices are $A^f_k=(a_{i,j,k})_{i\in[\ell],j\in[m]}\in\M(\ell\times m, \F)$. Generalizing elementary operations for matrices, we can define \emph{elementary operations} on tensors as follows: for any of the three directions, multiply one slice by a non-zero scalar, permute two slices, or add one slice to another. Then, two tensors are defined to be \emph{isomorphic} if they are the same up to elementary operations. 

Since elementary operations (for matrices) can be represented as generators of the general linear group, two $3$-tensors $\tA,\tB\in\T(\ell\times m\times n,\F)$ are isomorphic if and only if there exist $T\in\GL(\ell,\F)$, $S\in\GL(m,\F)$ and $R\in\GL(n,\F)$ such that $(T\otimes S\otimes R)\tA=\tB$. In other words, $\tA$ and $\tB$ are in the same \emph{orbit} under the natural group action of $\GL(\ell,\F)\times\GL(m,\F)\times\GL(n,\F)$ on $\T(\ell\times m\times n,\F)$.

\paragraph{Complexity.} In a series of papers \cite{GQarxiv,GQT,TI1,TI2,TI3,TI4,TI5}, the Tensor Isomorphism (TI) complexity class was introduced, and isomorphism problems for several algebraic structures, including (families of) groups, polynomials, and algebras, are shown to be TI-complete. A key observation is that these algebraic structures can be encoded by 3-way arrays (with possibly some structural conditions). Therefore, their isomorphism problems can be formulated as equivalences of 3-way arrays under different group actions.

Formally, let $U, V, W$ be vector spaces such that $U\cong\F^\ell$, $V\cong\F^m$, and $W\cong\F^n$. Let $U^*$ be the dual space of $U$. Let $\otimes, \odot$ and $\wedge$ denote the tensor, symmetric, and exterior products between vector spaces, respectively. Then we have the following list of problems. 
\begin{itemize}
    \item Tensor isomorphism  (\TsI), $U\otimes V\otimes W$: The tensor isomorphism problem asks whether $\tA, \tB\in U\otimes V\otimes W$ are in the same orbit under the natural action of $\GL(U)\times\GL(V)\times\GL(W)$~\cite{TI1}.
    \item Orthogonal tensor isomorphism (\OTI), $U\otimes V\otimes W$: The orthogonal tensor isomorphism problem asks whether $\tA, \tB\in U\otimes V\otimes W$ are in the same orbit under the natural action of $\rO(U)\times\rO(V)\times\rO(W)$~\cite{TI3}.
    \item Group isomorphism (\GpI), $U\wedge U\otimes V$: A long-standing bottleneck for finite group isomorphism is $p$-groups of class $2$ and exponent $p$. Testing the isomorphism of this family of groups is equivalent to asking whether $\tA, \tB\in U\wedge U\otimes V$ are in the same orbit under the natural action of $\GL(U)\times \GL(V)$ with $\F=\F_p$~\cite{Bae38} (see also~\cite[Fact 4.3]{GQ17}). 
    \item Polynomial isomorphism (\PolyIso), $U\odot U\odot U$: Over fields of characteristic not $2$ or $3$, the cubic form isomorphism is equivalent to the orbit problem for $\GL(U)$ acting on $U\odot U\odot U$~\cite{Pat96,AS06}.
    \item Algebra isomorphism (\AlgIso), $U\otimes U\otimes U^*$: Algebra Isomorphism asks whether two (possibly associative or Lie) algebras are isomorphic. Recall that the multiplication of an algebra is recorded as a bilinear map $f:U\times U\to U$. To test if two algebras are isomorphic is equivalent to the orbit problem for $\GL(U)$ acting on $U\otimes U\otimes U^*$.
    For commutative, associative, or Lie algebras, we can restrict ourselves to certain varieties of $U\otimes U\otimes U^*$~\cite{10.1007/978-3-540-31856-9_1,AS06,Kayal2006,GrochowLie}. We shall not assume these conditions in the following.
    
    \item Matrix code conjugacy (\MCConj), $U\otimes U^*\otimes V$: Matrix code\footnote{Following the coding theory convention \cite{CL25}, we shall refer to subspaces of $\M(n, \F)$ as matrix codes. In some other references (such as \cite{LQ17}) these are called matrix spaces.} conjugacy asks whether two subspaces $\cA$, $\cB$ of $\M(n, \F)$ are conjugate, i.e. $\exists T\in\GL(n, \F)$ such that $\cA=T\cB T^{-1}:=\{TB T^{-1}\mid B\in \cB\}$. This is equivalent to the orbit problem of $\GL(U)\times\GL(V)$ acting on $U\otimes U^*\otimes V$~\cite{TI1,CL25}.
\end{itemize}

\paragraph{Algorithms.} There has been nice progress in algorithms for Tensor Isomorphism and related problems over finite fields over the past few years. 
To state the results, it is convenient to introduce the \emph{length} of a tensor as the sum of its side lengths: a tensor in $U\otimes V\otimes W$ with $U\cong\F^\ell$, $V\cong\F^m$, and $W\cong\F^n$ is of length $\ell+m+n$. 

The most notable progress is Sun's breakthrough on testing isomorphism of $p$-groups of class $2$ and exponent $p$ \cite{Sun23}. The improvement of Sun's result in \cite{IMQSZ24} shows that \TsI for tensors of length $L$ over $\F_q$ can be solved in time $q^{\tilde O(L^{1.5})}$. By the linear-length reductions in \cite{TI4}, this implies that all the problems mentioned above can be solved in this time bound. In particular, testing isomorphism of $p$-groups of class $2$ and exponent $p$ of order $N$---widely regarded as a bottleneck case for the group isomorphism problem---can be solved in time $N^{\tilde O(\sqrt{\log N})}$, significantly improving the brute-force barrier $N^{\log N+O(1)}$.

Average-case algorithms for \TsI have been studied, also partly motivated by group isomorphism. Following the study of random graph isomorphism \cite{BES80,BK79,AKM25}, random tensor isomorphism can be formulated as one tensor being random, such as each entry is sampled independently and uniformly from $\F_q$, and the other tensor being arbitrary. In \cite{LQ17,BLQW20}, an average-case algorithm for \TsI in time $q^{O(L)}$ was devised, and it was extended to testing isomorphism of cubic forms, algebras, and multilinear forms \cite{GQT}. Average-case algorithms for \TsI over $\R$ or $\C$ under orthogonal or unitary group actions were recently studied in \cite{CEMQ26}.


\subsection{Motivations: average-case complexity, cryptography, and a recent result in quantum information}
\label{subsec:phenomenon}

\paragraph{Average-case differences despite worst-case equivalences.} It is an intriguing phenomenon that problems that are worst-case equivalent may demonstrate different difficult levels in the average-case setting.

Classical random-graph problems illustrate this distinction.
Finding a Hamiltonian cycle and finding a maximum clique are
polynomial-time Turing equivalent in the worst case~\cite{Karp72}.
Yet, in an Erd\H{o}s--R\'enyi random graph $G(n,1/2)$, a Hamiltonian
cycle can be found in polynomial time with high
probability~\cite{BFF87}, whereas no polynomial-time algorithm is
known to find even a clique of size $(1+\varepsilon)\log_2 n$
for any fixed $0<\varepsilon<1$, although the maximum clique has
size $(2+o(1))\log_2 n$ with high probability. The latter question
is Karp's longstanding open problem; see its equivalent
independent-set formulation in~\cite[Section~1]{GKW25}.

We will see such a contraction between \TsI and \OTI: all average-case algorithms for \TsI run in exponential time, while \OTI admits a polynomial-time average-case algorithm. 
Interestingly, the picture we will see in the tensor isomorphism setting is even finer: there are cases where polynomial-time algorithms exist for $1/O(q)$-fraction of the input over $\F_q$. For example, when $q$ is constant, it means that the algorithm works for a constant fraction of the input. Such a dependence on the field order, together with result on \OTI, suggest that the average-case complexity landscape of Tensor Isomorphism problems.

\paragraph{Cryptography.}
Another motivation for our work comes from the practice of TI-complete problems in cryptography. Recall in Section~\ref{subsec:TI-algo}, TI-complete problems include isomorphism problems for polynomials, groups, and algebras. Despite some nice progress, all the known worst-case and average-case algorithms still run in exponential time. These justify the use of TI-complete problems in cryptography. 

The first use of TI-complete problems in cryptography dates back to the 1990s. In \cite{Pat96}, Patarin proposed using polynomial isomorphism to realise a digital signature scheme based on the celebrated Goldreich--Micali--Wigderson zero-knowledge protocol \cite{GMW91} and the Fiat--Shamir transformation \cite{FS86}. This work led to much research on practical algorithms for polynomial isomorphism, usually using the Gr\"obner basis technique, cf.~\cite{MPG13,BFV13,BFP15}. Note that polynomial isomorphism is in fact a family of problems, depending on the choices of degrees, a single polynomial or a polynomial system, and so on. Therefore, some polynomial isomorphism instance is actually not TI-complete, such as the so-called isomorphism of polynomials with one secret problem, which turns out to be polynomial-time solvable \cite{IQ19}. 

Recently, there has been a renewed interest in TI-complete problems in cryptography in the context of group-action-based cryptography \cite{BY90,JQSY19,AFMP20}. In 2023, the National Institute of Standards and Technology of the US (NIST) issued an additional call for post-quantum digital signature schemes\footnote{See \url{https://csrc.nist.gov/projects/pqc-dig-sig}.}. In response to this call, two round-1 candidates are based on TI-complete problems, one called MEDS \cite{MEDSspecs,MEDS-paper} and the other called ALTEQ \cite{ALTEQ-paper,ALTEQspecs}. MEDS is based on tensor isomorphism (that is, $U\otimes V\otimes W$), while ALTEQ is based on alternating trilinear form isomorphism (that is, $U\wedge U\wedge U$). 

MEDS and ALTEQ motivated new practical algorithms for TI-complete problems from the cryptanalysis perspective \cite{Beu23,DBLP:conf/eurocrypt/NarayananQT24,RS24,CL25}. Some works led to the suspicion that Tensor Isomorphism-complete problems seem not equally hard in practice. 
For example, Couvreur and Levrat commented in \cite[Section 2.3]{CL25} that \MCConj seems easier than \TsI in practice. Their
main justification 
lies in their algorithms for matrix code conjugacy that work for $1/\Theta(q)$ fraction over a random choice of matrix codes. The first algorithm in \cite[Section 4.1]{CL25} is polynomial-time but heuristic, relying on an assumption there \cite[Assumption 2]{CL25}. The second algorithm in \cite[Section 4.2]{CL25} runs in subexponential time. 

From the above, we see that for matrix code equivalence, there are better practical algorithms compared to tensor isomorphism. 
Such algorithms have a practical implication if we use matrix code conjugacy as the basis of a digital signature scheme following Goldreich--Micali--Wigderson and Fiat--Shamir. For such a scheme, a pair of conjugate matrix codes forms the public key, and the conjugation matrix serves as the private key. A pair of conjugate matrix codes is usually---as in the cases of ALTEQ and MEDS---formed by randomly sampling one matrix code and then randomly sampling the other matrix code in the same orbit of the first one. In practice, $q$ could be set to a 20-bit prime as in MEDS \cite{MEDSspecs}, which would mean that one in about 100,000 users would receive such a weak public key. 

\paragraph{Average-case algorithms from quantum information theory.}
A classical problem in quantum information theory is to test whether a multipartite pure quantum state is \emph{local unitary (LU) equivalent} to another. This equivalence is important in the classification of multipartite entanglement~\cite{PhysRevA.82.032121,PhysRevLett.104.020504}. For tripartite pure states, testing LU equivalence is precisely the unitary tensor isomorphism problem studied in~\cite{TI3}.

One may consider the average-case scenario, where the goal is to test whether a Haar-random tripartite state is LU equivalent to an arbitrary tripartite state. The one-party reduced density matrices provide useful invariants: local unitary transformations act on them by conjugation. When each reduced density matrix has distinct eigenvalues, its ordered eigenbasis determines the corresponding local transformation up to a diagonal phase matrix; over $\R$, the analogous freedom consists of diagonal sign matrices. The remaining phase or sign conditions must still be solved. This is the basis of Kraus's canonical-form approach for qubit states~\cite{PhysRevA.82.032121,PhysRevLett.104.020504}, and of the higher-dimensional algorithms in~\cite{CEMQ26}. For Haar-random states in $\C^n\otimes\C^n\otimes\C^n$, the one-party spectra are simple with probability one.

For a numerical algorithm, a quantitative average-case analysis also requires control of the eigenvalue gaps. Recently, Chizewer et al.~\cite{CEMQ26} extended eigenvalue-repulsion estimates for sub-Gaussian Wishart matrices~\cite{christoffersen2025gapssingularvaluessample} and obtained average-case polynomial-time algorithms for orthogonal tensor isomorphism over $\R$ and unitary tensor isomorphism over $\C$. Their analysis uses random matrix theory over these fields, so a different argument is needed over finite fields.



\subsection{Our results: efficient average-case algorithms for some TI-complete problems} 


In Section~\ref{subsec:phenomenon}, we reviewed recent progress in cryptanalysis of schemes based on TI-complete problems and the LU equivalence problem in quantum information theory, which indicates that TI-complete problems could have different difficulty levels in practice despite their worst-case polynomial-time equivalence. This calls for a study of average-case algorithms for TI-complete problems that are efficient for $1/\Theta(q)$ fraction of the inputs, while previous work \cite{LQ17,BLQW20,GQT} focuses mainly on algorithms that work for $1-1/q^{\Omega(n)}$ fraction of the inputs. Note that if $q$ is viewed as a constant, $1/\Theta(q)$ means a constant fraction of the inputs. Of course, with the relaxation of the requirements for the fraction of inputs that the algorithm works on, we expect the running time to be substantially improved. 


Following \cite{LQ17,BLQW20,GQT}, we consider random tensors over $\F_q$ such that each entry is sampled from $\F_q$ independently and uniformly at random. More specifically, a random tensor $\tA=(a_{i,j,k})\in U\otimes V\otimes W$ with $U\cong \F_q^\ell, V\cong\F_q^m, W\cong\F_q^n$ is formed by sampling each entry $a_{i,j,k}$ independently and uniformly at random from $\F_q$. This is also the random model used in cryptographic schemes such as \cite{ALTEQ-paper,MEDS-paper}.

Throughout the average-case statements over finite fields, $q$ is fixed and $n$ tends to infinity. Success-fraction assertions are understood for all sufficiently large $n$, with the threshold allowed to depend on $q$. We retain the dependence on $q$ in the running times and probability estimates.

Our first result is an average-case polynomial-time algorithm for algebra isomorphism. 
\begin{theorem}\label{thm:alg-iso}
Fix a prime power $q$, and let $n$ be sufficiently large depending on $q$.
    Let $U\cong \F_q^n$. There is a randomized $\poly(n, \log q)$-time algorithm for algebra isomorphism: given $\tA, \tB \in U\otimes U\otimes U^*$, decide whether they are isomorphic under the natural action of $\GL(U)$. The algorithm is correct for $1/\Theta(q)$-fraction of $\tA\in U\otimes U\otimes U^*$.    
\end{theorem}
See Section~\ref{sec:alg-iso} for the proof of Theorem~\ref{thm:alg-iso}. For comparison, an average-case algorithm for algebra isomorphism was developed in \cite{GQT}; it runs in time $q^{O(n)}$ and works for ($1-1/q^{\Omega(n)}$)-fraction of $\tA\in U\otimes U\otimes U^*$. The randomization in Theorem~\ref{thm:alg-iso} is due to standard randomized algorithms for polynomial factorization over $\F_q$  \cite{CantorZassenhaus81,vzGG13}.

Our second result is an average-case polynomial-time algorithm for matrix code conjugacy. 
\begin{theorem}\label{thm:matrix-sp-conj}
Fix a prime power $q$, and let $n$ be sufficiently large depending on $q$.
Let $U\cong \F_q^n$ and $V\cong \F_q^n$. There is a randomized $\poly(n, \log q)$-time algorithm for matrix code conjugacy: given $\tA, \tB\in U\otimes U^*\otimes V$, decide whether they are isomorphic under the natural action of $\GL(U)\times\GL(V)$. The algorithm is correct for $1/\Theta(q)$ fraction of $\tA\in U\otimes U^*\otimes V$.
\end{theorem}
See Section~\ref{sec:matrix-sp-conj} for the proof of Theorem~\ref{thm:matrix-sp-conj}. 
A summary of algorithms for matrix code conjugacy with rigorous analyses is in Table~\ref{tab:msc-comparison}. Note that Theorem~\ref{thm:matrix-sp-conj} applies to $n$-dimensional matrix codes in $\M(n, q)$. While this is a restriction, we note that $n$-dimensional matrix codes in $\M(n, q)$ are used in cryptography \cite{MEDS-paper}, and practical algorithms for matrix code equivalence were first devised for such codes \cite{DBLP:conf/eurocrypt/NarayananQT24} and then later improved to any dimension \cite{CL25}. We leave the extension to arbitrary dimension $m$ as an interesting open problem (Problem~\ref{prob:ntom}).


\begin{table}[htpb]
\centering
\setlength{\tabcolsep}{6pt}
\renewcommand{\arraystretch}{1.2}
\begin{tabular}{@{}lccc@{}}
\toprule
\textbf{Regime}  & \textbf{Time} & \textbf{Success prob.\ on random $\tA$}\\
\midrule
Worst case \cite{IMQSZ24} &
$q^{\tilde{O}(n^{1.5})}$ &
$1$ \\

Average case \cite{LQ17,BLQW20}  &
$q^{O(n)}$ &
$1 - 1/q^{\Omega(n)}$ \\

Average case \cite[Lemma 13]{CL25}  &
$q^{O(\sqrt{n}/\log n)}$ &
$1/\Theta(q)$ \\

\textbf{Average case (Theorem~\ref{thm:matrix-sp-conj})} &
\textbf{$\poly(n, \log q)$} &
\textbf{$1/\Theta(q)$} \\
\bottomrule
\end{tabular}
\caption{Algorithms with rigorous analysis for matrix code conjugacy of $n$-dim codes in $\M(n, q)$.}
\label{tab:msc-comparison}
\end{table}

Interestingly, when we turn to $4$-tensor isomorphism, which subsumes and reduces to $3$-tensor isomorphism \cite{TI1}, something nontrivial could be done, as shown in the following result.
\begin{theorem}\label{thm:4-tensor-iso}
Fix a prime power $q$, and let $n$ be sufficiently large depending on $q$.
    Let $U, V, W, X\cong\F_q^n$. There is a randomized $\poly(n, q)$-time algorithm for the isomorphism problem of $U\otimes V\otimes W\otimes X$, such that the algorithm is correct for a $1/\poly(q)$-fraction of $\tA\in U\otimes V\otimes W\otimes X$.
\end{theorem}
See Section~\ref{sec:4-tensor-iso} for the proof of Theorem~\ref{thm:4-tensor-iso}.

We also derive an average-case polynomial-time algorithm for orthogonal tensor isomorphism over finite fields of odd characteristic, which addresses the orthogonal case in odd characteristic of Open Problem~4 in~\cite{CEMQ26}.
\begin{theorem}\label{thm:orthogonal-tensor-iso}
    Let $U, V, W\cong\F_q^n$, where $q$ is odd. There is a randomized $\poly(n,\log q)$-time algorithm for the isomorphism problem of $U\otimes V\otimes W$ under the orthogonal group action, such that the algorithm is correct for a $(1-1/\poly(n))$-fraction of $\tA\in U\otimes V\otimes W$.
\end{theorem}
See Section~\ref{sec:orth-tensor-iso} for the proof of Theorem~\ref{thm:orthogonal-tensor-iso}.


\paragraph{Main message.} Theorems~\ref{thm:alg-iso}, \ref{thm:matrix-sp-conj},~\ref{thm:4-tensor-iso}, and~\ref{thm:orthogonal-tensor-iso} suggest that average-case algorithms for TI-complete problems with different running times can be devised for different success-probability regimes. 
Some works achieve an almost-full success probability $1-1/q^{\Omega(n)}$ but require an exponential running time $q^{O(n)}$. 
Others (including this paper) trade the success probability down to about $1/\Theta(q)$ in order to obtain polynomial time. 
When we restrict the group action from $\GL(n,q)$ to $\rO(n,q)$, we can even achieve success probability $1-1/\poly(n)$ with polynomial running time. 
To make this landscape explicit, we summarize in Table~\ref{tab:avgcase-regimes}
what is currently known.

\begin{table}[h]
\centering
\resizebox{\linewidth}{!}{
\renewcommand{\arraystretch}{1.15}
\begin{tabular}{l|c|c|c|c}
\toprule
\diagbox[width=10em,height=3.5em]{\textbf{Problem}}{\textbf{Succ. prob.}}
& $1-1/q^{\Omega(n)}$
& $1-1/\poly(n)$
& $1/\Theta(q)$
& $1/\poly(q)$ \\
\midrule
3-\TsI
& \makecell{$q^{O(n)}$\\\cite{LQ17,BLQW20}}
& \makecell{$q^{O(n)}$\\\cite{LQ17,BLQW20}}
& Unknown
& Unknown \\[4pt]
\midrule
\AlgIso
& \makecell{$q^{O(n)}$\\\cite{GQT}}
& \makecell{$q^{O(n)}$\\\cite{GQT}}
& \makecell{\textbf{$\poly(n,\log q)$}\\\textbf{(Theorem~\ref{thm:alg-iso})}}
& \makecell{Subsumed by\\$1/\Theta(q)$} \\[4pt]
\midrule
\MCConj
& \makecell{$q^{O(n)}$\\\cite{LQ17,BLQW20}}
& \makecell{$q^{O(n)}$\\\cite{LQ17,BLQW20}}
& \makecell{\textbf{$\poly(n,\log q)$}\\\textbf{(Theorem~\ref{thm:matrix-sp-conj})}}
& \makecell{Subsumed by\\$1/\Theta(q)$} \\[4pt]
\midrule
4-\TsI
& \makecell{$q^{O(n)}$\\\cite{LQ17,BLQW20}}
& \makecell{$q^{O(n)}$\\\cite{LQ17,BLQW20}}
& Unknown
& \makecell{\textbf{$\poly(n,q)$}\\\textbf{(Theorem~\ref{thm:4-tensor-iso})}} \\[4pt]
\midrule
\OTI
& \makecell{\textbf{$q^{O(n)}$}\\\textbf{(Corollary~\ref{cor:orth-exponential-success})}}
& \makecell{\textbf{$\poly(n,\log q)$}\\\textbf{(Theorem~\ref{thm:orthogonal-tensor-iso})}}
& \makecell{Subsumed by\\$1-1/\poly(n)$}
& \makecell{Subsumed by\\$1-1/\poly(n)$} \\[2pt]
\bottomrule
\end{tabular}}
\caption{Average-case algorithms for tensor isomorphism and related problems under different success-probability regimes. For \MCConj, we consider $n$-dimensional codes in $\M(n,q)$; for \OTI, $q$ is odd.}
\label{tab:avgcase-regimes}
\end{table}

From the cryptography viewpoint, Theorems~\ref{thm:alg-iso} and~\ref{thm:matrix-sp-conj} address isomorphism problems of $U\otimes U\otimes U^*$ and $U\otimes U^*\otimes V$, but do not give anything more useful to $U\otimes V\otimes W$ or $U\wedge U\wedge U$. This indicates that $U\otimes V\otimes W$ and $U\wedge U\wedge U$ would be more preferred for cryptographic purposes, as explained at the end of Section~\ref{subsec:phenomenon}. This also poses the study of average-case algorithms for $1/\Theta(q)$-fraction of inputs for these cases as interesting open problems.


\subsection{Overview of the techniques}

For algebra and matrix code isomorphism, our main technical contribution is to introduce spectral properties of random matrices over finite fields into average-case algorithms for TI-complete problems. 

Note that spectral properties of matrices over finite fields behave quite differently from those of matrices over $\C$. Let $A\in\M(n, q)$ be an $n\times n$ matrix over $\F_q$. The roots of $\det(A-x \idmat_n)$ as a polynomial in $\F_q[x]$ are the eigenvalues of $A$. Note that it is possible that eigenvalues of $A$ lie in extension fields of $\F_q$. We are interested in those eigenvalues lying in the original field $\F_q$ with algebraic multiplicity $1$. In the following, we shall explain how these eigenvalues in $\F_q$ are used in algorithms, and how algorithms motivate new questions in random matrix theory over finite fields. 

For 4-tensor isomorphism over $U\otimes V\otimes W\otimes X$ where $U\cong V\cong W\cong X\cong \F_q^n$, we observe yet another global isomorphism invariant that holds with probability $1/q^{c}$ for some constant $c$. From this global isomorphism invariant, the algorithm follows relatively easily from some techniques for algebra and matrix code isomorphism. Therefore, we focus on the latter two problems in the introduction.

\subsubsection{Algorithmic techniques}\label{subsec:remark-design} 

\paragraph{Review of previous techniques.} In the literature, there are two main techniques for average-case algorithms for TI-complete problems. 

The first one, used in \cite{LQ17,BLQW20,GQT}, is inspired by the classical individualization and refinement technique for graph isomorphism (see e.g.~\cite{BES80}). Briefly speaking, the idea there is to examine the effect of guessing a constant number of columns of the transformation matrix, which usually leads to determining the full matrix. The guessing of these columns contributes to the $q^{O(n)}$ running time of these algorithms in \cite{LQ17,BLQW20,GQT}. 

The second one, used in \cite{CL25} and originated from \cite{Sen02}, is to start with some global isomorphism invariant. After this step, one may employ either polynomial solving or linear algebraic structures for the algorithm design. The global isomorphism invariant is non-trivial with $1/\Theta(q)$ fraction of the tensors which goes into the average-case analysis. 

\paragraph{The main innovation in our algorithms.} We now briefly explain our algorithm for algebra isomorphism. Suppose we wish to test if two algebras $\tA, \tB\in U\otimes U\otimes U^*$, $U\cong\F_q^n$, are isomorphic. It is natural to start with the global isomorphism invariant in \cite{CL25}, which gives a distinguished matrix $D_\tA\in U\otimes U^*$ for a $1/\Theta(q)$ fraction of tensors $\tA$. 

Following the individualization and refinement strategy as in \cite{LQ17,BLQW20,GQT}, we wish to extract more information from the canonical matrix $D_\tA$. From the experience in \cite{LQ17,BLQW20,GQT}, one canonical matrix does not produce enough constraints on the underlying transformation matrices, and we need a constant number. The goal becomes to deduce another ``canonical'' matrix in $U\otimes U^*$. 

This is where our main innovation, the spectral properties of $D_\tA$, comes to our aid. Suppose $D_\tA$ has a unique eigenvalue in $\F_q$ of algebraic multiplicity $1$. The $1$-dimensional eigenspace associated with $D_\tA$ is then another isomorphism invariant, from which we can extract further information, leading to relations in the matrix tuple conjugacy form of $TA_iT^{-1}=B_i$ for $A_i, B_i\in\M(n, q)$, $i\in[c]$ as in \cite{LQ17,BLQW20,GQT}. 
This allows the individualization-and-refinement idea to match up a constant number of pairs of matrices. 

We now outline the average-case algorithm for the algebra isomorphism problem; see Section~\ref{sec:alg-iso} for details.


\paragraph{Set up for algebra isomorphism.} Let $U\cong\F_q^n$, and let $\tA, \tB\in U\otimes U\otimes U^*$ be two algebras. In algorithms, after fixing a basis of $U$, $\tA$ is represented by its structure constants, that is, (by some abuse of notation) a $3$-way array $\tA=(a_{i,j,k})\in \T(n\times n\times n, q)$, so the bilinear map $g_\tA:\F_q^n\times\F_q^n\to\F_q^n$ represented by $\tA$ is defined by $g_\tA(\stdb_i,\stdb_j)=\sum_{k\in[n]}a_{i,j,k}\stdb_k$. We see that $\tA$ and $\tB$ are isomorphic as algebras if and only if there exists $T\in\GL(n, q)$ such that $(T\otimes T\otimes T^{-\mathrm{t}})\tA=\tB$.

Take the horizontal slices of $\tA$ as $\vA=(A_1, \dots, A_n)\in\M(n, q)^n$, where $A_i(j, k)=a_{i,j,k}$. Similarly, do this for $\tB$ and get $\vB=(B_1, \dots, B_n)\in\M(n, q)^n$. Then, the isomorphism testing between $\tA$ and $\tB$ as algebras translates to asking if there exists $T=(t_{i,i'})\in\GL(n, q)$, such that  
\begin{equation}\label{eq: algebra iso group action}    
\sum_{i'\in[n]} t_{i,i'}T A_{i'} T^{-1}=B_i~\forall\ i\in[n].
\end{equation}
Let $\cA=\linspan\{A_1, \dots, A_n\}\leq \M(n, q)$ and $\cB=\linspan\{B_1, \dots, B_n\}\leq \M(n, q)$. By the above, a necessary condition for $\tA$ and $\tB$ to be isomorphic is that $\cA$ and $\cB$ are conjugate as matrix codes. 

\paragraph{Hulls: a global isomorphism invariant.} We now make use of a global isomorphism invariant drawn from the study of matrix code conjugacy \cite{Sen02,CL25}. Define a bilinear form $f$ on $\M(n, q)$ as $f(A, B)=\Tr(AB)$ for any $A, B\in\M(n, q)$. The \emph{orthogonal space} of $\cA\leq\M(n, q)$ with respect to $f$ is $\cA^\perp:=\{A'\in \M(n, q)\mid \forall A\in\cA, f(A, A')=0\}$. The \emph{hull} of $\cA$ is $\hull(\cA):=\cA\cap\cA^\perp$. It is not difficult to see that the hull of $\cA$ is preserved under matrix code conjugacy, that is, if $\cA$ and $\cB$ are conjugate, then $\hull(\cA)$ and $\hull(\cB)$ are conjugate. 

We need a probability estimate for the hull of a random matrix code, following the use of hulls in code equivalence~\cite{Sen02,CL25}. Proposition~\ref{prop:hull-dim} shows that, for a uniform $n$-dimensional code $\cA\leq\M(n,q)$, with probability $\Theta(1/q)$ its hull is $1$-dimensional, so $\hull(\cA)=\linspan\{A\}$ for some $A\in\M(n,q)$. For $\cA$ and $\cB$ to be conjugate, it is necessary that $\hull(\cA)$ and $\hull(\cB)$ also be conjugate. Therefore, $\hull(\cB)=\linspan\{B\}$ is also $1$-dimensional, and any $T\in\GL(n,q)$ that induces algebra isomorphism must satisfy $T\linspan\{A\}T^{-1}=\linspan\{B\}$. In the following, we assume that this is the case. Note that by definition, $A$ satisfies $\Tr(A^2)=0$.

\paragraph{Exploiting the $1$-dimensional hull.} We seek a hull generator having exactly one eigenvalue in $\F_q$, which is nonzero and has algebraic multiplicity $1$. Dividing the generator by this eigenvalue gives a distinguished matrix $A$ whose only $\F_q$-eigenvalue is $1$. Define $B$ in the same way from the second hull. Any algebra isomorphism must conjugate $A$ to $B$ and match their one-dimensional left eigenspaces. Our random-matrix estimates show that the required spectral condition holds with constant conditional probability for a uniform one-dimensional hull, for fixed $q$ and sufficiently large $n$.

\paragraph{Matching a pair of matrices.} Let $v$ and $u$ be nonzero left eigenvectors of $A$ and $B$ with eigenvalue $1$. For a potential algebra isomorphism $T$, we have $v=cuT$ for some $c\in\F_q^\times$. Contracting the slices with these vectors gives
\[
    X=\sum_i v_iA_i,\qquad Y=\sum_i u_iB_i,
    \qquad TXT^{-1}=cY.
\]
When $\Tr(X)\ne0$, divide $X$ by its trace and do the same for $Y$. This removes the scalar ambiguity in choosing the eigenvectors. Thus every algebra isomorphism conjugates the pair $(A,X/\Tr(X))$ to $(B,Y/\Tr(Y))$. If the pair generates the full matrix algebra, its commuting algebra consists of scalar matrices, as in Corollary~\ref{cor:conj-bound}. Matrix-tuple conjugacy then determines the potential algebra isomorphism up to a scalar, which is recovered by substituting into Eq.~(\ref{eq: algebra iso group action}). The algorithm uses one contraction, with no spectral test on the contracted matrix, and runs in time $\poly(n,\log q)$.

\paragraph{Conditioning in the average-case analysis.} The hull generator and the contracted matrix are dependent, so the independent-pair generation theorem of Neumann and Praeger~\cite{NP95} cannot be applied directly. Instead, we condition on the matrix code. Its ordered slice basis is uniform, making $X$ a uniform nonzero element of that code. To analyse this distribution, fix the hull line $L$ and sample an $n$-dimensional code $\cA$ subject to $L\subseteq\cA\subseteq L^\perp$. An incidence count, using Proposition~\ref{prop:hull-dim}, shows that the additional event $\hull(\cA)=L$ has probability bounded below by a positive absolute constant for fixed $q$ and sufficiently large $n$. Except for an exponentially small probability that $X\in L$, this comparison bounds the density of the hull--contraction pair by $O(q^2)$ times that of two independent uniform matrices. The generation estimate therefore still applies after this conditioning. A separate linear-functional count shows that $\Tr(X)\ne0$ with constant conditional probability. Proposition~\ref{prop:hull-contraction} gives these estimates.

\paragraph{An intrinsic matrix for matrix code conjugacy.} For matrix codes, the change of slice basis is independent of the conjugating matrix, so the preceding eigenvector contraction does not give a code invariant. Instead, normalize the hull generator $A$ as above and let $P$ be its rank-one spectral projector. The complement $D=\cA\cap P^\perp$ has a nondegenerate trace pairing. We define $Z\in D$ by $\Tr(ZY)=\Tr(Y)$ for every $Y\in D$, using its Gram system. This matrix is intrinsic to the code and transforms by conjugation. The average-case proof again uses subspace incidence counts: after fixing the hull, $D$ is a uniform subspace in a nondegenerate ambient space, conditioned on nondegeneracy. Viewing $Z$ as an orthogonal projection, we bound its point probabilities by counting subspaces that contain a specified vector and lie in a specified hyperplane. Apart from two exceptional projections of exponentially small probability, the joint point probabilities of $(H,Z)$, where $H$ is an unnormalized uniform generator of the same hull line, are bounded by $O(q^4)$ times those of an independent uniform pair. Theorem~\ref{thm:gen-matrix-alg} then supplies the scalar-centralizer bound needed for the algorithm. Section~\ref{sec:matrix-sp-conj} gives the construction and Proposition~\ref{prop:matrix-code-projection} proves the distribution estimate.

The spectral estimate for the hull itself must also account for its trace constraint $\Tr(A^2)=0$. To obtain it, we combine the generating-function method for eigenvalue-free matrices of Neumann and Praeger~\cite{NEUMANN_PRAEGER_1998} with the character-sum method for traces of matrix powers of Gorodetsky and Rodgers~\cite{MR4273172}. We now explain this part of our average-case analysis.

\subsubsection{Average-case analysis}\label{subsec:remark-analysis}

As briefly explained in Section~\ref{subsec:remark-design}, a key technical result for the average-case analysis of our algorithm is to show that when the hull is $1$-dimensional spanned by $A\in\M(n, q)$, $A$ has a unique eigenvalue in $\F_q$ with algebraic multiplicity $1$ with constant probability. To the best of our knowledge, this question had not been studied before in random matrix theory over finite fields.

\paragraph{Some works on enumerating eigenvalue-free matrices.} The first evidence that this happens is a classical result from random matrix theory over finite fields. In~\cite{NEUMANN_PRAEGER_1998} Neumann and Praeger studied eigenvalue-free matrices in $\GL(n, q)$, which are invertible matrices with no eigenvalues in $\F_q$; such matrices are linear algebraic analogues of derangements, namely those permutations with no fixed points. They exhibited the generating function of the number of eigenvalue-free matrices in $\GL(n, q)$. To achieve this, one can utilize the primary decomposition to count the number of matrices in $\GL(n,q)$, and establish a relationship between the generating function of the number of such matrices with the generating function of the number of unipotent matrices, the latter of which has been extensively studied (see e.g.~\cite{MR96677} and \cite[Thm. 15.1]{MR230728}).

\paragraph{Generalizing to enumerating matrices with fixed eigenvalue multiplicity profile.}
Our first technical contribution is to generalize the result of Neumann and Praeger to count the number of matrices in $\M(n, q)$ with a given eigenvalue multiplicity profile. Namely, let $\phi:\F_q\to \N$ be a function that assigns to each element in $\F_q$ a non-negative integer. We say $A$ has eigenvalue multiplicity profile $\phi$ if for each $\lambda\in\F_q$, $\lambda$ is an eigenvalue of $A$ with algebraic multiplicity $\phi(\lambda)$. Using the generating functions for eigenvalue-free and unipotent matrices, we explicitly count the matrices in $\M(n,q)$ with eigenvalue multiplicity profile $\phi$ in Eq.~(\ref{eq:expression of Cmnq}). Their limiting probabilities are given in Theorem~\ref{thm:eigenvalue-multiplicity-profile}.

We then use such a counting result to prove that, with a constant probability, a random matrix $A$ has a unique eigenvalue in $\F_q$ with algebraic multiplicity $1$ (Theorem~\ref{thm:eigen-condition}). To see this, note that such eigenvalue multiplicity structure corresponds to the function $\phi$ such that $\phi(\lambda)=1$ for some $\lambda\in\F_q$, and $\phi(\mu)=0$ for any $\mu\in\F_q\setminus\{\lambda\}$. Therefore, we explicitly count the number of matrices with such eigenvalue multiplicity structure, and lower bound the proportion of such matrices in $M(n,q)$ by some absolute constant.

While we only used a special profile in this work, we hope that the generalization to arbitrary profiles in Theorem~\ref{thm:eigenvalue-multiplicity-profile} would be handy in other settings.

\paragraph{Equidistribution of trace powers with eigenvalue conditions.} However, for $A$ to be in the hull, an a priori condition is $\Tr(A^2)=0$. We call such a matrix $A$ \emph{self-dual}. The self-dual condition creates a technical difficulty that seems beyond the generating function method as in \cite{NEUMANN_PRAEGER_1998} or the cycle index method as in \cite{Kung81,Sto88,Ful02}. Note that the probability of a random matrix $A\in\M(n,\C)$ that satisfies certain spectral and trace power properties has been extensively studied in random matrix theory over complex fields (see, e.g.,~\cite{10.1214/20-PS346,10.1214/21-AOP1520,DS94} and the references therein). It is natural to expect that random matrices over finite fields behave similarly to those over real or complex fields. However, powerful tools like the method of moments and the representation theory of $\GL(n,\C)$ are missing in the finite field setting.

Fortunately, inspired by the work of Diaconis and Shahshahani \cite{DS94}, Gorodetsky and Rodgers \cite{MR4273172} and Gorodetsky and Kovaleva \cite{GK24} studied the equidistribution of trace powers over finite fields using the \emph{character sum method} from analytic number theory. Our second technical contribution on the average-case analysis is to extend the character sum method introduced by Gorodetsky and Rodgers \cite{MR4273172} to compute the desired probabilities of matrices with both spectral and self-dual conditions.




Note that in the literature on random matrix theory over finite fields, the results are usually stated in the range of $q\to \infty$, or if $q$ is fixed for $n\to \infty$. The most precise results are often taken in the form of the convergence probability for fixed $q$, and the error term between the convergence probability and fixed $n$. All our results on random matrices over finite fields are of this precision.

\paragraph{Some technical aspects of the character sum method.} 
Our first step is to show that, in order to study $\Tr(A^2)$ where $A$ is a matrix with a unique eigenvalue in $\F_q$, we can investigate matrices with no eigenvalues in $\F_q$. This reduction is presented in the first part of Theorem~\ref{thm:main-technical}, while we look into matrices with no eigenvalues in $\F_q$  in Theorem~\ref{thm: probability of a random eigenvalue-free matrix with trace conditions}. Our goal is to show that for these matrices $A$, $\Tr(A^2)$ is equidistributed. This is achieved via the character sum method following Gorodetsky and Rodgers  in \cite{MR4273172}. 

We briefly explain how characters of finite fields connect to random matrices. 

For a finite field $\F_q$ of characteristic $p$, let $\Tr_{\F_q/\F_p}: \F_q\to \F_p$ be the absolute trace function from $\F_q$ to $\F_p$. Then the function $\psi:~\F_q\to \C$ $
    \psi(a)=e^{\frac{2\pi i}{p}\Tr_{\F_q/\F_p}(a)}$
is a character of the \emph{additive group} of $\F_q$. Every additive character of $\F_q$ is of the form $\psi_b:~\F_q\to \C$ defined as $\psi_b(a)=\psi(ba)$ for all $a\in\F_q$ for $b\in\F_q$.

We first connect characters to polynomials. Let $f(t)$ be a monic polynomial of degree $n$ over $\F_q$ in the variable $t$. Let $\cM_{n,q}$ be the set of all such polynomials. For $\lambda\in\F_q$, we define a function $\chi_{\lambda}:\cM_{n,q}\to\C$ by 
\[
\chi_\lambda(f(t))= e^{\frac{2\pi i}{p}\Tr_{\F_q/\F_p}(\lambda(\alpha_1^2+\cdots+\alpha_n^2))}
\]
where $\alpha_1,\dots,\alpha_n$ are the roots of $f(t)$ in $\overline{\F_q}$, listed with multiplicities.

We now link to matrices. Let $f(t)=c_A(t)$ be the characteristic polynomial of some matrix $A\in \M(n,q)$. It is not hard to see that $\chi_\lambda(c_A(t))=\psi_\lambda(\Tr(A^2))$.
Recall that we are interested in learning the probability of $\Tr(A^2)=k$ for $A$ with certain eigenvalue conditions. It turns out that for any matrix $A\in \M(n,q)$, the indicator function of $\Tr(A^2)=k$ can be expressed as 
    \begin{equation*}\frac{1}{q}\sum_{\lambda\in\F_q}\chi_\lambda(c_A(t))\overline{\psi_\lambda(k)}=
         \left \{
         \begin{aligned}
             1 &  & \text{if}~\Tr(A^2)=k\\
             0 &  & \text{if}~\Tr(A^2)\neq k
         \end{aligned}
         ~.
         \right.
    \end{equation*}
It is useful to introduce the L-function over $\GL$ for $\chi_\lambda$ as
\[
L_{\GL}(u,\chi_\lambda)=\sum_{n=0}^\infty\sum_{f\in\cM_{n,q}^{gl}}\chi_\lambda(f)u^{n}
\]
where $\cM_{n,q}^{gl}$ is the set of characteristic polynomials of matrices in $\GL(n, q)$. By carefully analyzing $L_{\GL}(u,\chi_\lambda)$, equidistribution of $\Tr(A^2)$ can then be deduced with some effort as nicely carried out by Gorodetsky and Rodgers  in \cite{MR4273172}. 

Getting back to our setting, we use $\cal E$ to denote the set of matrices with no eigenvalues in $\F_q$. We introduce $L_{\cal E}(u,\chi_\lambda)$, the L-function over $\cal E$ for $\chi_\lambda$. By carefully analyzing the difference between $L_\GL$ and $L_\cE$ and making use of the methods and results in \cite{MR4273172}, we deduce that equidistribution of $\Tr(A^2)$ over $A\in\cal E$ still holds.

\paragraph{Random Gram matrices and orthogonal centralizers.} For orthogonal tensor isomorphism, we need to analyse the Gram matrix $S_\tA=M_\tA\trp{M_\tA}$, where $M_\tA\in\M(n\times n^2,q)$ is the flattened version of a uniform tensor $\tA$ and $q$ is odd. Although $S_\tA$ is symmetric, its distribution differs from the uniform distribution over symmetric matrices. Our first step is to compare these two distributions. Writing $S_\tA$ as a sum of independent outer products, we use additive characters and quadratic character sums to compute their $\chi^2$-divergence and show that their total variation distance is at most $\frac34q^{-n(n-1)/2}$ (Proposition~\ref{prop:orth-gram-uniform}). We then study the orthogonal centralizer $C_{\mathrm O}(S)=\{Q\in\mathrm O(n,q):QS=SQ\}$ for a uniform symmetric matrix $S$. Using an orthogonal cyclic decomposition and generating functions to count orthogonal conjugacy classes, together with the orbit--stabilizer theorem, we show that $\mathbb E_S|C_{\mathrm O}(S)|\leq12(n+1)$ (Proposition~\ref{prop:orth-centralizer-moment}). Markov's inequality and the distribution comparison then imply that $C_{\mathrm O}(S_\tA)$ has size $n^{O_s(1)}$ with probability $1-O_s(n^{-s})$ for every fixed positive integer $s$. This bounds the number of first-factor candidates that our algorithm needs to enumerate.

\subsection{Outlooks and open problems}

Our main contributions are average-case polynomial-time algorithms for some TI-complete problems over a finite field $\F_q$ that are efficient for $1-1/\poly(n)$, $1/\Theta(q)$ or $1/\poly(q)$ fraction of the inputs. These TI-complete problems include algebra isomorphism, matrix code conjugacy, 4-tensor isomorphism and orthogonal tensor isomorphism. We do not know how to extend these results to 3-tensor isomorphism ($U\otimes V\otimes W$), cubic form equivalence ($U\odot U\odot U$), and alternating trilinear form equivalence ($U\wedge U\wedge U$): the main reason is that we do not have a global isomorphism invariant for these problems that may be subject to random matrix analysis. As a first step, we pose the following problem.

\begin{problem}
Devise average-case subexponential-time ($q^{o(n)}$) algorithms for  3-tensor isomorphism, cubic form equivalence, and alternating trilinear form equivalence working for $1/\poly(q)$ fraction of the first input tensor.
\end{problem}
For matrix code conjugacy, our average-case algorithm only works for $n$-dimensional matrix codes in $\M(n,q)$. It is interesting to extend to $m$-dimensional matrix codes.
\begin{problem}\label{prob:ntom}
    Devise an average-case polynomial-time algorithm for matrix code conjugacy in the setting of $\F_q^{n}\otimes \F_q^n\otimes \F_q^m$ that works for $1/\Theta(q)$ fraction of the first input tensor. 
\end{problem}

For $4$-tensor isomorphism, our average-case algorithm works for $1/\poly(q)$-fraction of the input. Is it possible to improve this to $1/\Theta(q)$ as in the case of algebra isomorphism? In addition, our algorithm only works for $n\times n\times n\times n$ tensors (with possible extensions to $n\times n\times m\times m$ tensors). Is it possible to devise algorithms that work for $\ell\times m\times n\times s$ tensors in general? We make these problems precise in the following.
\begin{problem}
    Devise an average-case polynomial-time algorithm for 4-tensor isomorphism in the setting of $\F_q^{n}\otimes \F_q^n\otimes \F_q^n\otimes \F_q^n$ that works for $1/\Theta(q)$ fraction of the first input tensor. 
\end{problem}

\begin{problem}
    Devise an average-case polynomial-time algorithm for 4-tensor isomorphism in the general setting of $\F_q^{\ell}\otimes \F_q^m\otimes \F_q^n\otimes \F_q^s$ that works for $1/\poly(q)$ fraction of the first input tensor. 
\end{problem}



More broadly, the average-case complexity landscape of TI-complete problems seems quite rich and it would be a nice project to explore it much further. 

\paragraph{Structure of the paper.} The rest of this paper is organised as follows: In Section~\ref{sec:prel}, we review the necessary notation and present some known results. In Section~\ref{sec:alg-iso}, we describe our average-case algorithm for algebra isomorphism and prove Theorem~\ref{thm:alg-iso}. In Section~\ref{sec:matrix-sp-conj}, we present our average-case algorithm for matrix code conjugacy and prove Theorem~\ref{thm:matrix-sp-conj}. In Section~\ref{sec:4-tensor-iso}, we elaborate on our average-case algorithm for $4$-tensor isomorphism and prove Theorem~\ref{thm:4-tensor-iso}. Section~\ref{sec:orth-tensor-iso} presents the algorithm for orthogonal tensor isomorphism and proves Theorem~\ref{thm:orth-tensor-iso}. All the required results on random matrix theory over finite fields are proved (in a self-contained way) in Section~\ref{sec:random matrix}.

\section{Preliminary}\label{sec:prel}

\subsection{Some notations}
Let $[n]:=\{1,\dots,n\}$.
Let $\F$ be a field. Let $\F_q$ be the finite field of order $q$. We reserve $q$ for field orders and $p$ for field characteristics. In asymptotic statements over finite fields, $q$ is fixed and $n\to\infty$. We use $\F^\times$ to denote the multiplicative group of non-zero elements in $\F$.

\subsection{Matrices over finite fields}\label{subsec:matrix-prelim}

Let $\M(n, \F)$ be the vector space of all $n \times n$ matrices over a field $\F$.  
We write $\M(n, q)$ for $\M(n, \F_q)$. Let $\GL(n, q)$ be the general linear group of degree $n$ over $\F_q$. For $A\in \M(n, \F)$, $\trp{A}$ denotes the transpose of $A$. Two matrices $A, B \in \M(n, \F)$ are said to be \emph{conjugate} over $\F$ if there exists an invertible matrix $T \in \GL(n, \F)$ such that
\[
B = T A T^{-1}.
\]
Conjugacy defines an equivalence relation on $\M(n, \F)$, and each equivalence class corresponds to a distinct \emph{type of matrix similarity}.  
Conjugacy preserves the characteristic polynomial, the minimal polynomial, and all eigenvalues (with multiplicities).

Every matrix $A \in \M(n, \F)$ defines a linear operator on $\F^n$, and hence an $\F[t]$-module structure where $t$ acts by $A$.  
Since $\F[t]$ is a principal ideal domain (PID), the structure theorem for finitely generated modules over a PID gives a canonical decomposition
\[
\F^n \cong \bigoplus_{\phi}\bigoplus_j \F[t]/(\phi^{\lambda_{\phi,j}}),
\]
where the direct sum runs over all monic irreducible polynomials $\phi \in \F[t]$.  
Each $\phi$ corresponds to a \emph{primary component} of $A$, and the sequence $\lambda_\phi = (\lambda_{\phi,1} \ge \lambda_{\phi,2} \ge \cdots)$ is a partition that describes the sizes of Jordan blocks associated with $\phi$.  
We refer to $\lambda_\phi$ as the \emph{Jordan partition} for $\phi$.

In terms of these invariants, the characteristic and minimal polynomials of $A$ are given by
\[
c_A(t) = \prod_{\phi} \phi(t)^{|\lambda_\phi|}, 
\qquad
m_A(t) = \prod_{\phi} \phi(t)^{\lambda_{\phi,1}},
\]
where $|\lambda_\phi| = \sum_j \lambda_{\phi,j}$ is the total size of all Jordan blocks corresponding to $\phi$.  
Thus, the similarity class of $A$ is completely determined by the collection of pairs $\{(\phi, \lambda_\phi)\}_\phi$.  For an irreducible factor \(\phi(t) \in \mathbb F[t]\), the \emph{algebraic multiplicity} of $\phi$ is $|\lambda_\phi|$, i.e., the total exponent of $\phi$ in $c_A(t)$;
the \emph{geometric multiplicity} of $\phi$ is the number of Jordan blocks associated with $\phi$, i.e., the length  of the partition.

Let $\M(\ell\times n, q)$ be the linear space of $\ell\times n$ matrices over $\F_q$. A random matrix in $\M(\ell\times n, q)$ is a matrix whose entries are sampled independently and uniformly at random from $\F_q$. We denote this by $A\rin \M(\ell\times n, q)$. Standard rank counts~\cite[Eq.~(1)]{Fulman2015} give, for $0\leq c\leq n$,
\[
    \Pr[\rk(A)=n-c\mid A\rin\M(n,q)]=\Theta(q^{-c^2}),
\]
with absolute implied constants.

The following estimates on the order of $\GL(n, q)$ and more will be useful. 
\begin{proposition}\label{prop:glnq and mnq}
For any $n\in\N$ and finite field $\F_q$, we have
\[
    |\GL(n,q)|=\prod_{i=0}^{n-1}(q^n-q^i).
\]
Moreover, the ratio
\begin{equation}\label{eq: ratio of GLnq and Mnq}
    c_n(q):=\frac{|\GL(n,q)|}{|\M(n,q)|}=\prod_{i=0}^{n-1}(1-q^{i-n})=\prod_{i=1}^{n}(1-q^{-i})
\end{equation}
converges to a constant $c(q):=\lim_{n\to\infty} c_n(q)\in(1/4, 1)$. Moreover, we have
\begin{equation}\label{eq:convergence bound of cnq}
    |c(q)-c_n(q)|=c_n(q)|1-\prod_{i=n+1}^{\infty}(1-q^{-i})|\leq\frac{q^{-(n+1)}}{(1-q^{-1})(1-q^{-(n+1)})}\leq 4q^{-(n+1)}.
\end{equation}
\end{proposition}  

Since we will use $c(q)$ to represent the (asymptotic) probabilities in our average-case analysis, we adapt the estimate in~\cite[Theorem 4.4]{NEUMANN_PRAEGER_1998} to obtain the following.
\begin{proposition}\label{prop: cqq}
Let $c(q)$ be defined as above. For any $q\geq 2$, we have: 
\[
c(q)^q\ge \exp\!\left(-1-\frac{3}{2q}-\frac{29}{9q^2}\right)\ge e^{-23/9},
\]
and
\[
\frac{q\cdot c(q)^q}{q-1}\ge e^{-23/9}.
\]
\end{proposition}

\begin{proof}
Recall that
\[
c(q)=\prod_{i=1}^\infty (1-q^{-i}).
\]
We take logarithms of both sides, and expand $\log(1-q^{-i})$ by the Taylor series to obtain 
\[
\log\bigl(c(q)^q\bigr)
= q\sum_{i=1}^\infty \log(1-q^{-i})
= -q\sum_{i=1}^\infty\sum_{j=1}^\infty \frac{q^{-ij}}{j}.
\]
Since all the terms in the double summation are nonnegative, Tonelli's theorem allows us to interchange the sums:
\[
\log\bigl(c(q)^q\bigr)
= -q\sum_{j=1}^\infty \frac{1}{j}\sum_{i=1}^\infty q^{-ij}
= -q\sum_{j=1}^\infty \frac{q^{-j}}{j(1-q^{-j})}.
\]
Therefore
\begin{equation}\label{eq:log-cq-exact}
\log\bigl(c(q)^q\bigr)
=
-\frac{q}{q-1}
-\frac{1}{2q(1-q^{-2})}
-q\sum_{j=3}^{\infty}\frac{q^{-j}}{j(1-q^{-j})}.
\end{equation}

We derive a uniform lower bound for Eq.~(\ref{eq:log-cq-exact}) for any $q\geq 2$. For the first term, note that
\[
\frac{q}{q-1}
=1+\frac{1}{q}+\frac{1}{q(q-1)}
\le 1+\frac{1}{q}+\frac{2}{q^2}.
\]
For the second term, we have
\[
\frac{1}{2q(1-q^{-2})}
=
\frac{1}{2q}+\frac{1}{2q(q^2-1)}
\le \frac{1}{2q}+\frac{1}{3q^2}.
\]
For the summation, note that $q^{-j}\le q^{-2}\le 1/4$ for $j\ge 3$ and $q\geq 2$, then we have
\[
\frac{1}{1-q^{-j}}\le \frac{4}{3}.
\]
Thus, 
\[
q\sum_{j=3}^{\infty}\frac{q^{-j}}{j(1-q^{-j})}
\le
\frac{4q}{3}\sum_{j=3}^{\infty}\frac{q^{-j}}{j}
\le
\frac{4q}{9}\sum_{j=3}^{\infty}q^{-j}
=
\frac{4}{9q(q-1)}
\le
\frac{8}{9q^2},
\]
where the second inequality uses the fact that $1/j\leq 1/3$ for any $j\geq 3$ and the last inequality uses the fact that $1/(q(q-1))\leq 2/q^2$ for any $q\geq 2$.
Substituting these bounds into \eqref{eq:log-cq-exact}, we obtain
\[
\log\bigl(c(q)^q\bigr)
\ge
-1-\frac{3}{2q}-\frac{29}{9q^2}.
\]
Then we have
\[
c(q)^q
\ge
\exp\!\left(-1-\frac{3}{2q}-\frac{29}{9q^2}\right).
\]
Since $q\ge 2$,
\[
1+\frac{3}{2q}+\frac{29}{9q^2}
\le
1+\frac{3}{4}+\frac{29}{36}
=
\frac{23}{9},
\]
so we have
\[
c(q)^q\ge e^{-23/9},
\]
and
\[
\frac{q\cdot c(q)^q}{q-1}\ge e^{-23/9}
\]
follows since $q/(q-1)\geq 1$.
\end{proof}

\subsection{3-way arrays, matrix tuples, and matrix codes}

\paragraph{3-way arrays.} We use $\T(\ell\times m\times n, \F)$ to denote the linear space of $\ell\times m\times n$ 3-way arrays over $\F$. For $\tA=(a_{i,j,k})\in \T(\ell\times m\times n, \F)$, its horizontal matrix tuple is $(A_1, \dots, A_\ell)\in \M(m\times n, \F)^\ell$, where $A_i(j,k)=a_{i,j,k}$, its vertical matrix tuple is $(A_1', \dots, A_m')\in \M(\ell\times n, \F)^m$ where $A_j'(i,k)=a_{i,j,k}$, and its frontal matrix tuple is $(A_1'', \dots, A_n'')\in\M(\ell\times m, \F)^n$ where $A_k''(i,j)=a_{i,j,k}$.

A natural action of $\GL(\ell, \F)\times\GL(m, \F)\times\GL(n, \F)$ on $\T(\ell\times m\times n, \F)$ is as follows: $(L, R, T)\in\GL(\ell, \F)\times\GL(m,\F)\times\GL(n,\F)$, where $L=(\ell_{i,i'})$, $R=(r_{j,j'})$ and $T=(t_{k,k'})$, sends $\tA=(a_{i,j,k})\in\T(\ell\times m\times n, \F)$ to $\tB=(b_{i,j,k})\in\T(\ell\times m\times n, \F)$ where $b_{i,j,k}=\sum_{i',j',k'}\ell_{i,i'}r_{j,j'}t_{k,k'}a_{i',j',k'}$. We denote $\tB$ as $L\tA^T R^t$, to indicate that $L$ acts on the left, $R$ acts on the right, and $T$ acts on the third direction of $\tA$.


\paragraph{Matrix tuples.} Let $\vA=(A_1, \dots, A_m)\in\M(n, \F)^m$ be an $m$-tuple of $n\times n$ matrices over $\F$. 
The \emph{conjugacy group} of $\vA$ is $\conj(\vA)=\{T\in\GL(n, \F)\mid \forall i\in[m], TA_i=A_iT\}\leq\GL(n, \F)$. 

Two matrix tuples $\vA=(A_1, \dots, A_m)$ and $\vB=(B_1, \dots, B_m)\in\M(n, q)^m$ are conjugate, if there exists $T\in\GL(n, q)$ such that for any $i\in[m]$, $TA_iT^{-1}=B_i$.
The \emph{conjugacy coset} from $\vA$ to $\vB$ is $\conj_\to(\vA, \vB)=\{T\in\GL(n, q)\mid \forall i\in[m], TA_i=B_iT\}$. It is clear that $\conj_\to(\vA, \vB)$ is empty or a coset of $\conj(\vA)$; in the latter case, $|\conj_\to(\vA, \vB)|=|\conj(\vA)|$.

The following result of Neumann and Praeger is classical. 
\begin{theorem}[{\cite[Theorem 6.1]{NP95}}]\label{thm:gen-matrix-alg}
Let $A_1, A_2\in\M(n, q)$ be two independent uniformly random matrices. Then $A_1$ and $A_2$ generate the full matrix algebra $\M(n, q)$ with probability $1-1/q^{\Omega(n)}$.
\end{theorem}

The corollary of Theorem~\ref{thm:gen-matrix-alg} is as follows. 
\begin{corollary}\label{cor:conj-bound}
Let $\vA=(A_1, A_2)\in\M(n, q)^2$ consist of two independent uniformly random matrices and $\vB=(B_1, B_2)\in\M(n, q)^2$ be two arbitrary matrices. Then $|\conj_\to(\vA, \vB)|\leq q-1$ with probability $1-1/q^{\Omega(n)}$.
\end{corollary}
\begin{proof}
When $A_1$ and $A_2$ generate the full matrix algebra, $\conj(\vA)=\{\lambda I_n\mid \lambda\in\F_q^\times\}$. The result then follows immediately from Theorem~\ref{thm:gen-matrix-alg}.
\end{proof}

We also need the following results from computational algebra. 
\begin{theorem}[\cite{BL08,BO08,IKS10}]\label{thm:module-iso}
Given $\vA, \vB\in \M(n, q)^m$, there exist polynomial-time algorithms that decide if $\vA$ and $\vB$ are conjugate, and if so, compute $\conj_\to(\vA, \vB)$ represented by a coset representative and a generating set of $\conj(\vA)$.
\end{theorem}

\paragraph{Matrix codes.} Following the practice in coding theory, we call a subspace of $\M(\ell\times n, \F)$ a matrix code. Two matrix codes $\cA, \cB\leq \M(n, q)$ are conjugate, if there exists $T\in\GL(n, q)$ such that $\cA=T\cB T^{-1}=\{TBT^{-1}\mid B\in\cB\}$.

\section{Average-case algorithm for algebra isomorphism and matrix code conjugacy}

\subsection{Average-case algorithm for algebra isomorphism}\label{sec:alg-iso}

In this section, we present the average-case algorithm for algebra isomorphism to prove Theorem~\ref{thm:alg-iso}.

\paragraph{Problem set up.} Let $U\cong\F_q^n$, and let $\tA, \tB\in U\otimes U\otimes U^*$ be two algebras. In algorithms, after fixing a basis of $U$, $\tA$ is represented by its structure constants, that is, (by some abuse of notation) a $3$-way array $\tA=(a_{i,j,k})\in \T(n\times n\times n, q)$, so the bilinear map $g_\tA:\F_q^n\times\F_q^n\to\F_q^n$ represented by $\tA$ is defined by $g_\tA(\stdb_i,\stdb_j)=\sum_{k\in[n]}a_{i,j,k}\stdb_k$. We see that $\tA$ and $\tB$ are isomorphic as algebras if and only if there exists $T\in\GL(n, q)$ such that $\tB=(T\otimes T\otimes T^{-\mathrm t})\cdot\tA$.

We now slice $\tA$ along the first index to get the horizontal slices of $\tA$ as $\vA=(A_1, \dots, A_n)\in\M(n, q)^n$ where $A_i(j, k)=a_{i,j,k}$. Similarly, do this for $\tB$ and get $\vB=(B_1, \dots, B_n)\in\M(n, q)^n$. Then the isomorphism of $\tA$ and $\tB$ as algebras translates to asking if there exists $T=(t_{i,i'})\in\GL(n, q)$, such that  
\begin{equation}\label{eq:algebra-iso-to-conjugacy}
    \sum_{i'\in[n]} t_{i,i'}T A_{i'} T^{-1}=B_i~\forall\ i\in[n].
\end{equation}

Let $\cA=\linspan\{A_1, \dots, A_n\}\leq \M(n, q)$ and $\cB=\linspan\{B_1, \dots, B_n\}\leq \M(n, q)$. As $A_i$'s are random matrices in $\M(n, q)$, $\cA$ is of dimension $n$ with probability $1-1/q^{\Omega(n^2)}$, which we will assume in the following. By Eq.~(\ref{eq:algebra-iso-to-conjugacy}), a necessary condition for $\tA$ and $\tB$ to be isomorphic as algebras is that $\cA$ and $\cB$ are conjugate as matrix codes. This leads to the first step of the algorithm. 

\paragraph{Step 1. Computing the hull.} Following \cite{CL25}, we make the following definition. 
\begin{definition}\label{def:hull}
Define a bilinear form $f$ on $\M(n, q)$ as $f(A, B)=\Tr(AB)$ for $A, B\in\M(n, q)$. The \emph{orthogonal space} of $\cA\leq\M(n, q)$ with respect to $f$ is $\cA^\perp:=\{A'\in \M(n, q)\mid f(A, A')=0~\forall A\in\cA\}$. A subspace $\cA\leq\M(n, q)$ is \emph{self-orthogonal} if $\cA\subseteq \cA^\perp$. That is, for any $A, A'\in \cA$, $\Tr(AA')=0$. For $A\in\M(n, q)$, we say that $A$ is \emph{self-dual} if $\Tr(A^2)=0$.  A self-dual line will mean a one-dimensional self-orthogonal subspace. The \emph{hull} of $\cA$ is $\hull(\cA):=\cA\cap\cA^\perp$. 
\end{definition}
It is clear that the hull is self-orthogonal. 
Note that the bilinear form $f$ is a non-degenerate bilinear form on $\M(n, q)$. It is different from the more familiar bilinear form $(A, B)\mapsto \Tr(\trp{A}B)$.
The reason for adopting this non-standard bilinear form is the following observation.
\begin{lemma}[{\cite[Lemma 3]{CL25}}]\label{obs:hull}
Let $\cA$ and $\cB\leq\M(n, q)$. If $\cA$ and $\cB$ are conjugate, then $\hull(\cA)$ and $\hull(\cB)$ are conjugate. 
\end{lemma}

The hull dimension of a random code has been studied by Sendrier~\cite{Sendrierhull} (see also~\cite{CL25}). For the trace form used here, we need the following estimate for codes of dimension $n$.
\begin{proposition}\label{prop:hull-dim}
Let $\cA$ be a uniformly random $n$-dimensional subspace of $\M(n,q)$. For fixed $q$, as $n\to\infty$,
\[
    \Pr[\dim\hull(\cA)=1]
    =\pi_1(q)+o(1),\qquad
    \pi_1(q):=\frac{1}{q-1}\prod_{j=1}^{\infty}(1-q^{-(2j-1)}).
\]
In particular, $\pi_1(q)=\Theta(1/q)$, and the probability is $\Theta(1/q)$ for all sufficiently large $n$.
\end{proposition}
\begin{proof}
Choose independent uniform matrices $X_1,\ldots,X_n\in\M(n,q)$, put $N=n^2$, and let $G=(\Tr(X_iX_j))_{i,j\in[n]}$. The probability that the $X_i$ are linearly dependent is at most $\sum_{j=0}^{n-1}q^{j-N}=o(1)$. Conditioned on independence, their span is a uniform $n$-dimensional code and $\dim\hull(\cA)=\operatorname{corank}G$.

We show that the law $\nu$ of $G$ approaches the uniform law $U$ on symmetric $n\times n$ matrices in total variation. We use the total variation distance and $\chi^2$-divergence defined in Section~\ref{subsec:orth-average}, and fix a nontrivial additive character $\psi:\F_q\to\C^\times$. For odd $q$, diagonalize the nondegenerate trace form on $\M(n,q)$. The Fourier argument of Proposition~\ref{prop:orth-gram-uniform} applies to this diagonal form with $N$ nonzero coefficients: the squared modulus of the Fourier coefficient indexed by a symmetric matrix of rank $r$ is $q^{-Nr}$. The rank count used in its proof therefore gives
\[
    \chi^2(\nu\Vert U)
    \leq\frac{1}{c(q)}\sum_{r=1}^n
        q^{-(N-n)r-r(r-1)/2}=o(1).
\]
For even $q$, the trace form is the orthogonal sum of $n$ one-dimensional forms and $n(n-1)/2$ hyperbolic planes, by grouping the diagonal entries and the pairs of off-diagonal entries of a matrix. Consequently,
\[
    G=\sum_{a=1}^n x_a\trp{x_a}
       +\sum_{1\leq a<b\leq n}
          (y_{ab}\trp{z}_{ab}+z_{ab}\trp{y}_{ab}),
\]
where all displayed vectors are independent and uniform in $\F_q^n$. Write a character of the symmetric-matrix space as $\chi_D(G)=\psi(\sum_{i\leq j}d_{ij}G_{ij})$ and put $Q_D(x)=\sum_{i\leq j}d_{ij}x_ix_j$. Let $B_D$ be the alternating polar form of $Q_D$ and put $r=\rk B_D$. If $r=0$ and $\chi_D$ is nontrivial, one of the $x_a$ gives a vanishing character average. If $r>0$, then $r\geq2$. Squaring the quadratic character sum and applying character orthogonality gives
\[
    |\mathbb E_x\psi(Q_D(x))|^2
    =q^{-n}\left|\sum_{h\in\operatorname{rad}B_D}\psi(Q_D(h))\right|
    \leq q^{-r},
\]
while each hyperbolic pair gives the average $q^{-r}$. Thus $|\mathbb E\chi_D(G)|^2\leq q^{-Nr}$, and Plancherel gives
\[
    \chi^2(\nu\Vert U)\leq q^{n(n+1)/2-2N}=o(1).
\]
In both cases, Eq.~(\ref{eq:orth-tv-chi-square}) implies $d_{\mathrm{TV}}(\nu,U)=o(1)$. Conditioning on the $X_i$ being independent changes any event probability by $o(1)$. Finally, the symmetric-matrix rank distribution~\cite[Section~4, Eq.~(20) and Theorem~4.1]{Fulman2015}, valid in both characteristics, gives limiting corank-one probability $\pi_1(q)$. Since $c(q)>1/4$, we have $1/(4q)<\pi_1(q)\leq2/q$, proving the last assertion.
\end{proof}

We also need the following proposition.
\begin{proposition}\label{prop:random-self-dual}
Let $n\geq2$, and let $\cA$ be a uniformly random $n$-dimensional subspace of $\M(n,q)$. Conditioned on $\dim\hull(\cA)=1$, its hull is uniformly distributed among the self-dual lines, except that $\linspan\{\idmat_n\}$ is excluded when both $q$ and $n$ are even. Consequently, a uniformly chosen nonzero generator of the hull is uniform among the nonzero self-dual matrices, with all scalar matrices additionally excluded when both $q$ and $n$ are even.
\end{proposition}
\begin{proof}
Put $V=\M(n,q)$, with the nondegenerate symmetric bilinear form $f(X,Y)=\Tr(XY)$, and let $G$ be its group of linear isometries. For $g\in G$, the map $\cA\mapsto g\cA$ is a bijection on $n$-dimensional codes and satisfies $\hull(g\cA)=g\hull(\cA)$. Thus it suffices to show that the stated hull lines form a single $G$-orbit. For odd $q$, this follows from Witt's theorem~\cite[Theorem~7.4]{Tay92}.

Suppose that $q$ is even, and put $w=\idmat_n$. For every $X\in V$,
\[
    f(X,X)=\Tr(X)^2=f(X,w)^2.
\]
Nondegeneracy and the injectivity of squaring imply that every element of $G$ fixes $w$. If $n$ is odd, then $f(w,w)=1$ and $V=\linspan\{w\}\perp w^\perp$, where $f|_{w^\perp}$ is nondegenerate and alternating. Its symplectic group, extended by the identity on $w$, acts transitively on the nonzero vectors of $w^\perp$, which are exactly the nonzero self-dual matrices.

If $n$ is even, choose $z=E_{11}$, the matrix with its only nonzero entry equal to $1$ in position $(1,1)$. The restriction of $f$ to $\linspan\{w,z\}$ has Gram matrix $\left[\begin{smallmatrix}0&1\\1&1\end{smallmatrix}\right]$. Hence
\[
    V=\linspan\{w,z\}\perp U,
    \qquad U:=\linspan\{w,z\}^\perp,
\]
where $f|_U$ is nondegenerate and alternating. Every self-dual matrix has the form $aw+u$, with $a\in\F_q$ and $u\in U$. For $t\in U$, define the linear map
\[
    g_t(w)=w,\qquad g_t(z)=z+t,\qquad
    g_t(u)=u+f(u,t)w\quad(u\in U).
\]
Checking the pairings on $w,z,U$ shows that $g_t$ preserves $f$. Together with the symplectic group of $U$, extended by the identity on $\linspan\{w,z\}$, these maps act transitively on the self-dual vectors outside $\linspan\{w\}$: first send a nonzero $u$ to any prescribed nonzero $u'$, and then choose $t$ with $f(u',t)$ equal to the required change in the coefficient of $w$.

Finally, $\hull(\cA)=\linspan\{w\}$ would imply $\cA\subseteq w^\perp$, so $f|_{\cA}$ would be alternating. Its rank is even, and therefore its radical has even dimension because $\dim\cA=n$ is even. This rules out a one-dimensional hull. The remaining hull lines form the single orbit just described, proving uniformity. Choosing a uniform nonzero generator of a uniform line gives the stated matrix distributions.
\end{proof}

After this step, we assume that $\cA$, as a matrix code spanned by $n$ random $n\times n$ matrices, has its hull as $\hull(\cA)=\linspan\{A\}$. By Lemma~\ref{obs:hull}, if $\hull(\cB)$ is not of dimension $1$, then $\cA$ and $\cB$ are not conjugate. Therefore, we can assume that $\hull(\cB)=\linspan\{B\}$. 

\paragraph{Step 2. Extracting an eigenline from the hull.} We seek a nonzero hull generator having exactly one eigenvalue in $\F_q$, which is nonzero and has algebraic multiplicity $1$. The nonzero condition lets us normalize the generator without enumerating field elements. We use the following consequence of our random-matrix estimates.
\begin{theorem}\label{thm:main-technical-informal}
For each fixed prime power $q$, there exists $N_0(q)\geq3$ such that, for every $n\geq N_0(q)$, a uniform matrix $H\in\M(n,q)$ conditioned on $\Tr(H^2)=0$ has exactly one eigenvalue in $\F_q$, which is nonzero and has algebraic multiplicity $1$, with probability at least $\frac12e^{-23/9}$. Conditional on this event, its one-dimensional left eigenspace is uniformly distributed among the one-dimensional subspaces of $\F_q^{1\times n}$.
\end{theorem}
Theorem~\ref{thm:main-technical} gives the limiting probability $\gamma^*(\infty,q)=c(q)^q\geq e^{-23/9}$, proving the probability assertion for fixed $q$ and sufficiently large $n$. The conditional eigenspace is uniform because conjugation preserves the distribution and the spectral condition, and $\GL(n,q)$ acts transitively on the one-dimensional left subspaces. The restriction to sufficiently large $n$ is necessary: for $n=2$, a characteristic polynomial with a simple root in $\F_q$ has a second root in $\F_q$.

Proposition~\ref{prop:random-self-dual} transfers this estimate to a uniform nonzero hull generator. All matrices excluded there are scalar and fail the spectral condition for $n\geq2$, so excluding them can only increase its probability. The event and its eigenline are unchanged by nonzero rescaling; thus the algorithm may start with any nonzero hull generator.

Suppose that a generator of $\hull(\cA)$ has the required unique eigenvalue $\alpha\ne0$. Divide the generator by $\alpha$, and denote the resulting matrix by $A$. This is the unique matrix on the hull line whose only $\F_q$-eigenvalue is $1$. Define $B$ from $\hull(\cB)$ in the same way. If $\cB$ has no such generator, the algebras cannot be isomorphic. Otherwise, every algebra isomorphism $T$ satisfies
\[
    TAT^{-1}=B.
\]
Compute nonzero left eigenvectors $vA=v$ and $uB=u$ by polynomial factorization and Gaussian elimination~\cite{CantorZassenhaus81,vzGG13}. For every potential isomorphism $T$, there is a scalar $c\in\F_q^\times$ such that $v=cuT$.

\paragraph{Step 3. Obtaining a second matrix by contraction.} Put
\[
    X=\sum_{i=1}^n v_iA_i,\qquad
    Y=\sum_{i=1}^n u_iB_i.
\]
Using Eq.~(\ref{eq:algebra-iso-to-conjugacy}) and $v=cuT$, we obtain
\[
    TXT^{-1}
    =\sum_jv_jTA_jT^{-1}
    =c\sum_i u_iB_i=cY.
\]
If $\Tr(X)=0$, return ``Failure''. If $\Tr(X)\ne0$ but $\Tr(Y)=0$, report ``Not isomorphic''. Otherwise, put
\begin{equation}\label{eq:algebra-trace-normalization}
    A'=\frac{X}{\Tr(X)},\qquad B'=\frac{Y}{\Tr(Y)}.
\end{equation}
Since $\Tr(X)=c\Tr(Y)$, every algebra isomorphism conjugates $A'$ to $B'$. These matrices are independent of the choices of nonzero eigenvectors $v,u$. We therefore have two pairs that must be matched:
\[
    TAT^{-1}=B,\qquad TA'T^{-1}=B'.
\]
No spectral condition on $X$ is needed. Its conditional distribution and the centralizer of $(A,A')$ are analysed below.

\paragraph{Step 4. Completing the isomorphism test.} Compute the simultaneous commuting algebra
\[
    \mathcal Z(A,A'):=\{Z\in\M(n,q):ZA=AZ,\ ZA'=A'Z\}
\]
by linear algebra. Return ``Failure'' if its dimension exceeds $1$. Otherwise, this algebra consists of scalar matrices, and the conjugacy group of $(A,A')$ consists of the nonzero scalars. Use Theorem~\ref{thm:module-iso} to compute the conjugacy coset from $(A,A')$ to $(B,B')$. If it is empty, report ``Not isomorphic''. Otherwise, let $T$ be a representative. Every algebra isomorphism must then be $\lambda T$ for some $\lambda\in\F_q^\times$.

For each $i\in[n]$, compute
\[
    D_i:=\sum_{j=1}^n t_{ij}TA_jT^{-1}.
\]
Substituting $\lambda T$ into Eq.~(\ref{eq:algebra-iso-to-conjugacy}) gives the scalar equations
\begin{equation}\label{eq:algebra-final-scalar}
    \lambda D_i=B_i,\qquad i\in[n].
\end{equation}
The tuple $(D_1,\ldots,D_n)$ is nonzero because $T$ is invertible and the $A_i$ are linearly independent. One nonzero entry therefore determines the only possible $\lambda$. Check that it is nonzero and satisfies all the equations. If so, output $\lambda T$; otherwise, report ``Not isomorphic''.

\paragraph{Summary of the algorithm steps.} Given the structure tensors $\tA,\tB$, the algorithm proceeds as follows. A failure of a randomized subroutine is reported as ``Failure''.
\begin{enumerate}
    \item Form their horizontal slices. If the slices of $\tA$ are linearly dependent, return ``Failure''; if only those of $\tB$ are dependent, report ``Not isomorphic''.
    \item Compute the two hulls. If the first does not have dimension $1$, return ``Failure''; if only the second does not have dimension $1$, report ``Not isomorphic''.
    \item Test the hull generators for exactly one $\F_q$-eigenvalue, nonzero and of algebraic multiplicity $1$. If the first fails, return ``Failure''; if only the second fails, report ``Not isomorphic''. Normalize both generators to have eigenvalue $1$, obtaining $A,B$, and compute left eigenvectors $v,u$.
    \item Form $X,Y$ by contraction. If $\Tr(X)=0$, return ``Failure''; if only $\Tr(Y)=0$, report ``Not isomorphic''. Form $A',B'$ as in Eq.~(\ref{eq:algebra-trace-normalization}).
    \item Compute $\mathcal Z(A,A')$. If its dimension exceeds $1$, return ``Failure''. Otherwise, compute the conjugacy coset from $(A,A')$ to $(B,B')$ and perform the scalar check in Eq.~(\ref{eq:algebra-final-scalar}). Output the resulting isomorphism, or report ``Not isomorphic'' if the coset or scalar check has no solution.
\end{enumerate}

\paragraph{Correctness of the algorithm.} The hull and its normalized generator are preserved under algebra isomorphism by Lemma~\ref{obs:hull}. The calculation in Step~3 shows that the trace-normalized contraction is also preserved. Thus every actual algebra isomorphism lies in the conjugacy coset computed in Step~4. When the commuting algebra has dimension $1$, that coset consists of scalar multiples of its representative. Equation~(\ref{eq:algebra-final-scalar}) tests all these possibilities without enumeration. Every returned isomorphism is verified, and every ``Not isomorphic'' conclusion excludes a necessary condition. Consequently, every output other than ``Failure'' is correct.

\paragraph{Running time analysis.} All linear systems, hull computations, contractions, and scalar checks have size polynomial in $n$. Polynomial factorization and matrix-tuple conjugacy take randomized $\poly(n,\log q)$ time by the algorithms already cited and Theorem~\ref{thm:module-iso}. Testing the dimension of the commuting algebra uses linear equations and avoids enumerating its group of units. The algorithm therefore runs in randomized $\poly(n,\log q)$ time.

\paragraph{Average-case analysis.} We analyse the contraction by first conditioning on the matrix code and then comparing probability distributions.
\begin{proposition}\label{prop:hull-contraction}
Fix $q$, and sample $A_1,\ldots,A_n$ independently and uniformly from $\M(n,q)$. Condition on their linear independence, $\dim\hull(\cA)=1$, and the hull generator having exactly one $\F_q$-eigenvalue, which is nonzero and of algebraic multiplicity $1$. Let $A$ be its normalization to eigenvalue $1$, choose a nonzero left eigenvector $v=v(\cA)$ by a fixed rule depending only on $\cA$, and put $X=\sum_i v_iA_i$. Let $\varepsilon_{\mathrm{gen}}(n,q)=q^{-\Omega(n)}$ be the failure probability in Theorem~\ref{thm:gen-matrix-alg}. For all sufficiently large $n$ depending on $q$,
\begin{align}
    \Pr[A,X\text{ do not generate }\M(n,q)]
    &\leq\frac{q-1}{q^n-1}+Kq^2\varepsilon_{\mathrm{gen}}(n,q)
      =q^{-\Omega(n)},\label{eq:hull-contraction-generation}\\
    \Pr[\Tr (X)\neq0]
    &\geq(1-q^{-1})(1-16q^{1-n}),\label{eq:hull-contraction-trace}
\end{align}
where $K$ is an absolute constant. Consequently, $A,X$ generate $\M(n,q)$ and $\Tr (X)\neq0$ jointly with probability at least $1/3$.
\end{proposition}
\begin{proof}
Put $V=\M(n,q)$ and $N=n^2$. Conditional on $\cA$, the ordered slices form a uniform ordered basis. Since $v$ depends only on $\cA$, $X$ is uniform on $\cA\setminus\{0\}$.

Let $\mathcal L$ be the self-dual lines allowed in Proposition~\ref{prop:random-self-dual}. Consider the incidence experiment that chooses $L\in\mathcal L$ uniformly and then chooses an $n$-dimensional code $C$ uniformly subject to $L\subseteq C\subseteq L^\perp$. Put $D=\{\hull(C)=L\}$. The isometry transitivity proved in Proposition~\ref{prop:random-self-dual} shows that $p_n:=\Pr[D\mid L]$ is independent of $L$. Conditional on $D$, this experiment gives a uniform hull-one code and its hull. Writing $h_n$ for the hull-one probability of an unconditioned uniform $n$-dimensional code, double counting gives
\[
    p_n=h_n\frac{\genfrac{[}{]}{0pt}{}{N}{n}_q}
        {|\mathcal L|\genfrac{[}{]}{0pt}{}{N-2}{n-1}_q},
    \qquad
    \frac{\genfrac{[}{]}{0pt}{}{N}{n}_q}
         {\genfrac{[}{]}{0pt}{}{N-2}{n-1}_q}
    =\frac{(q^N-1)(q^{N-1}-1)}{(q^n-1)(q^{N-n}-1)}.
\]
Here $\genfrac{[}{]}{0pt}{}{r}{k}_q$ counts the $k$-dimensional subspaces of $\F_q^r$. Theorem~\ref{thm:random matrix being self-dual} gives $|\mathcal L|\sim q^{N-1}/(q-1)$; the exceptional scalar line in characteristic two changes this count by only one. Proposition~\ref{prop:hull-dim} therefore yields
\[
    p_n\longrightarrow(q-1)\pi_1(q)
      =\prod_{j\geq1}(1-q^{-(2j-1)})>\tfrac14.
\]
Thus $p_n\geq1/8$ eventually.

For the analysis, choose $H$ uniformly among the nonzero generators of $L$, and condition on its required nonzero spectral property. This gives the uniform distribution on the self-dual matrices with that property: every excluded matrix in Proposition~\ref{prop:random-self-dual} is scalar and fails it. Equation~(\ref{eq:selfdual-nonzero-lower}) bounds the ambient probability of this event below by $a/(2q)$, where $a=e^{-23/9}$. Hence the density of $H$ relative to a uniform ambient matrix is at most $2q/a$.

For every code $C$ in the incidence experiment, a uniform $X\in C\setminus\{0\}$ lies in $L$ with probability $(q-1)/(q^n-1)$. Conditional on $X\notin L$ and before conditioning on $D$, it is uniform on $L^\perp\setminus L$, since each vector outside $L$ lies in the same number of such codes. This distribution has density at most $2q$ relative to a uniform ambient matrix. The probability of $X\notin L$ is the same for every $C$, so conditioning on it leaves $\Pr[D\mid L]=p_n$. Conditioning further on $D$ increases densities by at most $1/p_n\leq8$. Thus the joint law of $(H,X)$ off $L$ has density at most $Kq^2$ relative to two independent uniform matrices. Theorem~\ref{thm:gen-matrix-alg} gives Eq.~(\ref{eq:hull-contraction-generation}). Replacing $H$ by its normalized generator $A$ does not change the generated algebra.

Since $H$ is nonscalar, the trace functional is nonzero on $L^\perp$. If $\Tr (H)\neq0$, no code containing $L$ is contained in $\idmat_n^\perp$. Otherwise, before conditioning on $D$,
\[
    \Pr[C\subseteq\idmat_n^\perp\mid L]
    =\frac{\genfrac{[}{]}{0pt}{}{N-3}{n-1}_q}
           {\genfrac{[}{]}{0pt}{}{N-2}{n-1}_q}
    =\frac{q^{N-n-1}-1}{q^{N-2}-1}\leq2q^{1-n}.
\]
After conditioning on $D$, this probability is at most $16q^{1-n}$. On every remaining code, a uniform nonzero $X$ has nonzero trace with probability $(q-1)q^{n-1}/(q^n-1)\geq1-q^{-1}$, proving Eq.~(\ref{eq:hull-contraction-trace}). Subtracting the generation-failure bound from this trace-success bound proves the joint lower bound $1/3$ for sufficiently large $n$.
\end{proof}

The horizontal slices are linearly independent with probability $1-q^{-\Omega(n^2)}$. Proposition~\ref{prop:hull-dim} gives hull dimension one with probability $\Theta(1/q)$. Conditional on this event, Proposition~\ref{prop:random-self-dual} and the nonzero spectral estimate in Theorem~\ref{thm:main-technical} give a constant probability of the required hull spectrum. Proposition~\ref{prop:hull-contraction} then gives a constant probability that trace normalization succeeds and the resulting pair generates $\M(n,q)$. Its centralizer consists of scalars, as in the proof of Corollary~\ref{cor:conj-bound}. The correctness argument therefore applies for every second tensor $\tB$. Bounded repetition reduces randomized polynomial-factorization failures to any fixed constant with polynomial overhead. Consequently, the algorithm succeeds for a $\Theta(1/q)$ fraction of first tensors, in time $\poly(n,\log q)$.

\subsection{Average-case algorithm for matrix code conjugacy}\label{sec:matrix-sp-conj}

In this section, we prove Theorem~\ref{thm:matrix-sp-conj}. We use the normalized hull generator from Section~\ref{sec:alg-iso} and construct a second matrix by the trace pairing. Both matrices depend only on the code.

\paragraph{Problem set up.} Let $\cA,\cB\leq\M(n,q)$ be matrix codes of dimension $n$, with bases $(A_1,\ldots,A_n)$ and $(B_1,\ldots,B_n)$. They are conjugate if there exist $S\in\GL(n,q)$ and $T=(t_{ij})\in\GL(n,q)$ such that
\begin{equation}\label{eq:matrix-conjugacy}
    B_i=\sum_{j=1}^n t_{ij}SA_jS^{-1},\qquad i\in[n].
\end{equation}
Equivalently, $S\cA S^{-1}=\cB$. Once $S$ is known, $T$ can be recovered by expressing each $B_i$ in the conjugated basis. If the codes are given by $n$ possibly dependent slices, first perform the rank checks from Section~\ref{sec:alg-iso}: return ``Failure'' if the first span has dimension less than $n$, and report ``Not conjugate'' if only the second does.

\paragraph{Step 1. Normalizing the hull generators.} Compute the hulls using Definition~\ref{def:hull}. As in Steps~1 and~2 of Section~\ref{sec:alg-iso}, require the first hull to be one-dimensional and its nonzero generators to have exactly one $\F_q$-eigenvalue, which is nonzero and has algebraic multiplicity $1$. If either condition fails, return ``Failure''. If the first code satisfies these conditions but the second does not, report ``Not conjugate''.

Normalize the generators by their unique eigenvalues, obtaining $A\in\hull(\cA)$ and $B\in\hull(\cB)$ whose only $\F_q$-eigenvalue is $1$. Every code conjugacy $S$ then satisfies $SAS^{-1}=B$.

\paragraph{Step 2. Constructing a trace-dual matrix.} Write
\[
    c_A(t)=(t-1)g_A(t),\qquad g_A(1)\ne0,
\]
and compute
\begin{equation}\label{eq:code-spectral-projector}
    P_A:=\frac{g_A(A)}{g_A(1)}.
\end{equation}
This is the rank-one spectral projector onto the $1$-eigenspace of $A$: it is the identity on that eigenspace and zero on the complementary primary components. In particular,
\[
    P_A^2=P_A,\qquad AP_A=P_A,\qquad
    \Tr(AP_A)=\Tr(P_A^2)=1.
\]
Define the following complement of the hull:
\begin{equation}\label{eq:code-canonical-complement}
    D_A:=\cA\cap P_A^\perp
       =\{Y\in\cA:\Tr(P_AY)=0\}.
\end{equation}
Since $\Tr(AP_A)=1$, we have $\cA=\linspan\{A\}\oplus D_A$ and $\dim D_A=n-1$. The trace form restricted to $D_A$ is nondegenerate: a vector orthogonal to $D_A$ within $D_A$ is orthogonal to all of $\cA$, hence lies in its hull and must be zero.

There is therefore a unique matrix $Z_A\in D_A$ satisfying
\begin{equation}\label{eq:code-trace-dual}
    \Tr(Z_AY)=\Tr(Y),\qquad\forall Y\in D_A.
\end{equation}
To compute it, choose a basis $(D_1,\ldots,D_{n-1})$ of $D_A$, form the nonsingular Gram matrix $G=(\Tr(D_iD_j))_{i,j}$, and solve
\[
    Gz=(\Tr(D_1),\ldots,\Tr(D_{n-1}))^{\mathrm t}.
\]
Then $Z_A=\sum_i z_iD_i$. Uniqueness shows that the result does not depend on the chosen basis. Construct $P_B,D_B,Z_B$ in the same way.

Every code conjugacy sends $P_A$ to $P_B$ by Eq.~(\ref{eq:code-spectral-projector}), and consequently sends $D_A$ to $D_B$. Conjugation preserves both trace and the trace pairing, so Eq.~(\ref{eq:code-trace-dual}) gives
\begin{equation}\label{eq:code-pair-conjugacy}
    SAS^{-1}=B,\qquad SZ_AS^{-1}=Z_B.
\end{equation}

\paragraph{Step 3. Completing the conjugacy test.} Compute the simultaneous commuting algebra
\[
    \mathcal Z(A,Z_A)=\{W\in\M(n,q):WA=AW,\ WZ_A=Z_AW\}
\]
by linear equations. If its dimension exceeds $1$, return ``Failure''. Otherwise, its invertible elements are exactly the nonzero scalar matrices. Use Theorem~\ref{thm:module-iso} to compute the conjugacy coset from $(A,Z_A)$ to $(B,Z_B)$. If it is empty, report ``Not conjugate''. If it is nonempty, let $S$ be a representative and test $S\cA S^{-1}=\cB$ by linear algebra. If the test succeeds, output $S$ and, if required, recover $T$ in Eq.~(\ref{eq:matrix-conjugacy}). Otherwise, report ``Not conjugate''. A randomized subroutine failure is reported as ``Failure''.

\paragraph{Correctness of the algorithm.} Every actual code conjugacy satisfies Eq.~(\ref{eq:code-pair-conjugacy}). When the source commuting algebra has dimension $1$, all simultaneous conjugators are nonzero scalar multiples of the computed representative. These multiples induce the same conjugation on matrix codes, so the final span-equality test checks every possible code conjugacy. Hence every output other than ``Failure'' is correct.

\paragraph{Running time analysis.} The hulls, spectral projectors, complements, Gram systems, commuting algebras, and final span test are computable with polynomially many field operations. Polynomial factorization and tuple conjugacy use the algorithms already cited in Section~\ref{sec:alg-iso} and Theorem~\ref{thm:module-iso}. The running time is randomized $\poly(n,\log q)$, with no enumeration of field elements or conjugators.

\paragraph{Average-case analysis.} The trace-dual matrix is determined by the code, so its distribution must be analysed after conditioning on the hull. We use the same incidence comparison as in Proposition~\ref{prop:hull-contraction}.
\begin{proposition}\label{prop:matrix-code-projection}
Fix $q$, and let $C$ be a uniform $n$-dimensional subspace of $\M(n,q)$, conditioned on having a one-dimensional hull whose nonzero generators have exactly one $\F_q$-eigenvalue, nonzero and of algebraic multiplicity $1$. Let $A$, $P$, $D$, and $Z$ be constructed as above. Put $k=n-1$, $m=n^2-2$, and let $\varepsilon_{\mathrm{gen}}(n,q)$ be the independent-pair generation failure probability in Theorem~\ref{thm:gen-matrix-alg}. For all sufficiently large $n$ depending on $q$,
\begin{equation}\label{eq:matrix-code-projection-generation}
    \Pr[A,Z\text{ do not generate }\M(n,q)]
    \leq16(q^{-k}+q^{k-m})
       +Kq^4\varepsilon_{\mathrm{gen}}(n,q)
    =q^{-\Omega(n)},
\end{equation}
where $K$ is an absolute constant.
\end{proposition}
\begin{proof}
Write $L=\linspan\{A\}$ and $N=n^2$. Since $A$ has been normalized to eigenvalue $1$ and $P$ is its rank-one spectral projector,
\[
    \Tr(A^2)=0,\qquad \Tr(AP)=\Tr(P^2)=1.
\]
Thus the restriction of the trace form to $\linspan\{A,P\}$ has Gram matrix $\left[\begin{smallmatrix}0&1\\1&1\end{smallmatrix}\right]$, which is nonsingular in every characteristic. Its orthogonal complement
\[
    U:=\linspan\{A,P\}^\perp
\]
is therefore nondegenerate of dimension $m=N-2$.

Fix the hull line $L$ and first sample $C$ uniformly subject to $L\subseteq C\subseteq L^\perp$, $\dim C=n$. The map $C\mapsto D=C\cap P^\perp$ is a bijection from these codes to the $k$-dimensional subspaces of $U$: its inverse is $D\mapsto L\oplus D$, since $\Tr(AP)=1$. Moreover, $\hull(C)=L$ holds exactly when $D$ is nondegenerate. The incidence argument in the proof of Proposition~\ref{prop:hull-contraction} shows that this additional event has probability $p_n\geq1/8$ for each allowed $L$, for fixed $q$ and sufficiently large $n$. Hence $D$ is a uniform $k$-subspace of $U$ conditioned on nondegeneracy, with conditioning cost at most $8$.

Put $\tau=\Tr (A)$ and
\[
    R:=\idmat_n-(1-\tau)A-\tau P.
\]
The displayed trace identities give $R\in U$. Also, $R\neq0$. Otherwise, restriction to $\ker P$ gives $\idmat=(1-\tau)A|_{\ker P}$. If $\tau=1$, this is impossible; otherwise, $A$ is scalar over $\F_q$ on the nonzero space $\ker P$, contradicting its unique simple $\F_q$-eigenvalue. Since every $Y\in D$ is orthogonal to $A$ and $P$, we have $\Tr(RY)=\Tr(Y)$. Consequently, $Z$ is the orthogonal projection of $R$ onto $D$.

For a fixed $x\in U\setminus\{0,R\}$, the event $Z=x$ requires $x\in D\subseteq(R-x)^\perp$.
If $x$ is not orthogonal to $R-x$, there are no such subspaces. Otherwise, there are $\genfrac{[}{]}{0pt}{}{m-2}{k-1}_q$ possible subspaces before imposing nondegeneracy. Thus
\begin{align*}
    \Pr[Z=x\mid L]
    &\leq\frac{1}{p_n}
       \frac{\genfrac{[}{]}{0pt}{}{m-2}{k-1}_q}
            {\genfrac{[}{]}{0pt}{}{m}{k}_q}=\frac{1}{p_n}
       \frac{(q^k-1)(q^{m-k}-1)}{(q^m-1)(q^{m-1}-1)}
      \leq32q^{-m+1}=32q^3q^{-N}.
\end{align*}
The exceptional projections satisfy $Z=0$ exactly when $D\subseteq R^\perp$, and $Z=R$ exactly when $R\in D$. The corresponding subspace counts give
\[
    \Pr[Z\in\{0,R\}\mid L]
    \leq\frac{1}{p_n}
       \frac{q^{m-k}+q^k-2}{q^m-1}
    \leq16(q^{-k}+q^{k-m}).
\]

To compare with independent ambient matrices, choose an unnormalized generator $H$ uniformly from the nonzero elements of $L$. Conditional on the required spectral property, Proposition~\ref{prop:random-self-dual} makes $H$ uniform among all self-dual matrices with that property: the excluded scalar matrices fail the spectral condition. Equation~(\ref{eq:selfdual-nonzero-lower}) bounds its density relative to a uniform matrix in $\M(n,q)$ by $2e^{23/9}q$. The preceding point bound holds conditional on every such $H$, with $A$ its normalization. Outside the two exceptional projections, the joint point probabilities of $(H,Z)$ are therefore bounded by $64e^{23/9}q^4$ times the independent uniform law. Theorem~\ref{thm:gen-matrix-alg} bounds generation failure on this part of the distribution. Adding the exceptional probability proves Eq.~(\ref{eq:matrix-code-projection-generation}), since rescaling $H$ to $A$ does not change the generated algebra.
\end{proof}

Proposition~\ref{prop:hull-dim} gives a one-dimensional hull with probability $\Theta(1/q)$. Conditional on this event, Proposition~\ref{prop:random-self-dual} and Theorem~\ref{thm:main-technical} give a constant probability of the required nonzero simple eigenvalue. Proposition~\ref{prop:matrix-code-projection} then gives a generating canonical pair with probability $1-q^{-\Omega(n)}$. Its commuting algebra consists of scalars, as in the proof of Corollary~\ref{cor:conj-bound}, so the algorithm succeeds for every second code on these first codes. Standard bounded repetition controls randomized polynomial-factorization failures with polynomial overhead. Thus the algorithm succeeds for a $\Theta(1/q)$ fraction of uniform first codes in time $\poly(n,\log q)$, for fixed $q$ and sufficiently large $n$.

For the tensor formulation in Theorem~\ref{thm:matrix-sp-conj}, the $n$ slices are linearly independent with probability $1-q^{-\Omega(n^2)}$. Conditional on independence, their span is a uniform $n$-dimensional code. The preliminary rank checks therefore extend the same success-fraction bound to that formulation.

\section{Average-case algorithm for 4-tensor isomorphism}\label{sec:4-tensor-iso}

In this section, we present the average-case algorithm for $4$-tensor isomorphism to prove Theorem~\ref{thm:4-tensor-iso}.

\subsection{Algorithm description}
\paragraph{Problem set up.} Let $U,V,W,X\cong\F_q^n$, and let $\tA, \tB\in U\otimes V\otimes W\otimes X$ be two $4$-tensors. After fixing bases of these vector spaces, $\tA$ is represented by a $4$-way array $\tA=(a_{i,j,k,\ell})\in \F_q^n\otimes\F_q^n\otimes\F_q^n\otimes\F_q^n$. The goal is to determine whether there exist invertible matrices $L, R, S, T \in \GL(n,q)$ such that:
\[\tB = (L \otimes R \otimes S \otimes T) \cdot \tA.\]
A useful viewpoint to deal with higher-order tensor isomorphism is to consider the \emph{flattened version} of higher-order tensors. More precisely, we flatten a $4$-tensor $\tA\in \F_q^n\otimes\F_q^n\otimes\F_q^n\otimes\F_q^n$ into an $n^2\times n^2$ matrix $\tilde{A}\in \M(n^2,q)$ by mapping
\[
    \stdb_i\otimes\stdb_j\otimes\stdb_k\otimes\stdb_\ell\mapsto (\stdb_i\otimes\stdb_j)\trp{(\stdb_k\otimes\stdb_\ell)}
\]
for any $i,j,k,\ell\in[n]$. We shall call $\tilde{A}$ the \emph{flattened version} of $\tA$. Moreover, for any two $4$-tensors $\tA,\tB$ and their flattened versions $\tilde{A},\tilde{B}$, it is clear that 
\[
    \tB = (L \otimes R \otimes S \otimes T) \cdot \tA~\Leftrightarrow~\tilde{B}=(L\otimes R)\cdot \tilde{A}\cdot\trp{(S\otimes T)}.
\]

\paragraph{Isomorphism invariant: left kernel spaces.}
We define the \emph{left kernel space} $\mathcal{L}_\tA$ of a $4$-tensor $\tA\in\F_q^n\otimes\F_q^n\otimes\F_q^n\otimes\F_q^n$ as
\[
    \mathcal{L}_\tA := \{ \trp{v}\in \F_q^{1 \times n^2}:~\trp{v}\tilde{A}=0\}.
\]
We then have the following:
\begin{proposition}\label{prop:left-kernel-space is isomorphism invariant}
    If $4$-tensors $\tA,\tB\in \F_q^n\otimes\F_q^n\otimes\F_q^n\otimes\F_q^n$ are isomorphic with respect to invertible matrices $L,R,S,T\in\GL(n,q)$, then $\trp{v}\in\cL_\tB$ if and only if $\trp{v}(L\otimes R)\in\cL_\tA$. Moreover, we have $\cL_{\tB}=\cL_{\tA}\cdot(L\otimes R)^{-1}$.
\end{proposition}
\begin{proof}
    Recall that $\tilde{B}=(L\otimes R)\cdot\tilde{A}\cdot\trp{(S\otimes T)}$. Thus, if $\trp{v}\tilde{B}=0$, we have $\trp{v}(L\otimes R)\cdot\tilde{A}\cdot\trp{(S\otimes T)}=0$. Since $S\otimes T$ is invertible, we have $\trp{v}(L\otimes R)\in\cL_\tA$. 

    On the other hand, if $\trp{v}(L\otimes R)\in\cL_\tA$, we also have $\trp{v}(L\otimes R)\cdot\tilde{A}\cdot\trp{(S\otimes T)}=0$ since $\trp{v}(L\otimes R)\cdot\tilde{A}=0$. Thus $\trp{v}\in\cL_\tB$.

    Moreover, since $L\otimes R$ is invertible and acts linearly in $\F_q^{1 \times n^2}$, $L\otimes R:~\trp{v}\mapsto\trp{v}(L\otimes R)$ is an invertible linear map in $\F_q^{1 \times n^2}$. It follows that $\cL_{\tB}=\cL_{\tA}\cdot(L\otimes R)^{-1}$.
\end{proof}

Proposition~\ref{prop:left-kernel-space is isomorphism invariant} indicates that the left kernel space of a $4$-tensor is an isomorphism invariant. Moreover, such an invariant is easy to compute by solving a system of linear equations of size $n^2$. 

\paragraph{Algorithm description.}
We reshape each row vector in $\cL_\tA$ into an $n\times n$ matrix and denote the resulting matrix code by $\cC_\tA$; define $\cC_\tB$ similarly. If row vectors satisfy $\trp{u}=\trp{v}(L\otimes R)$, their reshaped matrices satisfy $A_u=\trp{L}A_vR$. Therefore, Proposition~\ref{prop:left-kernel-space is isomorphism invariant} gives
\[
    \cC_\tB=L^{-\mathrm t}\cC_\tA R^{-1},
    \qquad\text{equivalently}\qquad
    \trp{L}\cC_\tB R=\cC_\tA.
\]
The order of $\cC_\tA$ and $\cC_\tB$ in the latter equation is important when recovering the tensor factors. We also define the right kernel, written as a row space, by
\[
    \mathcal R_\tA:=\{\trp{w}:\trp{w}\trp{\tilde A}=0\},
\]
and denote its reshaped matrix code by $\mathcal D_\tA$; define $\mathcal R_\tB$ and $\mathcal D_\tB$ in the same way. Applying Proposition~\ref{prop:left-kernel-space is isomorphism invariant} to the transposed flattenings gives
\[
    \mathcal D_\tB=S^{-\mathrm t}\mathcal D_\tA T^{-1}.
\]
We will enumerate the possible pairs $(L,R)$ and $(S,T)$ from these two conditions, and then test the tensor equation directly.

The algorithm is as follows.
\begin{enumerate}
    \item Compute $\tilde A,\tilde B$ and the four kernel codes $\cC_\tA,\cC_\tB,\mathcal D_\tA,\mathcal D_\tB$.
    \item If $\rk(\tilde A)\neq\rk(\tilde B)$, output ``$\tA$ and $\tB$ are not isomorphic''. If their common corank is not $3$, return ``Fail''.
    \item Apply the kernel-matching procedure below to $(\cC_\tA,\cC_\tB)$ to obtain the set $K_1$ of possible pairs $(L,R)$.
    \item Apply the same procedure to $(\mathcal D_\tA,\mathcal D_\tB)$ to obtain the set $K_2$ of possible pairs $(S,T)$.
    \item For each $(L,R)\in K_1$ and $(S,T)\in K_2$, test whether $\tB=(L\otimes R\otimes S\otimes T)\cdot\tA$. If a pair passes, output the isomorphism $(L,R,S,T)$. If none passes, output ``$\tA$ and $\tB$ are not isomorphic''.
\end{enumerate}
If either kernel-matching procedure or any randomized subroutine returns ``Fail'', the algorithm also returns ``Fail''.

We describe the kernel-matching procedure for two $3$-dimensional codes $\cC_\tA,\cC_\tB$. Enumerate the at most $q^3$ elements of $\cC_\tA$ to find an invertible matrix $A_1$. If none exists, return ``Fail''. Extend $A_1$ to an ordered basis $(A_1,A_2,A_3)$, and put $X_i=A_1^{-1}A_i$ for $i=2,3$. Compute generators of
\[
    Z_\tA:=\{Q\in\GL(n,q):QX_i=X_iQ\text{ for }i=2,3\}
\]
using Theorem~\ref{thm:module-iso}. Starting from the identity, enumerate the generated group by multiplying the discovered matrices by generators and their inverses, recording each matrix once. Stop and return ``Fail'' as soon as $q+1$ distinct elements have been found. Otherwise, the whole group $Z_\tA$ has been enumerated.

Now enumerate all ordered bases $(B_1,B_2,B_3)$ of $\cC_\tB$. Skip a basis if $B_1$ is singular. For the other bases, the equations for the actual tensor factors are
\begin{equation}\label{eq: matrix tuple equivalence}
    \trp{L}B_iR=A_i,\qquad i=1,2,3.
\end{equation}
The first equation gives $\trp{L}=A_1R^{-1}B_1^{-1}$, and substitution into the other two gives
\begin{equation}\label{eq: matrix tuple conjugacy}
    R^{-1}B_1^{-1}B_iR=A_1^{-1}A_i,\qquad i=2,3.
\end{equation}
Use Theorem~\ref{thm:module-iso} to find one solution $R_0$ of Eq.~(\ref{eq: matrix tuple conjugacy}). If no solution exists, continue to the next basis. If a solution exists, the complete set of solutions is the right coset $R_0Z_\tA$. For every $R\in R_0Z_\tA$, put
\[
    \left(\trp{(A_1R^{-1}B_1^{-1})},R\right)
\]
into $K_1$. This finishes the procedure. The same construction using $\mathcal D_\tA,\mathcal D_\tB$ produces $K_2$.

We explain the correctness and running time. Whenever the procedure finishes, Eq.~(\ref{eq: matrix tuple equivalence}) shows that every pair in $K_1$ satisfies the required left-kernel equation. Conversely, any pair satisfying the left-kernel equation determines the ordered basis $B_i=L^{-\mathrm t}A_iR^{-1}$ of $\cC_\tB$, whose first member is invertible. This basis is considered by the procedure, and the full coset $R_0Z_\tA$ includes its second factor. Thus $K_1$ contains every possible pair $(L,R)$. The same argument applies to $K_2$. It follows that every output other than ``Fail'' is correct: an isomorphism is verified directly, and a nonisomorphism output follows either from unequal ranks or from testing all possible pairs of factors.

There are $|\GL(3,q)|<q^9$ ordered bases per target code, and the group enumeration is capped at $q$. Hence $|K_1|,|K_2|\leq q^{10}$. Finding an invertible source matrix, computing the tuple conjugacy cosets, and enumerating the bounded groups take $\poly(n,q)$ time. The final tests take at most $q^{20}\poly(n,\log q)$ bit operations. The randomized subroutines can be implemented with a correct result or ``Fail'' using bounded independent attempts, as in Section~\ref{subsec:orth-algorithm}. There are at most $N_4:=2q^9+2$ such calls, so we give each call failure probability at most $1/(4N_4)$ and use fresh randomness. The union bound then gives a total subroutine failure probability of at most $1/4$, while retaining the $\poly(n,q)$ running-time bound on every input and every run.

\subsection{Average-case analysis}
We shall show that, for fixed $q$ and sufficiently large $n$, the above algorithm succeeds for a $1/\poly(q)$ fraction of uniform $4$-tensors $\tA$ and any second tensor $\tB$. We first estimate the probability that the kernel dimension is $3$, and then bound the pointwise stabilizers of the two kernel codes.

\paragraph{Random $4$-tensors with constant-dimensional kernels.}
\begin{proposition}\label{prop:random 4 tensor has constant left kernel}
Let $N=n^2$ and $0\leq c\leq N$. For $\tA\rin\T(n\times n\times n\times n,q)$,
\begin{equation}\label{eq: random matrix have constant corank}
    \Pr[\dim\cL_\tA=c]
    =q^{-c^2}\frac{c_N(q)^2}{c_c(q)^2c_{N-c}(q)}.
\end{equation}
In particular, for each fixed $c$, this probability is $\Theta(q^{-c^2})$ for $N\geq c$.
\end{proposition}
\begin{proof}
The flattening $\tilde A$ is uniform in $\M(N,q)$. The classical rank-counting formula~\cite[Eq.~(1)]{Fulman2015} gives Eq.~(\ref{eq: random matrix have constant corank}). The stated bounds follow from $1/4<c(q)\leq c_j(q)\leq1$, with $c_0(q)=1$. In particular,
\[
    \Pr[\dim\cL_\tA=3]\geq \frac{q^{-9}}{16}.
\]
\end{proof}
Conditional on $\rk(\tilde A)=N-3$, its left and right kernels are independent uniform $3$-dimensional subspaces. Indeed, for any prescribed right kernel and left kernel, $\tilde A$ induces an isomorphism from the quotient by the right kernel to the annihilator of the left kernel. There are $|\GL(N-3,q)|$ such isomorphisms, independently of the prescribed pair. Reshaping therefore gives independent uniform $3$-dimensional matrix codes $\cC_\tA$ and $\mathcal D_\tA$.

\paragraph{Bounding the kernel-code stabilizers.}
For a $3$-dimensional matrix code $\cC$, define its pointwise stabilizer by
\[
    G(\cC):=\{(P,Q)\in\GL(n,q)^2:\trp{P}XQ=X\text{ for every }X\in\cC\}.
\]
This group depends only on the code, and not on its ordered basis. We call $\cC$ good if it contains an invertible matrix and
\[
    G(\cC)=\{(\lambda\idmat_n,\lambda^{-1}\idmat_n):\lambda\in\F_q^\times\}.
\]
For a good code, any invertible choice of $A_1$ and any extension to a basis give $Z_\tA=\{\lambda\idmat_n:\lambda\in\F_q^\times\}$. Indeed, the correspondence
\[
    Q\longmapsto\left(\trp{(A_1Q^{-1}A_1^{-1})},Q\right)
\]
is a bijection from $Z_\tA$ to $G(\cC)$. Thus the deterministic basis choices in the algorithm do not affect whether the group cutoff is satisfied.

To estimate the probability of a good code, first sample $A_1,A_2,A_3$ independently and uniformly from $\M(n,q)$. Conditional on their linear independence, their span is a uniform $3$-dimensional code. Conditional on $A_1$ being invertible, $X_2=A_1^{-1}A_2$ and $X_3=A_1^{-1}A_3$ are independent uniform matrices. Let $\varepsilon_{\mathrm{gen}}(n,q)=q^{-\Omega(n)}$ be their probability of failing to generate $\M(n,q)$, as in Theorem~\ref{thm:gen-matrix-alg}. If they generate $\M(n,q)$, their common centralizer consists of scalar matrices, and the resulting code is good. This event also forces $\idmat_n,X_2,X_3$, and hence $A_1,A_2,A_3$, to be linearly independent when $n\geq2$. Consequently,
\[
    \Pr[\cC\text{ is good}]
    \geq c_n(q)(1-\varepsilon_{\mathrm{gen}}(n,q)).
\]
For fixed $q$ and sufficiently large $n$, the right-hand side is at least $1/8$. By the conditional independence of the two kernel codes, both are good with probability at least $1/64$, conditional on corank $3$.

It follows that a uniform $\tA$ has corank $3$ and two good kernel codes with probability at least $q^{-9}/1024$ for sufficiently large $n$. On every such $\tA$, both kernel-matching procedures finish for every $\tB$ unless a randomized subroutine fails. The subroutine failure probability is at most $1/4$, including when $\tB$ depends on $\tA$ but is independent of the internal randomness. Therefore, the total success probability is at least $3q^{-9}/4096=1/\poly(q)$, with running time $\poly(n,q)$. This proves Theorem~\ref{thm:4-tensor-iso} in the fixed-$q$ asymptotic regime.

\section{Average-case algorithm for orthogonal tensor isomorphism}\label{sec:orth-tensor-iso}

In this section, we present the average-case algorithm for orthogonal tensor isomorphism to prove Theorem~\ref{thm:orth-tensor-iso}.

\subsection{Algorithm description}\label{subsec:orth-algorithm}

\paragraph{Problem set up.}
Let $q$ be odd, let $n\geq2$, and let $U,V,W\cong\F_q^n$. Write
\[
    \mathrm O(n,q)=\{Q\in\GL(n,q):\trp{Q}Q=\idmat_n\}.
\]
After fixing bases of $U,V,W\cong\F_q^n$, let $\tA=(a_{i,j,k}),\tB=(b_{i,j,k})\in U\otimes V\otimes W$ be two $3$-tensors. The goal is to determine whether there exist orthogonal matrices $Q_1,Q_2,Q_3\in\mathrm O(n,q)$ such that
\[
    \tB=(Q_1\otimes Q_2\otimes Q_3)\cdot\tA,
\]
and, if so, to compute such matrices. As in Section~\ref{sec:4-tensor-iso}, we first consider a flattened version of the tensors. For $\tA$, define $M_\tA\in\M(n\times n^2,q)$ by $(M_\tA)_{i,(j,k)}=a_{i,j,k}$, and define $M_\tB$ in the same way. Then
\[
    \tB=(Q_1\otimes Q_2\otimes Q_3)\cdot\tA
    \quad\Longleftrightarrow\quad
    M_\tB=Q_1M_\tA\trp{(Q_2\otimes Q_3)}.
\]

\paragraph{Isomorphism invariant: Gram matrices.}
Denote $\sym(n,q)=\{S\in\M(n,q):\trp{S}=S\}$. We define the \emph{Gram matrix} of $\tA$ as $S_\tA:=M_\tA\trp{M_\tA}\in\sym(n,q)$, and define $S_\tB$ similarly. Taking Gram matrices eliminates the actions in the second and third directions, giving the following.
\begin{proposition}\label{prop:orth-gram-invariant}
If $\tA$ and $\tB$ are orthogonally isomorphic with respect to $Q_1,Q_2,Q_3\in\mathrm O(n,q)$, then
\begin{equation}\label{eq:orth-gram}
    M_\tB=Q_1M_\tA\trp{(Q_2\otimes Q_3)},\qquad
    S_\tB=Q_1S_\tA\trp{Q_1}.
\end{equation}
\end{proposition}
\begin{proof}
The first equality follows from the definition of the flattened versions. Since $Q_2\otimes Q_3$ is an element of $\mathrm O(n^2,q)$, we have
\begin{equation}\label{eq: gram equation}
    S_\tB=M_\tB\trp{M_\tB}
    =Q_1M_\tA\trp{(Q_2\otimes Q_3)}(Q_2\otimes Q_3)\trp{M_\tA}\trp{Q_1}
    =Q_1S_\tA\trp{Q_1}.
\end{equation}
\end{proof}
Proposition~\ref{prop:orth-gram-invariant} indicates that the orthogonal conjugacy class of the Gram matrix is an isomorphism invariant. Moreover, the Gram matrix is easy to compute by matrix multiplication.

\paragraph{Algorithm description.}
We first use the Gram matrices to constrain the choices of $Q_1$. For $S\in\sym(n,q)$, let
\[
    C_{\mathrm O}(S):=\{g\in\mathrm O(n,q):gS=Sg\}\leq \mathrm O(n,q)
\]
be the centralizer of $S$ inside the orthogonal group. Denote $z_{\mathrm O}(S):=|C_{\mathrm O}(S)|$.
Suppose $Q_0\in\mathrm O(n,q)$ satisfies $S_\tB =Q_0S_\tA \trp{Q_0}$. Then the orthogonal matrices $Q$ satisfying $S_\tB =QS_\tA \trp{Q}$ form the coset
\begin{equation}\label{eq:orth-coset}
    \{Q\in\mathrm O(n,q):S_\tB =QS_\tA \trp{Q}\}=Q_0C_{\mathrm O}(S_\tA).
\end{equation}
Indeed, writing $Q=Q_0g$, the Gram equation becomes $gS_\tA \trp{g}=S_\tA $, which is equivalent to $gS_\tA =S_\tA g$ since $g$ is orthogonal. Thus, once one solution $Q_0$ is known, we can enumerate all solutions to the Gram equation~\eqref{eq: gram equation} by enumerating the orthogonal centralizer of $S_\tA $. 

Once we obtain solutions to the Gram equation~\eqref{eq: gram equation}, let $Q_1$ be such a fixed solution and define $\tC=(\trp{Q_1}\otimes\idmat_n\otimes\idmat_n)\tB$. Extending this particular $Q_1$ to an orthogonal tensor isomorphism is equivalent to finding $Q_2,Q_3\in\mathrm O(n,q)$ such that $\tC=(\idmat_n\otimes Q_2\otimes Q_3)\tA$.

Let $(A_1,\ldots,A_n)$ and $(C_1,\ldots,C_n)$ be the horizontal matrix tuples of $\tA$ and $\tC$, respectively. The remaining task is to find $Q_2,Q_3\in\mathrm O(n,q)$ such that for any $i\in[n]$,
\begin{equation}\label{eq:orth-slices}
    C_i=Q_2A_i\trp{Q_3}.
\end{equation}
This is a matrix-tuple equivalence problem with two orthogonal transformations. We can solve it by testing isometry of symmetric matrix tuples as follows. For $X\in\M(n,q)$, let 
\begin{equation}\label{eq:orth-augmented-tuples}
    \widehat X=\begin{pmatrix}0&X\\\trp{X}&0\end{pmatrix},~~
    E=\begin{pmatrix}\idmat_n&0\\0&0\end{pmatrix},~~\text{and}~~
    \cT_\tA=(\idmat_{2n},E,\widehat A_1,\ldots,\widehat A_n).
\end{equation}
Define $\widehat C_i$ for $i\in[n]$ and $\cT_\tC$ in the same way. Then, $\cT_\tA$ and $\cT_\tC$ are both symmetric matrix tuples of size $n+2$. Note that any invertible matrix $R$ such that $\trp{R}\cT_\tA R=\cT_\tC$ should have a block diagonal form $R=\begin{bmatrix}R_2&0\\0&R_3\end{bmatrix}$ with $R_2,R_3\in \mathrm O(n,q)$: First, $\trp{R}\idmat_{2n}R=\idmat_{2n}$ implies that $R\in\mathrm O(2n,q)$. Then, $\trp{R}ER=E$ implies $RE=ER$, which forces $R$ to have the block diagonal form $\begin{bmatrix}R_2&0\\0&R_3\end{bmatrix}$ with $R_2\in\mathrm O(n,q)$. This also implies $R_3\in\mathrm O(n,q)$ since $R\in\mathrm O(2n,q)$. 

Let
\[
    G_\tA:=C_{\mathrm O}(\cT_\tA)
    =\{h\in\mathrm O(2n,q):hT=Th\text{ for every entry }T\text{ of }\cT_\tA\}.
\]
Because the tuple entries are symmetric, this is also their common isometry group. If $R_0\in\mathrm O(2n,q)$ satisfies $\trp{R_0}\cT_\tA R_0=\cT_\tC$, then the complete solution set is
\begin{equation}\label{eq:orth-extension-coset}
    \{R\in\mathrm O(2n,q):\trp{R}\cT_\tA R=\cT_\tC\}
    =G_\tA R_0.
\end{equation}
Indeed, every $hR_0$ with $h\in G_\tA$ is a solution, and every solution $R$ gives $R R_0^{-1}\in G_\tA$. Each solution is block diagonal, so we set $Q_2=\trp{R_2}$ and $Q_3=\trp{R_3}$ to recover its two orthogonal factors. The group $G_\tA$ depends only on $\tA$ and can be computed once for all first-factor candidates.

The above procedure gives the following algorithm. The inputs are $\tA,\tB\in\T(n\times n\times n,q)$, a positive integer $L$, and a failure bound $\eta\in(0,1)$.
\begin{enumerate}
    \item Compute the flattened versions $M_\tA,M_\tB$, the Gram matrices $S_\tA,S_\tB$, and the tuple $\cT_\tA$. Initialize $\cO=\emptyset$.
    \item Compute $Q_0\in\mathrm O(n,q)$ such that $S_\tB=Q_0S_\tA\trp{Q_0}$. If no such matrix exists, output ``$\tA$ and $\tB$ are not orthogonally isomorphic''.
    \item Compute generators of $C_{\mathrm O}(S_\tA)$ and enumerate the group. If more than $L$ distinct elements are found, return ``Fail''.
    \item Compute generators of $G_\tA=C_{\mathrm O}(\cT_\tA)$ and enumerate the group. If more than $2$ distinct elements are found, return ``Fail''.
    \item For each $g\in C_{\mathrm O}(S_\tA)$, let $Q_1=Q_0g$ and compute $\tC$ and $\cT_\tC$. Compute $R_0$ such that $\trp{R_0}\cT_\tA R_0=\cT_\tC$. If no such matrix exists, continue to the next $g$. Otherwise, for every $R\in G_\tA R_0$, write $R=\operatorname{diag}(R_2,R_3)$ and put $(Q_1,\trp{R_2},\trp{R_3})$ into $\cO$ if $\tB=(Q_1\otimes\trp{R_2}\otimes\trp{R_3})\cdot\tA$.
    \item If $\cO=\emptyset$, output ``$\tA$ and $\tB$ are not orthogonally isomorphic''; otherwise, output $\cO$.
\end{enumerate}
Each randomized subroutine uses fresh randomness and has failure probability at most $\eta$. If any subroutine reports ``Fail'', the algorithm also returns ``Fail''.

We explain some steps in the algorithm and analyse its time complexity. First, in Step 2, a solution $Q_0$ can be computed by the matrix-tuple isometry algorithm of Ivanyos and Qiao~\cite[arXiv version, Theorem~7(1)]{IQ19}, applied to the ordered symmetric tuples $(\idmat_n,S_\tB)$ and $(\idmat_n,S_\tA)$. 
Its equations are $\trp{Q_0}Q_0=\idmat_n$ and $\trp{Q_0}S_\tB Q_0=S_\tA$, which give precisely the required conditions. The algorithm runs in randomized time $\poly(n,\log q,\log(1/\eta))$.

Second, in Step 3, the common isometry group of $(\idmat_n,S_\tA)$ is exactly $C_{\mathrm O}(S_\tA)$. Indeed, under $\trp{g}g=\idmat_n$, the equation $\trp{g}S_\tA g=S_\tA $ is equivalent to $gS_\tA =S_\tA g$. A generating set can therefore be computed by Brooksbank and Wilson~\cite[Theorem~1.2]{BW12}, applied to the symmetric bilinear map $b_{S_\tA }(x,y)=(\trp{x}y,\trp{x}S_\tA y)$. Starting from $\idmat_n$, we enumerate the generated group by multiplying each discovered element by every generator and its inverse, recording each new element once. The search stops as soon as $L+1$ distinct elements have been found. If $z_{\mathrm O}(S_\tA)\leq L$, the whole group is enumerated, and Eq.~(\ref{eq:orth-coset}) ensures that every possible first factor is considered.

In Step 4, generators of $G_\tA$ can again be computed by Brooksbank--Wilson~\cite{BW12}, now applied to $\cT_\tA$, and the bounded group enumeration is performed in the same way as in Step 3. In Step 5, isometry of $\cT_\tA$ and $\cT_\tC$ can be tested by Ivanyos--Qiao~\cite{IQ19}. The first two equations give $\trp{R}R=\idmat_{2n}$ and $\trp{R}ER=E$. Hence $ER=RE$, so $R=\operatorname{diag}(R_2,R_3)$ with $R_2,R_3\in\mathrm O(n,q)$. The remaining equations become
\[
    \trp{R_2}A_iR_3=C_i,\qquad\forall i\in[n].
\]
These are precisely Eq.~(\ref{eq:orth-slices}) with $Q_2=\trp{R_2}$ and $Q_3=\trp{R_3}$. Conversely, any solution of Eq.~(\ref{eq:orth-slices}) gives an isometry $R=\operatorname{diag}(\trp{Q_2},\trp{Q_3})$. Thus Step 5 tests exactly whether the chosen $Q_1$ can be extended to an orthogonal tensor isomorphism. Equation~(\ref{eq:orth-extension-coset}) enumerates all such extensions.

It follows that every isomorphism returned is correct. If both group enumerations finish and no subroutine fails, the output contains all orthogonal isomorphisms, since every possible first factor and every extension are considered. A nonisomorphism output is also correct: it either follows from the Gram equation having no solution, or from testing all possible first factors without finding an extension. The algorithm returns ``Fail'' when either enumeration bound is exceeded or a randomized subroutine fails.

The time complexity of the above algorithm is dominated by two parts: enumerating $C_{\mathrm O}(S_\tA)$ and testing the resulting first factors. There are at most $L$ first-factor candidates, at most $2L$ output triples, and at most $L+3$ subroutine calls, on at most $n+2$ symmetric forms of dimension at most $2n$. The algorithms of Ivanyos--Qiao and Brooksbank--Wilson admit implementations with a correct result or ``Fail'' in time polynomial in the input size and $\log(1/\eta)$~\cite[arXiv version, Theorem~18 and Section~3.5]{IQ19}\cite{BW12}. This follows by using bounded independent attempts to reduce their failure probabilities. The generating set $\Gamma_A$ has polynomial size. Using a balanced search tree to record the group elements, the enumeration takes $O((L+1)|\Gamma_A\cup\Gamma_A^{-1}|\log(L+2))$ matrix operations, where $\Gamma_A$ is the generating set in Step 3. All other work per candidate is polynomial in $n$ and $\log q$. Consequently, the time complexity of the above algorithm is at most 
\begin{equation}\label{eq:orth-time}
(L+1)\log(L+2)\poly(n,\log q,\log(1/\eta)).
\end{equation}

\subsection{Average-case analysis}\label{subsec:orth-average}
We shall show that, for a uniform $3$-tensor $\tA$ in $\F_q^n\otimes\F_q^n\otimes\F_q^n$ and an arbitrary $\tB$, the above algorithm succeeds in time $\poly_s(n,\log q)$ with probability $1-O_s(n^{-s})$, for every fixed positive integer $s$. For this, we bound the sizes of both $C_{\mathrm O}(S_\tA)$ and $C_{\mathrm O}(\cT_\tA)$. We first state the two random-matrix estimates whose proofs are given in Section~\ref{subsec:orth-random-matrices}.

\paragraph{Random Gram matrices.}
Write $U_{n,q}$ for the uniform distribution on $\sym(n,q)$ and $\mu_{n,m,q}$ for the distribution of $M\trp{M}$ when $M\rin\M(n\times m,q)$. For random $\tA$, its flattened version $M_\tA\in\M(n\times n^2,q)$ has independent uniform entries, so $S_\tA$ has distribution $\mu_{n,n^2,q}$. For distributions $\rho,\sigma$ on a finite set $\Omega$, their total variation distance is defined as
\[
    d_{\mathrm{TV}}(\rho,\sigma)
    =\frac12\sum_{x\in\Omega}|\rho(x)-\sigma(x)|
    =\max_{F\subseteq\Omega}|\rho(F)-\sigma(F)|.
\]
If $\sigma(x)>0$ for every $x\in\Omega$, their $\chi^2$-divergence is defined as
\begin{equation}\label{eq:orth-chi-square}
    \chi^2(\rho\Vert\sigma)
    :=\sum_{x\in\Omega}\frac{(\rho(x)-\sigma(x))^2}{\sigma(x)}.
\end{equation}
By Cauchy--Schwarz and $\sum_x\sigma(x)=1$, these quantities satisfy
\begin{equation}\label{eq:orth-tv-chi-square}
    d_{\mathrm{TV}}(\rho,\sigma)
    =\frac12\sum_{x\in\Omega}\sqrt{\sigma(x)}
        \frac{|\rho(x)-\sigma(x)|}{\sqrt{\sigma(x)}}
    \leq\frac12\sqrt{\chi^2(\rho\Vert\sigma)}.
\end{equation}
In particular, a bound on the $\chi^2$-divergence controls the probability difference of every event under the two distributions.
The following proposition gives the required comparison with a uniform symmetric matrix, as well as the bound on $\delta_{n,q}$.

\begin{proposition}\label{prop:orth-gram-uniform}
For odd $q$ and $m>n$, let $R_{n,r}=|\{T\in\sym(n,q):\rk(T)=r\}|$. Then
\begin{equation}\label{eq:orth-gram-fourier}
 \chi^2(\mu_{n,m,q}\Vert U_{n,q})
 =\sum_{r=1}^n R_{n,r}q^{-mr},
\end{equation}
and
\begin{equation}\label{eq:orth-gram-tv}
 d_{\mathrm{TV}}(\mu_{n,m,q},U_{n,q})
 \leq \frac{1}{2\sqrt{c(q)}}\left(\sum_{r=1}^n q^{-(m-n)r-r(r-1)/2}\right)^{1/2}
 \leq\frac12\left(\frac{q^{-(m-n)}}{c(q)(1-q^{-(m-n)})}\right)^{1/2}.
\end{equation}
In particular, for $n\geq2$,
\begin{equation}\label{eq:orth-delta}
\begin{split}
 d_{\mathrm{TV}}(\mu_{n,n^2,q},U_{n,q})
 &\leq\delta_{n,q}:=\frac12\left(\frac{q^{-n(n-1)}}{c(q)(1-q^{-n(n-1)})}\right)^{1/2}\\
 &\leq\frac34q^{-n(n-1)/2}.
\end{split}
\end{equation}
\end{proposition}

See Section~\ref{subsubsec:orth-gram-proof} for the proof of Proposition~\ref{prop:orth-gram-uniform}.

\paragraph{Bounding the Gram centralizer.}
To bound $z_{\mathrm O}(S_\tA)$, we use the following first-moment estimate in the uniform symmetric-matrix model.

\begin{proposition}\label{prop:orth-centralizer-moment}
For $S$ uniform in $\sym(n,q)$ with $q$ odd,
\begin{equation}\label{eq:orth-centralizer-moment}
 \mathbb{E}_S z_{\mathrm O}(S)
 \leq \frac{2(1+q^{-1})}{c(q)^2}(n+1)
 \leq12(n+1).
\end{equation}
\end{proposition}

See Section~\ref{subsubsec:orth-centralizer-proof} for the proof of Proposition~\ref{prop:orth-centralizer-moment}.
Fix a positive integer $s$ and choose
\begin{equation}\label{eq:orth-budget}
    L_s:=\left\lceil24(n+1)^{s+1}\right\rceil,\qquad
    N_s:=L_s+3,\qquad
    \eta_s:=\frac{1}{2(n+1)^sN_s}.
\end{equation}
Denote by $\mathcal A_s$ the algorithm with enumeration bound $L_s$ for $C_{\mathrm O}(S_\tA)$, enumeration bound $2$ for $C_{\mathrm O}(\cT_\tA)$, and subroutine failure bound $\eta_s$. For $S$ uniform in $\sym(n,q)$, Proposition~\ref{prop:orth-centralizer-moment} and Markov's inequality give
\[
    \Pr[z_{\mathrm O}(S)>L_s]
    \leq\frac{12(n+1)}{L_s}
    \leq\frac{1}{2(n+1)^s}.
\]
Applying Proposition~\ref{prop:orth-gram-uniform} to the event $\{S:z_{\mathrm O}(S)>L_s\}$, we obtain
\begin{equation}\label{eq:orth-enumeration-failure}
    \Pr_{\tA}[z_{\mathrm O}(S_\tA)>L_s]
    \leq\frac{1}{2(n+1)^s}+\delta_{n,q}.
\end{equation}

\paragraph{Bounding the tuple centralizer.}
We now bound $C_{\mathrm O}(\cT_\tA)$, which is enumerated in Step 4. 
Define
\begin{equation}\label{eq:orth-tuple-failure}
    \varepsilon_{n,q}:=\Pr_{\tA}\bigl[
    C_{\mathrm O}(\cT_\tA)\neq\{\pm\idmat_{2n}\}\bigr].
\end{equation}
The following proposition shows that the probability of $C_{\mathrm O}(\cT_\tA)\neq\{\pm\idmat_{2n}\}$ is exponentially small. In particular, the tuple centralizer has order at most $2$ with high probability.
\begin{proposition}\label{prop:orth-tuple-centralizer}
Let $q$ be odd and $n\geq3$, and choose $\tA\rin\T(n\times n\times n,q)$. Let $\varepsilon_{\mathrm{gen}}(n,q)$ be the probability that two independent uniform matrices in $\M(n,q)$ fail to generate $\M(n,q)$. Then
\begin{equation}\label{eq:orth-tuple-centralizer-bound}
    \Pr_{\tA}\bigl[C_{\mathrm O}(\cT_\tA)\neq\{\pm\idmat_{2n}\}\bigr]
    \leq(1-c_n(q))^{n-2}+\varepsilon_{\mathrm{gen}}(n,q)
    =q^{-\Omega(n)}.
\end{equation}
Consequently, $C_{\mathrm O}(\cT_\tA)=\{\pm\idmat_{2n}\}$ with probability $1-q^{-\Omega(n)}$.
\end{proposition}
\begin{proof}
The horizontal slices $A_1,\ldots,A_n$ are independent uniform matrices in $\M(n,q)$. Let $j$ be the first index such that $A_j$ is invertible, with $j=\infty$ if no such index exists. By Proposition~\ref{prop:glnq and mnq}, each slice is invertible with probability $c_n(q)$, so
\[
    \Pr[j>n-2]=(1-c_n(q))^{n-2}.
\]
Condition on $j\leq n-2$ and on $A_1,\ldots,A_j$. The next two slices remain independent and uniform. Since $A_j$ is invertible, the matrices
\[
    X=A_j^{-1}A_{j+1},\qquad Y=A_j^{-1}A_{j+2}
\]
are also independent and uniform in $\M(n,q)$. By Theorem~\ref{thm:gen-matrix-alg}, they generate $\M(n,q)$ except with probability $\varepsilon_{\mathrm{gen}}(n,q)=q^{-\Omega(n)}$.

Let $h\in C_{\mathrm O}(\cT_\tA)$. Commutation with $E$ forces $h=\operatorname{diag}(U,V)$, where $U,V\in\mathrm O(n,q)$. Commutation with every $\widehat A_i$ gives
\[
    UA_i=A_iV,\qquad\forall i\in[n].
\]
As $A_j$ is invertible, $U=A_jVA_j^{-1}$. Substituting this into the equations for $i=j+1,j+2$ shows that $V$ commutes with $X$ and $Y$. When $X,Y$ generate $\M(n,q)$, $V$ must therefore be scalar, as in the proof of Corollary~\ref{cor:conj-bound}. Thus $V=\lambda\idmat_n$ and $U=\lambda\idmat_n$. Orthogonality gives $\lambda^2=1$, so $h=\pm\idmat_{2n}$ because $q$ is odd. Both sign matrices always belong to the centralizer.

The exceptional probability is at most $(1-c_n(q))^{n-2}+\varepsilon_{\mathrm{gen}}(n,q)$. Finally,
\[
    c_n(q)\geq1-\sum_{i=1}^nq^{-i}\geq1-\frac{1}{q-1},
\]
so $(1-c_n(q))^{n-2}\leq(q-1)^{-(n-2)}=q^{-\Omega(n)}$. This proves Eq.~(\ref{eq:orth-tuple-centralizer-bound}).
\end{proof}
On this event, every nonempty extension coset $C_{\mathrm O}(\cT_\tA)R_0$ consists of $R_0$ and $-R_0$, corresponding to the two pairs $(Q_2,Q_3)$ and $(-Q_2,-Q_3)$. The event depends only on $\tA$ and holds simultaneously for every second tensor $\tB$ and every first-factor candidate $Q_1$.

\paragraph{Failures of the randomized subroutines.}
Both group enumerations finish on the inputs satisfying the two centralizer bounds. The remaining possible failures come from the randomized matrix-tuple computations. We have the following.
\begin{proposition}\label{prop:orth-conditional}
Suppose $z_{\mathrm O}(S_\tA)\leq L_s$ and $|C_{\mathrm O}(\cT_\tA)|\leq2$. Then, for every $\tB$, the probability that the algorithm returns a correct answer is at least $1-\tfrac12(n+1)^{-s}$.
\end{proposition}
\begin{proof}
There are at most $N_s=L_s+3$ subroutine calls. These comprise two calls to Brooksbank--Wilson and at most $L_s+1$ calls to Ivanyos--Qiao. Each uses fresh randomness and has failure probability at most $\eta_s$, conditional on the preceding computation. A union bound therefore gives total failure probability at most $N_s\eta_s=\tfrac12(n+1)^{-s}$. If no call fails, the two centralizer bounds ensure that both group enumerations finish, and the correctness argument in Section~\ref{subsec:orth-algorithm} gives the result.
\end{proof}

Combining these estimates gives the following result. 
\begin{theorem}\label{thm:orth-tensor-iso}
Let $q$ be odd and $n\geq3$. For every fixed positive integer $s$, the randomized algorithm described in Sec.~\ref{subsec:orth-algorithm} uses $\poly_s(n,\log q)$ bit operations on every input $\tA,\tB\in\T(n\times n\times n,q)$ and every run. It outputs the complete set of orthogonal isomorphisms, correctly reports nonisomorphism, or returns ``Fail''. Every output other than ``Fail'' is correct. For $\tA$ uniform and any choice $\tB=\tB(\tA)$ independent of the internal randomness $\omega$,
\[
    \Pr_{\tA,\omega}[\mathcal A_s(\tA,\tB;\omega)\text{ returns a correct answer}]
    \geq1-(n+1)^{-s}-\delta_{n,q}-\varepsilon_{n,q}.
\]
On a successful run there are at most $2L_s$ orthogonal isomorphisms. In particular, the success probability is $1-O_s(n^{-s})$ as $n\to\infty$, uniformly over odd $q$.
\end{theorem}
\begin{proof}
The running time follows from Eq.~(\ref{eq:orth-time}), since $L_s=O_s(n^{s+1})$ and $\log(1/\eta_s)=O_s(\log(n+1))$. Correctness of every completed output follows from Section~\ref{subsec:orth-algorithm}. By Eq.~(\ref{eq:orth-enumeration-failure}) and the definition of $\varepsilon_{n,q}$, the probability that either group exceeds its enumeration bound is at most
\[
    \tfrac12(n+1)^{-s}+\delta_{n,q}+\varepsilon_{n,q}.
\]
No independence between the two centralizer conditions is needed for this union bound. On each remaining input, Proposition~\ref{prop:orth-conditional} bounds the subroutine failure probability by $\tfrac12(n+1)^{-s}$. This estimate holds for every $\tB$, including a choice depending on $\tA$ but independent of $\omega$. Averaging over $\tA$ gives the stated bound. Finally, both $\delta_{n,q}$ and $\varepsilon_{n,q}$ decay exponentially in $n$, proving the last assertion.
\end{proof}
Choosing $s=1$ gives success probability $1-O(1/n)$ with enumeration bound $L_1=24(n+1)^2$. More generally, any fixed positive integer $s$ gives failure probability $O_s(n^{-s})$, with the polynomial degree depending on $s$.

\begin{corollary}\label{cor:orth-exponential-success}
For odd $q$ and $n\geq3$, orthogonal tensor isomorphism admits a randomized $q^{O(n)}$-time algorithm with success probability $1-q^{-\Omega(n)}$ when the first tensor is uniform and the second tensor is arbitrary and independent of the internal randomness. Every output other than ``Fail'' is correct.
\end{corollary}
\begin{proof}
Use the algorithm of Section~\ref{subsec:orth-algorithm} with $L=q^n$ and $\eta=q^{-n}/(L+3)$. Proposition~\ref{prop:orth-centralizer-moment}, Markov's inequality and Proposition~\ref{prop:orth-gram-uniform} bound the probability that the Gram centralizer exceeds $L$ by $12(n+1)q^{-n}+\delta_{n,q}$. The tuple-centralizer exception has probability $\varepsilon_{n,q}=q^{-\Omega(n)}$, and the subroutine failures have total probability at most $(L+3)\eta=q^{-n}$. Their sum is $q^{-\Omega(n)}$. Since $\log(1/\eta)=O(n\log q)$, Eq.~(\ref{eq:orth-time}) gives the claimed running time.
\end{proof}

\section{Random matrices over finite fields}\label{sec:random matrix}
In this section, we prove the results regarding random matrices over finite fields required for our algorithm analysis. All limits are taken with $q$ fixed, and every quantitative estimate applies in its stated range of $n$ and $q$. The properties of random matrices that interest us are as follows. 
\begin{itemize}
    \item Eigenvalue condition: matrices with a unique eigenvalue of algebraic multiplicity $1$. See Theorem~\ref{thm:eigen-condition} and Corollary~\ref{cor: probability of unique nonzero eigenvalue} and their proofs in Section~\ref{subsec:eigen-condition}; these results describe the spectral probabilities without a trace constraint.
    \item Self-dual condition: matrices $A$ with $\Tr(A^2)=0$. See Theorem~\ref{thm:random matrix being self-dual} and its proof in Section~\ref{subsec: random self-dual}; this result is required to prove Theorem~\ref{thm:main-technical}.
    \item Both eigenvalue and self-dual conditions: matrices $A$ with a unique eigenvalue of algebraic multiplicity $1$, conditioned on $\Tr(A^2)=0$. See Theorem~\ref{thm:main-technical} and its proof in Section~\ref{subsec: random main}; this result supports Theorem~\ref{thm:main-technical-informal}.
    \item Gram matrices and orthogonal centralizers: comparison of random Gram matrices with uniform symmetric matrices, and a first-moment bound for orthogonal centralizer orders. See Propositions~\ref{prop:orth-gram-uniform} and~\ref{prop:orth-centralizer-moment}, stated in Section~\ref{subsec:orth-average} and proved in Section~\ref{subsec:orth-random-matrices}; these estimates support Theorem~\ref{thm:orth-tensor-iso}.
\end{itemize}

\subsection{Random matrices with a unique eigenvalue of algebraic multiplicity 1}\label{subsec:eigen-condition}
We shall estimate the probability of a random matrix $A\in\M(n,q)$ with exactly one eigenvalue in $\F_q$, of algebraic multiplicity $1$. But before that, we present a generalization of the counting method proposed by Neumann and Praeger~\cite{NEUMANN_PRAEGER_1998} to count matrices over finite fields with prefixed algebraic multiplicities.

\paragraph{Matrices over $\F_q$ with prefixed algebraic multiplicities.} 
Let $m_A:\F_q\to\N$ be the function that maps $\F_q$-eigenvalues of $A$ to its algebraic multiplicity. It is clear that $m_A$ is a vector over $\N$ of length $q$. Denote $D(m_A)=\sum_{\lambda\in\F_q} m_A(\lambda)$ and $r_A=n-D(m_A)$. Let $c_A(t)=\det(t\idmat-A)=f_0(t)\prod_{\lambda\in\F_q}(t-\lambda)^{m_A(\lambda)}$ be the characteristic polynomial of $A$, where $f_0(t)$ is a degree-$r_A$ polynomial over $\F_q$ satisfying $f_0(t)\neq 0$ for any $t\in\F_q$. The set of $n\times n$ matrices over $\F_q$ with respect to a given function $m:\F_q\to\N$ is denoted as
\[
C_m(n,q):=\{A\in \M(n,q):~c_A(t)=f_0(t)\prod_{\lambda\in\F_q}(t-\lambda)^{m(\lambda)},~f_0(t)\neq 0~\forall\ t\in \F_q\}.
\]

We call a matrix $A\in\M(n,q)$ $\lambda$-potent, if $c_A(t)=(t-\lambda)^n$ for some $\lambda\in\F_q$. Equivalently, $A$ is conjugate to $\lambda \idmat_n+N$ for some nilpotent matrix $N$. $1$-potent matrices are more widely known as unipotent matrices. Note that the number of $\lambda$-potent matrices for any fixed $\lambda$ in $\M(n,q)$ is $q^{n^2-n}$~\cite{MR96677} (see also~\cite[Thm. 15.1]{MR230728}).

Neumann and Praeger studied the problem of counting eigenvalue-free matrices over $\F_q$~\cite{NEUMANN_PRAEGER_1998} (i.e.~$m(\lambda)=0$ for any $\lambda\in\F_q$) via a generating function method to obtain asymptotic formulas and quantitative bounds for such counting problems. 

We shall adapt these results to study $|C_m(n,q)|$. For this, we recall some notation and results in~\cite{NEUMANN_PRAEGER_1998} (see also~\cite{Ful02}). We denote by
\begin{equation}\label{eq: eigenvalue-free}
v(\GL,n,q):=\frac{|\{\text{eigenvalue-free matrices in}~\GL(n,q)\}|}{|\GL(n,q)|}=\frac{|C_0(n,q)|}{|\GL(n,q)|}
\end{equation}
the proportion of $n\times n$ eigenvalue-free matrices over $\F_q$ in $\GL(n,q)$, and by
\begin{equation}\label{eq: unipotent}
u(\GL,n,q):=\frac{|\{\text{unipotent matrices in}~\GL(n,q)\}|}{|\GL(n,q)|}     
\end{equation}
the proportion of $n\times n$ unipotent matrices over $\F_q$ in $\GL(n,q)$. For $n=0$, we let $u(\GL,0,q):=v(\GL,0,q)=1$. The generating functions of $v(\GL,n,q)$ and $u(\GL,n,q)$ are denoted as 
\[
V(\GL,q,z)=\sum_{n=0}^\infty v(\GL,n,q)z^n~\text{and}~U(\GL,q,z)=\sum_{n=0}^\infty u(\GL,n,q)z^n,
\]
respectively. We summarise some useful results on the expressions, relations, and quantitative bounds between $V(\GL,q,z)$ and $U(\GL,q,z)$ as follows:
\begin{proposition}\label{lem: gen functions of v and u}
Recall that $c_n(q)=\prod_{i=1}^n(1-q^{-i})$ and $c(q)=\lim_{n\to\infty}c_n(q)$ as in Eq.~(\ref{eq: ratio of GLnq and Mnq}), we have
\begin{itemize}
    \item $V(\GL,q,z)U(\GL,q,z)^{q-1}=(1-z)^{-1}$~\cite[Theorem 4.1]{NEUMANN_PRAEGER_1998};
    \item $u(\GL,n,q)=\frac{1}{c_n(q)\cdot q^n}$~\cite[Theorem 4.2]{NEUMANN_PRAEGER_1998};
    \item Let $v(\GL,\infty,q):=\lim_{n\to\infty} v(\GL,n,q)$. Then $v(\GL,\infty,q)=c(q)^{q-1}$ and 
    \[
        v(\GL,\infty,q)-v(\GL,n,q)=(-1)^{n+1}\varepsilon_n,~0<\varepsilon_n<\frac{1}{2}16^{q-1}q^{-\frac{(n+1)(n+q)}{2(q-1)}}
    \]
    whenever $n>5(q-1)^2$~\cite[Theorem 4.3]{NEUMANN_PRAEGER_1998}.
\end{itemize}
\end{proposition}

Note that the polynomial factorization of $c_A(t)=f_0(t)\prod_{\lambda\in\F_q}(t-\lambda)^{m_A(\lambda)}$ corresponds to the primary decomposition of $\F_q^n=V_0\oplus\bigoplus_{\lambda\in\F_q} V_\lambda$, where $\dim(V_0)=r_A$ and $\dim(V_\lambda)=m_A(\lambda)$ for any $\lambda\in\F_q$. In particular, $A$ is conjugate to a block diagonal matrix, where $A_0$ acts on $V_0$ as an eigenvalue-free matrix and $A_\lambda$ acts on $V_\lambda$ as a $\lambda$-potent matrix.

This block-diagonal form leads to the following strategy to count the matrices in $C_m(n,q)$ for any given function $m:\F_q\to\N$: We go over all the direct-sum decompositions of $\F_q^n=V_0\oplus\bigoplus_{\lambda\in\F_q} V_\lambda$, where $\dim(V_0)=r=n-D(m)$ and $\dim(V_\lambda)=m(\lambda)$ for any $\lambda\in\F_q$. Then we assign each $V_0$ an eigenvalue-free matrix $A_0$ and each $V_\lambda$ a $\lambda$-potent matrix for each $\lambda\in\F_q$. This results in the following expression of $|C_m(n,q)|$:
\begin{equation}\label{eq:expression of Cmnq}
\begin{split}
|C_m(n,q)|&=|\GL(n,q)|\cdot v(\GL,r,q)\cdot\prod_{\lambda\in\F_q}\frac{q^{m(\lambda)^2-m(\lambda)}}{|\GL(m(\lambda),q)|}\\
&=|\GL(n,q)|\cdot v(\GL,r,q)\cdot\prod_{\lambda\in\F_q}\frac{q^{-m(\lambda)}}{|c_{m(\lambda)}(q)|},
\end{split}
\end{equation}
where the second equality uses the convention that $|\GL(m(\lambda),q)|/|\M(m(\lambda),q)|=c_{m(\lambda)}(q)$ as defined in Eq.~(\ref{eq: ratio of GLnq and Mnq}). With Eq.~(\ref{eq:expression of Cmnq}), it is clear that we can express the probability of a random matrix $A\in\M(n,q)$ whose $\F_q$-eigenvalues' algebraic multiplicities are given by a fixed function $m:\F_q\to\N$ as
\[
\alpha_m(n,q):=\frac{|C_m(n,q)|}{|\M(n,q)|}=c_n(q)\cdot v(\GL,r,q)\cdot\prod_{\lambda\in\F_q}\frac{q^{-m(\lambda)}}{|c_{m(\lambda)}(q)|}.
\]

In particular, if the function $m:\F_q\to\N$ satisfies that $D(m)=O(1)$, we have the following:
\begin{theorem}\label{thm:eigenvalue-multiplicity-profile}
    Let $\alpha_m(n,q)$ be defined as above with respect to some function $m:\F_q\to\N$ satisfying $D(m)=O(1)$. Then we have
    \[
    \alpha_m(\infty,q):=\lim_{n\to\infty}\alpha_m(n,q)=c(q)^q\cdot\prod_{\lambda\in\F_q}\frac{q^{-m(\lambda)}}{|c_{m(\lambda)}(q)|}.
    \]
\end{theorem}
\begin{proof}
    Since $D(m)=O(1)$, $m(\lambda)=O(1)$ for any $\lambda\in\F_q$ and $r\to\infty$ as $n\to \infty$. Thus, $\prod_{\lambda\in\F_q}\frac{q^{-m(\lambda)}}{|c_{m(\lambda)}(q)|}=O(1)$ and $\lim_{n\to\infty}c_n(q)\cdot v(\GL,r,q)=c(q)\cdot v(\GL,\infty,q)=c(q)^q$.
\end{proof}

Recall that we aim to compute the probability of a random matrix $A$ that has only $1$ eigenvalue in $\F_q$ with algebraic multiplicity $1$. These matrices correspond to functions $m:\F_q\to\N$ whose vector forms in $\N^q$ are precisely the standard basis $\stdb_0,\dots,\stdb_{q-1}\in\N^q$. Let 
\[
C_1(n,q):=\{A\in\M(n,q):~A~\text{has only}~1~\text{eigenvalue in}~\F_q~\text{with algebraic multiplicity}~1\}.
\]
It is clear that $C_1(n,q)=\bigcup_{i=0}^{q-1}C_{\stdb_i}(n,q)$. Then we have the following:

\begin{theorem}\label{thm:eigen-condition}
Let $\alpha(n,q):=\frac{|C_1(n,q)|}{|\M(n,q)|}$ be the probability of a random matrix $A\in\M(n,q)$ which has only $1$ eigenvalue in  $\F_q$ with algebraic multiplicity $1$. Then we have
\[
\alpha(\infty,q):=\lim_{n \to \infty}\alpha(n,q)=\frac{q\cdot c(q)^{q}}{q-1}\geq e^{-23/9}.
\]
Moreover, for any $n>5(q-1)^2$, we have:
\begin{equation}\label{eq:~convergence bound alpha}
|\alpha(\infty,q)-\alpha(n+1,q)|\leq \frac{q}{q-1}(\frac{1}{2}16^{q-1}q^{-(n+1)(n+q)/(2(q-1))}+4q^{-(n+1)}).    
\end{equation}
\end{theorem}
\begin{proof}
Note that $\alpha_{\stdb_i}(n,q)=c_n(q)v(\GL,n-1,q)\frac{q^{-1}}{1-q^{-1}}$ for any $i=0,\dots,q-1$. Then we have
\begin{equation}\label{eq: alpha n q}
    \alpha(n,q)=\sum_{i=0}^{q-1}\alpha_{\stdb_i}(n,q)=\frac{q}{q-1}v(\GL,n-1,q)c_n(q).
\end{equation}
Thus,
\[
\alpha(\infty,q)=\lim_{n \to \infty}\alpha(n,q)=\frac{q\cdot c(q)^{q}}{q-1},    
\] 
where we use $v(\GL,n,q)\to v(\GL,\infty,q)=c(q)^{q-1}$ as in Proposition~\ref{lem: gen functions of v and u} and $c_n(q)\to c(q)$ as in Proposition~\ref{prop:glnq and mnq}.
For the convergence rate, Eq.~(\ref{eq: alpha n q}) gives
\[
\begin{split}
|\alpha(\infty,q)-\alpha(n+1,q)|
&=\frac{q}{q-1}\left|c(q)v(\GL,\infty,q)-c_{n+1}(q)v(\GL,n,q)\right|\\
&\leq\frac{q}{q-1}\left(c(q)|v(\GL,\infty,q)-v(\GL,n,q)|
    +v(\GL,n,q)|c(q)-c_{n+1}(q)|\right).
\end{split}
\]
For $n>5(q-1)^2$, Proposition~\ref{lem: gen functions of v and u} bounds the first term by
$\frac12 16^{q-1}q^{-(n+1)(n+q)/(2(q-1))}$.
The second is at most $4q^{-(n+2)}\leq4q^{-(n+1)}$ by Proposition~\ref{prop:glnq and mnq} and $v(\GL,n,q)\leq1$. This proves Eq.~(\ref{eq:~convergence bound alpha}). The lower bound on $\alpha(\infty,q)$ follows from Proposition~\ref{prop: cqq}.
\end{proof}

Theorem~\ref{thm:eigen-condition} also gives the probability that the only $\F_q$-eigenvalue of a random matrix is nonzero and has algebraic multiplicity $1$. We restrict the summation in Eq.~(\ref{eq: alpha n q}) to the nonzero eigenvalues. More precisely, let
\[
C_1^*(n,q):=\left\{A\in\M(n,q):~c_A(t)=(t-\lambda)f_0(t),~
\lambda\in\F_q^\times,~f_0(a)\neq0~\forall a\in\F_q\right\}.
\]
The corresponding nonzero-eigenvalue probability is as follows.
\begin{corollary}\label{cor: probability of unique nonzero eigenvalue}
    Let $\alpha^*(n,q):=\frac{|C_1^*(n,q)|}{|\M(n,q)|}$ be the probability that the only $\F_q$-eigenvalue of a random matrix $A\in\M(n,q)$ is nonzero and has algebraic multiplicity $1$. Then we have
\[
\alpha^*(\infty,q):=\lim_{n \to \infty}\alpha^*(n,q)=c(q)^{q}> e^{-23/9}.
\]
Moreover, for any $n>5(q-1)^2$, we have:
\begin{equation}\label{eq:~convergence bound alpha^*}
|\alpha^*(\infty,q)-\alpha^*(n+1,q)|\leq \frac{1}{2}16^{q-1}q^{-(n+1)(n+q)/(2(q-1))}+4q^{-(n+1)}.    
\end{equation}
\end{corollary}
\begin{proof}
    We have 
    \[
    \alpha^*(n,q)=\frac{|C_1^*(n,q)|}{|\M(n,q)|}=v(\GL,n-1,q)c_n(q).
    \]
    Thus, $\alpha^*(\infty,q)=c(q)^q$, and the error bound follows by the same argument as in the proof of Theorem~\ref{thm:eigen-condition}.
\end{proof}

\subsection{Random matrices that are self-dual}\label{subsec: random self-dual}
We then estimate the probability of a random matrix being self-dual. We have the following.
\begin{theorem}\label{thm:random matrix being self-dual}
    Let $\sigma(n,q)$ be the probability of a random matrix $A\in\M(n,q)$ being self-dual, i.e.~$\Tr(A^2)=0$. Then we have
    \[
        \sigma(\infty,q):=\lim_{n\to\infty}\sigma(n,q)=\frac{1}{q}.
    \]
    Moreover, we have 
    \begin{equation}\label{eq: convergence bound sigma}
        \left|\sigma(\infty,q)-\sigma(n,q)\right|\leq \Sigma(n,q):= (q-1)q^{-\frac{n^2}{2}-1}.
    \end{equation}
\end{theorem}

To prove Theorem~\ref{thm:random matrix being self-dual}, we distinguish between the cases of odd and even $\Char(\F_q)$. For $\Char(\F_q)=2$, the calculation is relatively straightforward. 

For odd $\Char(\F_q)$, we could note that the number of solutions of $\Tr(A^2)=0$ over $\F_q$ for odd $q$ is equivalent to counting points on a nonsingular quadric over $\F_q$. This has been extensively studied in projective geometry over finite fields (see~\cite[Sec.~6]{Lidl_Niederreiter_1996} or~\cite[Sec.~5]{10.1093/oso/9780198502951.001.0001}), and the desired result could be deduced by following these standard estimates. Still, we shall take a different approach, namely the character sum method as used by Gorodetsky and Rodgers~\cite{MR4273172}, to prove this case for Theorem~\ref{thm:random matrix being self-dual}. This is because this method will also be adopted to prove Theorem~\ref{thm:main-technical}, and we hope that this proof can serve as a warm-up for the readers who are not familiar with character sums.

\subsubsection{Preparations for characters over finite fields}
We first review some necessary notation and results about characters of finite fields before presenting the proof.

For a finite field $\F_q$, we have $q=p^m$, where $p=\Char(\F_q)$ is a prime number called the characteristic of $\F_q$. Let $\Tr_{\F_q/\F_p}: \F_q\to \F_p$ be the absolute trace function from $\F_q$ to $\F_p$, defined as 
\[
\Tr_{\F_q/\F_p}(a)=a+a^p+\cdots +a^{p^{m-1}}.
\]
Then the function $\psi:~\F_q\to \C$, defined as
\begin{equation}\label{eq: additive character}
    \psi(a)=e^{\frac{2\pi i}{p}\Tr_{\F_q/\F_p}(a)}
\end{equation}
is a character of the \emph{additive group} of $\F_q$. We shall call characters of the additive group of $\F_q$ \emph{additive characters}. In fact, all additive characters of $\F_q$ can be expressed in terms of $\psi$:
\begin{lemma}[Theorem 5.7 in~\cite{Lidl_Niederreiter_1996}]
    For $b\in\F_q$, the function $\psi_b:~\F_q\to \C$ defined as $\psi_b(a)=\psi(ba)$ for all $a\in\F_q$ is an additive character of $\F_q$, and every additive character of $\F_q$ is obtained in this way.
\end{lemma}
By the representation theory of finite abelian groups, we have the following orthogonality result:
\begin{lemma}\label{lem:orthogonal}
    $\forall\ a_1,a_2\in\F_q$, 
    \begin{equation*}
         \frac{1}{q}\sum_{b\in\F_q}\psi_b(a_1)\overline{\psi_b(a_2)}=
         \left\{
         \begin{aligned}
             1 &  & \text{if}~a_1=a_2\\
             0 &  & \text{if}~a_1\neq a_2
         \end{aligned}~.
         \right.    
    \end{equation*}
\end{lemma}

We use $\cM_{n,q}\subseteq \F_q[t]$ to denote the set of all monic polynomials of degree $n$ over $\F_q$ and $\cM_q=\bigcup_{n\geq 0}\cM_{n,q}$. We also use $\cM_{n,q}^{gl}\subseteq \cM_{n,q}$ to denote the set of all monic polynomials of degree $n$ over $\F_q$ which are coprime with the polynomial $f(t)=t$ and $\cM_q^{gl}=\bigcup_{n\geq 0}\cM_{n,q}^{gl}$. It is clear that if $A\in\GL(n,q)$, then its characteristic polynomial $c_A(t)\in\cM_{n,q}^{gl}$.

As we will deal with $\Tr(A^2)$ for $A\in\M(n, q)$, we introduce a function $\chi$ that is useful for connecting to $\Tr(A^2)$. Let $f(t)\in \cM_{n, q}$. For $\lambda\in\F_q$, we define a function $\chi_{\lambda}:\cM_q\to\C$ by
\begin{equation}\label{eq:chi}
   \chi_\lambda(f(t))= e^{\frac{2\pi i}{p}\Tr_{\F_q/\F_p}(\lambda(\alpha_1^2+\cdots+\alpha_n^2))}
\end{equation}
where $\alpha_1,\dots,\alpha_n $ are the roots of $f(t)$ in $\overline{\F_q}$, listed with multiplicities.
 In particular, if $f(t)=c_A(t)$ for some matrix $A\in \M(n,q)$, we have
\begin{equation}\label{eq:characteristic polynomial and character of Fq}
   \chi_\lambda(c_A(t))= e^{\frac{2\pi i}{p}\Tr_{\F_q/\F_p}(\lambda\Tr(A^2))}=\psi_\lambda(\Tr(A^2)).
\end{equation}
Combining Eq.~(\ref{eq:characteristic polynomial and character of Fq}) with Lemma~\ref{lem:orthogonal}, where we set $a_1=\Tr(A^2)$ and $a_2=k$, we obtain the following character-sum expression for the condition of $\Tr(A^2)=k$:
\begin{proposition}\label{lem: character sum expression for self-dual}
    For any matrix $A\in \M(n,q)$, we have the following.
    \begin{equation*}
         \mathbb{I}\{\Tr(A^2)=k\}=\frac{1}{q}\sum_{\lambda\in\F_q}\psi_\lambda(\Tr(A^2))\overline{\psi_\lambda(k)}=\frac{1}{q}\sum_{\lambda\in\F_q}\chi_\lambda(c_A(t))\overline{\psi_\lambda(k)}=
         \left \{
         \begin{aligned}
             1 &  & \text{if}~\Tr(A^2)=k\\
             0 &  & \text{if}~\Tr(A^2)\neq k
         \end{aligned}
         ~.
         \right.
    \end{equation*}
\end{proposition}

\subsubsection{Proof of Theorem~\ref{thm:random matrix being self-dual}}

With the indicator function above, we are ready to compute the probability of a random matrix $A\in\M(n,q)$ being self-dual, i.e.~$\Tr(A^2)=0$. 
\begin{proof}[Proof of Theorem~\ref{thm:random matrix being self-dual}]
Our main focus will be on counting the solutions of $\Tr(A^2)=0$.
Note that for $A=(a_{i,j})_{i,j\in[n]}\in\M(n,q)$, we have
\begin{equation}\label{eq: expression of trace} \Tr(A^2)=\sum_{i=1}^n\sum_{j=1}^na_{i,j}a_{j,i}=\sum_{i=1}^na_{i,i}^2+\sum_{1\le i \ne j \le n}a_{i,j}a_{j,i}=\sum_{i=1}^na_{i,i}^2+2\sum_{1\le i<j \le n}a_{i,j}a_{j,i}.
\end{equation}

\paragraph{Case 1. $\Char(\F_q)=2$.} In this case, $2\sum_{1\le i<j \le n}a_{i,j}a_{j,i}=0$, and hence $\Tr(A^2)=\sum_{i=1}^na_{i,i}^2=\left(\sum_{i=1}^na_{i,i}\right)^2$.
We only need to calculate the number of solutions to the linear equation:
\[
\sum_{i=1}^na_{i,i}=0,
\]
which is $q^{n-1}$.
 Thus, we have 
 \[
     \Pr[\Tr(A^2)=0|~A\in_R\M(n,q)]=\frac{q^{n-1}}{q^n}=\frac{1}{q}.
 \]


\paragraph{Case 2. $\Char(\F_q)$ is odd.} 
Now we estimate the probability of a random matrix $A\in\M(n,q)$ being self-dual, i.e.~$\Tr(A^2)=0$.

By Proposition~\ref{lem: character sum expression for self-dual}, we know that
\begin{equation}\label{eq: probability of self-dual}
\begin{split}
    \sigma(n,q)=\Pr[\Tr(A^2)=0|~A\in_R\M(n,q)]&=\frac{1}{|\M(n,q)|}\sum_{A\in\M(n,q)}\mathbb{I}\{\Tr(A^2)=0\}\\
    &=\frac{1}{|\M(n,q)|}\sum_{A\in\M(n,q)}\left(\frac{1}{q}\sum_{\lambda\in\F_q}\psi_\lambda(\Tr(A^2))\right)\\
    &=\frac{1}{q}\sum_{\lambda\in\F_q}\underbrace{\left(\frac{1}{|\M(n,q)|}\sum_{A\in\M(n,q)}\psi(\lambda\cdot\Tr(A^2))\right)}_{:=\mathbb{E}[\psi(\lambda\cdot\Tr(A^2))]}.
\end{split}
\end{equation}
Note that when $\lambda=0$, we have
\begin{equation}\label{eq: expectation value when lambda=0}
    \mathbb{E}[\psi(0\cdot\Tr(A^2))]=1.
\end{equation}
We shall focus on estimating $\mathbb{E}[\psi(\lambda\cdot\Tr(A^2))]$ for any $\lambda\in\F_q^\times$.
 We utilize Eq.~(\ref{eq: expression of trace}) to obtain
 \begin{equation}\label{eq: expectation decomposition}
 \begin{split}
     \mathbb{E}[\psi(\lambda\cdot \Tr(A^2))]&=\mathbb{E}[\psi(\lambda (\sum_{i=1}^na_{i,i}^2+2\sum_{1\le i<j \le n}a_{i,j}a_{j,i}))]\\
     &=\mathbb{E}[\prod_{i=1}^n\psi(\lambda a_{i,i}^2)\cdot\prod_{1\le i<j \le n}\psi(2\lambda a_{i,j}a_{j,i})]\\
     &=\mathbb{E}[\psi(\lambda a^2)]^{n}\cdot\mathbb{E}[\psi(2 \lambda xy)]^{n(n-1)/2},
 \end{split}
 \end{equation}
 where the second equality uses $\psi(a+b)=\psi(a)\psi(b)$ for any $a,b\in\F_q$ and the third equality uses the fact that all the $a_{i,j}$'s are chosen independently and uniformly at random.

We first estimate
\[
\mathbb{E}[\psi(2\lambda xy)]=\frac{1}{q^2}\sum_{x,y\in\F_q}\psi(2\lambda xy)
\]
for any $\lambda\in\F_q^\times$.
Note that for any $x\in\F_q^\times$, the map $y\mapsto 2\lambda xy$ is a bijection in $\F_q$. Thus, 
$\sum_{y\in\F_q}\psi(2\lambda xy)=0$. This implies that for any $\lambda\in\F_q^\times$
\begin{equation}\label{eq: estimate of E2}
\mathbb{E}[\psi(2\lambda xy)]=\frac{1}{q^2}\sum_{x,y\in\F_q}\psi(2\lambda xy)=\frac{1}{q^2}\sum_{y\in\F_q}\psi(0)=\frac{1}{q}.    
\end{equation}

Now, we estimate
 \[
 \mathbb{E}[\psi(\lambda a^2)]=\frac{1}{q}\sum_{x\in \F_q}\psi(\lambda x^2)
 \] 
 for any $\lambda\in\F_q^\times$. Recall that an element $\lambda\in\F_q$ is called a \emph{quadratic residue} if $\lambda=c^2$ for some $c\in\F_q$; otherwise, $\lambda$ is called a quadratic nonresidue. Denote by $\mathbf{qr}(q)$ the set of nonzero quadratic residues in $\F_q$, and denote by $\mathbf{nqr}(q)$ the set of nonzero quadratic nonresidues in $\F_q$. It is clear that $\F_q^\times=\mathbf{qr}(q)\cup \mathbf{nqr}(q)$ and $\mathbf{qr}(q)\cap\mathbf{nqr}(q)=\emptyset$.
 We have the following lemma.
 \begin{lemma}\label{lem: sum of psi lambda x2}
 For any $\lambda\in\F_q^\times$, we have
 \begin{equation*}
         \sum_{x\in \F_q}\psi(\lambda x^2)=
         \left\{
         \begin{aligned}
             &\sum_{x\in\F_q}\psi(x^2) &  & \text{if}~\lambda\in\mathbf{qr}(q)&\\
             &-\sum_{x\in\F_q}\psi(x^2) &  & \text{if}~\lambda\in\mathbf{nqr}(q)&
         \end{aligned}~.
         \right.    
    \end{equation*}
 \end{lemma}
 \begin{proof}
 Note that if $\lambda\in\mathbf{qr}(q)$, then $\lambda x^2\in\mathbf{qr}(q)$ for any $x\in\F^\times_q$. Moreover, $\lambda x^2$ goes through $\mathbf{qr}(q)$ twice when $x$ goes through $\F_q^\times$, since the quadratic map $x\mapsto x^2$ is $2$-to-$1$. Similarly, if $\lambda\in\mathbf{nqr}(q)$, then $\lambda x^2\in\mathbf{nqr}(q)$ for any $x\in\F^\times_q$ and $\lambda x^2$ goes through $\mathbf{nqr}(q)$ twice when $x$ goes through $\F_q^\times$. Thus,
 \begin{equation*}
         \sum_{x\in \F_q}\psi(\lambda x^2)=
         \left\{
         \begin{aligned}
             &1+2\sum_{x\in\mathbf{qr}(q)}\psi(x) &  & \text{if}~\lambda\in\mathbf{qr}(q)&\\
             &1+2\sum_{x\in\mathbf{nqr}(q)}\psi(x) &  & \text{if}~\lambda\in\mathbf{nqr}(q)&
         \end{aligned}~.
         \right.    
    \end{equation*}
On the other hand,  we know that $1+2\sum_{x\in\mathbf{qr}(q)}\psi(x)=\sum_{x\in\F_q}\psi(x^2)$ and
\[
1+2\sum_{x\in\mathbf{qr}(q)}\psi(x)+ 1+2\sum_{x\in\mathbf{nqr}(q)}\psi(x)=2\sum_{x\in\F_q}\psi(x)=0.
\]
This establishes the desired equality.
\end{proof}
The expression of $\sum_{x\in \F_q}\psi(\lambda x^2)$ in Lemma~\ref{lem: sum of psi lambda x2} gives the following estimate:
 \begin{proposition}\label{prop: estimate of E_1}
     For any $\lambda\in\F_q^\times$, we have:
     \[|\sum_{x\in \F_q}\psi(\lambda x^2)|=\sqrt{q}\]
     Moreover, $|\mathbb{E}[\psi(\lambda a^2)]|=\frac{1}{q}|\sum_{x\in \F_q}\psi(\lambda x^2)|=q^{-1/2}$.
 \end{proposition}
\begin{proof}
    By Lemma~\ref{lem: sum of psi lambda x2}, we can focus on $\lambda\in\mathbf{qr}(q)$. Note that
    \[|\sum_{x\in \F_q}\psi(\lambda x^2)|^2=(\sum_{x\in \F_q}\psi(x^2))(\sum_{y\in \F_q}\overline{\psi(y^2)})=\sum_{x,y\in\F_q}\psi(x^2-y^2).\]
    Substituting $y=x+u$, we have:
    \[|\sum_{x\in \F_q}\psi(\lambda x^2)|^2=\sum_{x,y\in\F_q}\psi(x^2-y^2)=\sum_{x,u\in\F_q}\psi(-2xu-u^2)=\sum_{u\in\F_q}\psi(-u^2)(\sum_{x\in\F_q}\psi(-2xu)).\]
    Note that for any $u\in\F_q^\times$, $\sum_{x\in\F_q}\psi(-2xu)=0$. Then, we have
    \[
    |\sum_{x\in \F_q}\psi(\lambda x^2)|^2=\sum_{u\in\F_q}\psi(-u^2)(\sum_{x\in\F_q}\psi(-2xu))=\psi(0)(\sum_{x\in\F_q}\psi(0))=q.
    \]
    Thus, $|\sum_{x\in \F_q}\psi(\lambda x^2)|=\sqrt{q}$ and $|\mathbb{E}[\psi(\lambda a^2)]|=q^{-1/2}$ for any $\lambda\in\mathbf{qr}(q)$. The proof for $\lambda\in\mathbf{nqr}(q)$ is identical, which concludes the proof.
\end{proof}

Combining Eq.~(\ref{eq: estimate of E2}) and Proposition~\ref{prop: estimate of E_1}, we have:
\[|\mathbb{E}[\psi(\lambda\cdot \Tr(A^2))]|=|\mathbb{E}[\psi(\lambda a^2)]|^{n}\cdot|\mathbb{E}[\psi(2 \lambda xy)]|^{n(n-1)/2}=q^{-n/2}q^{-n(n-1)/2}=q^{-n^2/2}.\]
Therefore, we have
\[\left|\sum_{\lambda\in \F_q^\times}\mathbb{E}[\psi(\lambda\cdot \Tr(A^2))]\right|\le(q-1)q^{-n^2/2}.\]
Together with Eq.~(\ref{eq: probability of self-dual}) and Eq.~(\ref{eq: expectation value when lambda=0}), we obtained the desired estimate in Theorem~\ref{thm:random matrix being self-dual} for odd $q$.
\end{proof}
\subsection{Random matrices with both eigenvalue and trace conditions}\label{subsec:main-technical}\label{subsec: random main}

Now, we estimate the probability of a random self-dual matrix admitting a unique eigenvalue with algebraic multiplicity $1$.
\begin{theorem}\label{thm:main-technical}
    Let $\gamma(n,q)$ be the probability that a uniform matrix $A\in\M(n,q)$ conditioned on $\Tr(A^2)=0$ has exactly one eigenvalue in $\F_q$, of algebraic multiplicity $1$. Then we have
    \[
        \gamma(\infty,q):=\lim_{n\to\infty}\gamma(n,q)=\frac{q\cdot c(q)^q}{q-1}>e^{-23/9}.
    \]
    If $\gamma^*(n,q)$ denotes the same conditional probability with the additional requirement that the unique eigenvalue be nonzero, then
    \begin{equation}\label{eq:selfdual-nonzero-limit}
        \lim_{n\to\infty}\gamma^*(n,q)=c(q)^q\geq e^{-23/9}.
    \end{equation}
    Moreover, for $n-1>5(q-1)^2$, we have
    \[
         |\gamma(\infty,q)-\gamma(n,q)|\leq \frac{q^2}{1-(q-1)q^{-\frac{n^2}{2}}}\left((q-1)q^{-\frac{n^2}{2}-1}+\frac{q}{q-1} (q^{-(n+1)}+\Gamma(n-1,q))\right),
    \]
    where 
    \[
        \Gamma(n,q):=\frac{1}{2}16^{q-1}q^{-\frac{(n+1)(n+q)}{2(q-1)}-1}+(n+1)^2\binom{n+q-2}{q-2}(\frac{q}{q-1})^n\cdot q^{-\frac{n^2}{2(q+1)}+\frac{(q-1)n}{2(q+1)}+\frac{q-1}{4(q+1)}}.
    \]

\end{theorem}
\begin{proof}
We shall use Bayes' rule to express
\begin{equation}\label{eq:bayes}
\gamma(n,q)=\frac{\Pr[A~\text{has a unique eigenvalue in}~\F_q~\text{of algebraic multiplicity}~1~\text{and}~\Tr(A^2)=0]}{\Pr[\Tr(A^2)=0|A\in_R \M(n,q)]}.
\end{equation}
The denominator term in Eq.~(\ref{eq:bayes}) is exactly $\sigma(n,q)$, where we resort to Theorem~\ref{thm:random matrix being self-dual}. For the numerator in Eq.~(\ref{eq:bayes}), let 
\[
    D_1(n,q):=\left \{ A\in \M(n,q)| c_A(t)=f_0(t)(t-\lambda)~\text{for some}~\lambda \in \F_q,~f_0(t)\neq 0~\forall\ t\in\mathbb F_q,~\Tr(A^2)=0\right \}
\]
be the set of matrices $A\in\M(n,q)$ that admit exactly $1$ eigenvalue of algebraic multiplicity $1$ and satisfy $\Tr(A^2)=0$. We introduce the following method to count $|D_1(n,q)|$: Recall that a matrix $A\in \M(n,q)$ admitting a unique eigenvalue of algebraic multiplicity $1$ is conjugate to 
\[
    TAT^{-1}=\begin{bmatrix}\lambda & 0\\0 & A_0\end{bmatrix}
\]
through some invertible matrix $T\in\GL(n,q)$, where the matrix $A_0$ is eigenvalue-free in $\F_q$. This implies that $A_0\in\GL(n-1,q)$ (otherwise, $A_0$ has $0$ as its eigenvalue). Moreover, if $A$ is also self-dual, we have
\[
    0=\Tr(A^2)=\Tr(TAT^{-1}TAT^{-1})=\lambda^2+\Tr(A_0^2).
\]
Let $\beta(n-1,q,k)$ be the probability of a random matrix $A_0\in\GL(n-1,q)$ that satisfies (1) $A_0$ is eigenvalue-free in $\F_q$ and (2) $\Tr(A_0^2)=k$. More precisely,
\[
    \beta(n-1,q,k)=\Pr[c_{A_0}(t)\neq 0~\forall\ t\in\F_q,~\Tr(A_0^2)=k|~A_0\in_R\GL(n-1,q)].
\]
Then, we can compute $|D_1(n,q)|$ as 
\[
    |D_1(n,q)|=\frac{|\GL(n,q)|}{|\GL(n-1,q)|\cdot |\GL(1,q)|}\left(\sum_{\lambda\in\F_q}\beta(n-1,q,-\lambda^2)\cdot|\GL(n-1,q)|\right),
\]
where the fraction term counts the number of ordered direct-sum decompositions $\F_q^n=V_0\oplus V_\lambda$, with $\dim(V_0)=n-1$ and $\dim(V_\lambda)=1$, and the summation term counts the number of eigenvalue-free matrices $A_0\in \M(n-1,q)$ satisfying $\Tr(A_0^2)+\lambda^2=0$. Thus, the probability of a random matrix $A\in\M(n,q)$ that has a unique eigenvalue in $\F_q$ of algebraic multiplicity $1$ and is self-dual is 
\begin{equation}\label{eq: expression of delta}
    \delta(n,q):=\frac{|D_1(n,q)|}{|\M(n,q)|}=\frac{c_n(q)}{q-1}\sum_{k\in\F_q}\beta(n-1,q,-k^2).
\end{equation}

To estimate $\gamma(\infty,q)$, we obtain the following estimates for $\beta(n-1,q,k)$:
\begin{theorem}\label{thm: probability of a random eigenvalue-free matrix with trace conditions}
    For any $k\in\F_q$, let $\beta(n,q,k)$ be the probability of a random matrix $A\in\GL(n,q)$ that satisfies that (1) $A$ is eigenvalue-free in $\F_q$ and (2) $\Tr(A^2)=k$. Then the limit $\beta(\infty,q,k)=\lim_{n\to\infty}\beta(n,q,k)$ exists for any $k\in\F_q$ and 
    \[
    \beta(\infty,q,k)=\frac{1}{q}v(\GL,\infty,q)=\frac{1}{q}c(q)^{q-1}.
    \] 
    Moreover, for any $n>5(q-1)^2$ and any $k\in\F_q$, we have 
    \begin{equation}\label{eq:~convergence bound beta}
    |\beta(\infty,q,k)-\beta(n,q,k)|\leq \Gamma(n,q),
    \end{equation} 
    where
    \[
    \Gamma(n,q):=\frac{1}{2}16^{q-1}q^{-\frac{(n+1)(n+q)}{2(q-1)}-1}+(n+1)^2\binom{n+q-2}{q-2}(\frac{q}{q-1})^n\cdot q^{-\frac{n^2}{2(q+1)}+\frac{(q-1)n}{2(q+1)}+\frac{q-1}{4(q+1)}}.
    \]
\end{theorem}
We shall prove Theorem~\ref{thm: probability of a random eigenvalue-free matrix with trace conditions} in Section~\ref{subsubsec:trace-conditioned-eigenfree-proof}. 

We continue the proof of Theorem~\ref{thm:main-technical}. For fixed $q$, Theorem~\ref{thm: probability of a random eigenvalue-free matrix with trace conditions} gives the same limit $\beta(\infty,q,k)=q^{-1}c(q)^{q-1}$ for every $k\in\F_q$. Taking limits in the finite sum in Eq.~(\ref{eq: expression of delta}), while retaining the arguments $-x^2$, gives
\[
    \delta(\infty,q):=\lim_{n\to\infty}\delta(n,q)
    =\frac{c(q)}{q-1}\sum_{x\in\F_q}\beta(\infty,q,-x^2)
    =\frac{c(q)^q}{q-1}.
\]
Together with $\sigma(\infty,q)=1/q$ from Theorem~\ref{thm:random matrix being self-dual}, Eq.~(\ref{eq:bayes}) yields
\[
\gamma(\infty,q)=\frac{\delta(\infty,q)}{\sigma(\infty,q)}
    =\frac{q\cdot c(q)^q}{q-1}>e^{-23/9}.
\]
For the nonzero-eigenvalue assertion, let $D_1^*(n,q)$ be the subset of $D_1(n,q)$ whose unique eigenvalue is nonzero, and put $\delta^*(n,q)=|D_1^*(n,q)|/q^{n^2}$. Restricting the sum in Eq.~(\ref{eq: expression of delta}) gives
\begin{equation}\label{eq:selfdual-nonzero-density}
    \delta^*(n,q)=\frac{c_n(q)}{q-1}
        \sum_{\lambda\in\F_q^\times}\beta(n-1,q,-\lambda^2)
        \longrightarrow\frac{c(q)^q}{q}.
\end{equation}
Dividing by $\sigma(n,q)\to1/q$ proves Eq.~(\ref{eq:selfdual-nonzero-limit}). In particular, for each fixed $q$ and all sufficiently large $n$,
\begin{equation}\label{eq:selfdual-nonzero-lower}
    \delta^*(n,q)\geq\frac{e^{-23/9}}{2q},\qquad
    \gamma^*(n,q)\geq\frac12e^{-23/9}.
\end{equation}
For the error rate, note that
\begin{equation}\label{eq: error of gamma}
\begin{split}
    |\gamma(\infty,q)-\gamma(n,q)|&=\left|\frac{\delta(\infty,q)}{\sigma(\infty,q)}-\frac{\delta(n,q)}{\sigma(n,q)}\right|\\
    &=\left|\frac{\delta(\infty,q)\sigma(n,q)-\delta(n,q)\sigma(\infty,q)}{\sigma(\infty,q)\sigma(n,q)}\right|\\
    &\leq\frac{q}{\sigma(n,q)}(|\sigma(\infty,q)-\sigma(n,q)|+|\delta(\infty,q)-\delta(n,q)|),
\end{split}
\end{equation}
where the last inequality uses the triangle inequality, $\sigma(\infty,q)=1/q\leq1$ and $\delta(\infty,q)\leq1$. Put $b_q=q^{-1}c(q)^{q-1}$. For $n-1>5(q-1)^2$, Eq.~(\ref{eq: expression of delta}) gives
\[
\begin{split}
|\delta(\infty,q)-\delta(n,q)|
&=\frac{1}{q-1}\left|\sum_{x\in\F_q}
    \left(c(q)b_q-c_n(q)\beta(n-1,q,-x^2)\right)\right|\\
&\leq\frac{q b_q}{q-1}|c(q)-c_n(q)|\\
&\quad+\frac{c_n(q)}{q-1}\sum_{x\in\F_q}
    |b_q-\beta(n-1,q,-x^2)|\\
&\leq\frac{q}{q-1}\left(q^{-(n+1)}+\Gamma(n-1,q)\right).
\end{split}
\]
For the last step, the beta error bound is uniform in its target. We also used
\[
|c_n(q)-c(q)|
\leq c_n(q)\sum_{i=n+1}^{\infty}q^{-i}
\leq (1-q^{-1})\frac{q^{-(n+1)}}{1-q^{-1}}
=q^{-(n+1)},
\]
and the bounds $\beta(\infty,q,k),c_n(q)\leq1$.

Together with the convergence bound of $\sigma(\infty,q)$ (cf.~Eq.~(\ref{eq: convergence bound sigma})), the convergence bound of $\gamma(\infty,q)$ can be estimated through Eq.~(\ref{eq: error of gamma}) as
\[
    |\gamma(\infty,q)-\gamma(n,q)|\leq \frac{q^2}{1-(q-1)q^{-\frac{n^2}{2}}}\left((q-1)q^{-\frac{n^2}{2}-1}+\frac{q}{q-1} (q^{-(n+1)}+\Gamma(n-1,q))\right).\qedhere
\]
\end{proof}



\subsubsection{Proof of Theorem~\ref{thm: probability of a random eigenvalue-free matrix with trace conditions}}\label{subsubsec:trace-conditioned-eigenfree-proof}

We are left to prove Theorem~\ref{thm: probability of a random eigenvalue-free matrix with trace conditions}. Note that such a probability when there is no restriction on the eigenvalues was studied by Gorodetsky and Rodgers~\cite{MR4273172}. Roughly speaking, Gorodetsky and Rodgers use the character sum to express the indicator function $\mathbb{I}\{\Tr(A^2)=k\}$ for $A\in\GL(n, q)$ (cf.~Proposition~\ref{lem: character sum expression for self-dual}), which translates sums over invertible matrices to sums over characteristic polynomials. Then, the number of matrices with fixed trace powers can be estimated using the L-functions of the characters.

In our setting, the goal is to bound $\beta(n,q,k)$, the probability that a uniform matrix $A\in\GL(n,q)$ is eigenvalue-free in $\F_q$ and satisfies $\Tr(A^2)=k$. 
\begin{definition}
    Let
\[
  \cN(n,q,k)=\{A\in \GL(n,q):~A~\text{is eigenvalue-free and}~\Tr(A^2)=k\}
\] 
and let 
\[
\cE(n,q)=\{f\in\cM_{n,q}^{gl}:~f(t)\neq 0~\forall~t\in\F_q\}
\] 
be the set of characteristic polynomials of invertible matrices that are eigenvalue-free.
\end{definition}
Note that
\begin{equation}
\begin{split}
    |\cN(n,q,k)|&=\sum_{A\in\GL(n,q)}\mathbb{I}\{c_A(t)\in\cE(n,q)\}\cdot\mathbb{I}\{\Tr(A^2)=k\}\\
    &=\sum_{A\in\GL(n,q)}\mathbb{I}\{c_A(t)\in\cE(n,q)\}\cdot\left(\frac{1}{q}\sum_{\lambda\in\F_q}\chi_\lambda(c_A(t))\overline{\psi_\lambda(k)}\right).
\end{split}
\end{equation}
This suggests working with the characteristic polynomials of eigenvalue-free matrices instead of invertible matrices. In particular, if we replace $\cE(n,q)$ by $\cM_{n,q}^{gl}$, we arrive at the setting in~\cite{MR4273172}.
Reiner~\cite{10.1215/ijm/1255629830} and Gerstenhaber~\cite{10.1215/ijm/1255629831} derived explicit formulae to compute the probability of a random invertible matrix whose characteristic polynomial is exactly some fixed monic polynomial $f\in\cM_q^{gl}$, which is denoted as $P_{\GL}(f)$:
\[
    P_{\GL}(f)=\Pr_{A\in\GL(n,q)}[c_A(t)=f(t)]=\frac{|\{A\in\GL(n,q):~c_A(t)=f(t)\}|}{|\GL(n,q)|}.
\]
 
\begin{lemma}[Theorem 2 in~\cite{10.1215/ijm/1255629830}, Section 2 in~\cite{10.1215/ijm/1255629831}]\label{lem:P_GL}
    If $f\in\cM_q^{gl}$ factorises as $f=g_1^{e_1}\cdots g_r^{e_r}$ with each $g_i$ an irreducible monic polynomial, then
    \[
        P_{\GL}(f)=\prod_{i=1}^r\frac{q^{\deg(g_i)e_i(e_i-1)}}{|\GL(e_i,q^{\deg(g_i)})|}.
    \]
\end{lemma}
Then we can further rewrite $|\cN(n,q,k)|$ as
\begin{equation}\label{eq: expression of Nnqk}
    |\cN(n,q,k)|=\frac{|\GL(n,q)|}{q}\sum_{\lambda\in\F_q}\sum_{f\in\cE(n,q)}P_{\GL}(f)\chi_\lambda(f)\overline{\psi_\lambda(k)},
\end{equation}
which translates the sum over invertible matrices into the sum over characteristic polynomials. 

To estimate Eq.~(\ref{eq: expression of Nnqk}), we first introduce some known results for the L-functions of characters introduced in~\cite{MR4273172}, then we exhibit how to adapt the proof template in~\cite{MR4273172} to estimate $\beta(n,q,k)$.

\paragraph{The L-function of the character $\chi_\lambda$.} 
Recall that we defined $\chi_\lambda$ in Eq.~(\ref{eq:chi}). The L-function over $\GL$ for $\chi_\lambda$ is defined as
\[
L_{\GL}(u,\chi_\lambda)=\sum_{n=0}^\infty\sum_{f\in\cM_{n,q}^{gl}}\chi_\lambda(f)u^{n}.
\]
\begin{lemma}[{\cite[Theorem 3.5]{MR4273172}}]\label{lem:L functions}
    For any $\lambda\in\F_q$ and $|u|<1/q$, we have
    \begin{equation}\label{eq: euler product of LGL}
        L_{\GL}(u,\chi_\lambda)=\sum_{n=0}^\infty\sum_{f\in\cM_{n,q}^{gl}}\chi_\lambda(f)u^{n}=\prod_{P(t)\neq t}\frac{1}{1-\chi_\lambda(P)u^{\deg(P)}},
    \end{equation}
    where the second equality is due to the Euler product, and the product is over all \emph{irreducible} monic polynomials in $\cM_q^{gl}$ (Thus $P(t)\neq t$). 
    This sum converges absolutely for $|u|<1/q$, and for $|u|<1$ we have
    \[
        \sum_{n=0}^\infty\sum_{f\in\cM^{gl}_{n,q}} P_{\GL}(f)\chi_\lambda(f)u^{n}=\prod_{i=1}^\infty L_{\GL}(\frac{u}{q^i},\chi_\lambda),
    \]
    with both the left-hand sum and the right-hand product converging absolutely.
\end{lemma}

Inspired by the above, to estimate Eq.~(\ref{eq: expression of Nnqk}), we define the L-function over eigenvalue-free matrices for $\chi_\lambda$ as 
\begin{equation}\label{eq: L function of E_n}
    L_{\cE}(u,\chi_\lambda)=\sum_{n=0}^\infty\sum_{f\in\cE(n,q)}\chi_\lambda(f)u^{n},
\end{equation}
for any $\lambda\in\F_q$. Using the Euler product, we have the following:
\begin{equation}\label{eq: L function of E_n Euler product}
    L_{\cE}(u,\chi_\lambda)=\prod_{\substack{P\text{ monic irreducible}\\\deg P\geq2}}\frac{1}{1-\chi_\lambda(P)u^{\deg P}}.
\end{equation}
By Lemma~\ref{lem:P_GL}, $P_{\GL}$ is multiplicative on $\cM_q^{gl}$. For a monic irreducible polynomial $g\neq t$ and an integer $e\geq1$, put $G=q^{\deg g}$. Then
\[
    P_{\GL}(g^e)=\frac{G^{e(e-1)}}{(G^e-1)(G^e-G)\cdots(G^e-G^{e-1})}.
\]
Thus, for $|u|<1$, we can compute
\[
    \begin{split}
        &\sum_{n=0}^\infty\sum_{f\in\cE(n,q)}P_{\GL}(f)\chi_\lambda(f)u^n\\
       =&\prod_{\substack{g\text{ monic irreducible}\\\deg g\geq2}}\left(1+\sum_{e=1}^\infty\frac{G^{e(e-1)}}{(G^e-1)(G^e-G)\cdots(G^e-G^{e-1})}(\chi_\lambda(g)u^{\deg g})^e\right)\\
       =&\prod_{\substack{g\text{ monic irreducible}\\\deg g\geq2}}\prod_{i=1}^\infty\frac{1}{1-\chi_\lambda(g)(u/q^i)^{\deg g}}\\
       =&\prod_{i=1}^\infty L_{\cE}(\frac{u}{q^i},\chi_\lambda),
    \end{split}
\]
where the second equality is due to Euler (cf.~\cite[Theorem 2.4 and Eq.~(2.8)]{MR4273172}) and the order of the two infinite products can be interchanged since for $|u|<1$, the Euler product is absolutely convergent (cf.~\cite[Proof of Thm.~3.5]{MR4273172}). 
Utilising the second equality, we obtain the following:
\begin{proposition}\label{prop: L_GL and L_E}
    For any $0\neq \lambda\in\F_q$, we have 
    \begin{equation}\label{eq: characteristic sum over E_n and M_n,q}
        \sum_{n=0}^\infty\sum_{f\in\cE(n,q)} P_{\GL}(f)\chi_\lambda(f)u^{n}=\prod_{c\in\F_q^\times}\prod_{i=1}^\infty(1-\psi(\lambda c^2)\cdot \frac{u}{q^i})\sum_{n=0}^\infty\sum_{f\in\cM^{gl}_{n,q}} P_{\GL}(f)\chi_\lambda(f)u^{n}.
    \end{equation}
\end{proposition}
\begin{proof}
 Using product expressions in terms of $L_{\GL}$ and $L_{\cE}$ and comparing the Euler product expressions of $L_{\GL}$ in~Eq.~(\ref{eq: euler product of LGL}) and of $L_\cE$ in~Eq.~(\ref{eq: L function of E_n Euler product}), we have
    \[
        \sum_{n=0}^\infty\sum_{f\in\cE(n,q)} P_{\GL}(f)\chi_\lambda(f)u^{n}=\left(\prod_{c\in\F_q^\times}\prod_{i=1}^\infty(1-\chi_\lambda(t-c)\cdot \frac{u}{q^i})\right)\sum_{n=0}^\infty\sum_{f\in\cM^{gl}_{n,q}} P_{\GL}(f)\chi_\lambda(f)u^{n}.
    \]
    The result then follows since $\chi_\lambda(t-c)=\psi(\lambda c^2)$.
\end{proof}

These L-function products enable us to estimate coefficients of the generating functions. For example, Gorodetsky and Rodgers derived an estimate on $\sum_{f\in\cM^{gl}_{n,q}}P_{\GL}(f)\chi_\lambda(f)$~\cite{MR4273172}:
\begin{lemma}[{\cite[Theorem 3.6 (2)]{MR4273172}}]\label{lem: estimation of average character}
    For any $\lambda\in\F_q^\times$, we have
    \begin{equation}\label{eq:GR3.6}
    \left |\sum_{f\in \cM^{gl}_{n,q}}P_{\GL}(f)\chi_\lambda(f)\right |=\frac{1}{|\GL(n,q)|}\left |\sum_{A\in\GL(n,q)}\chi_\lambda(c_{A}(t))\right |\le (n+1)\cdot q^{-\frac{n^2}{4}}\cdot (\frac{q}{q-1})^n.
    \end{equation}
\end{lemma}
\begin{proof}
Write a monic polynomial as $f(t)=t^j+a_1t^{j-1}+a_2t^{j-2}+\cdots$, with absent coefficients taken to be zero. Newton's identity gives $\sum_i\alpha_i^2=a_1^2-2a_2$, so
\[
\chi_\lambda(f)=\psi\bigl(\lambda(a_1^2-2a_2)\bigr).
\]
The function $\chi_\lambda$ is multiplicative, since the roots of a product are the roots of its factors with multiplicities, and it depends only on the first two coefficients below the leading coefficient. Thus it is a Hayes character of at most two leading coefficients in the terminology of~\cite[Section~2.1]{MR4273172}. For odd characteristic it is nontrivial by varying $a_2$; in characteristic two it reduces to the nontrivial additive character $a_1\mapsto\psi(\lambda a_1^2)$, since squaring permutes $\F_q$. Theorem~3.6 of~\cite{MR4273172} therefore gives $d:=\deg L_{\GL}(u,\chi_\lambda)\leq2$. For $d>0$, its estimate~(3.9) is at most the bound in Eq.~(\ref{eq:GR3.6}); if $d=0$, the positive-degree coefficients of the product in Lemma~\ref{lem:L functions} vanish, so the same bound holds.
\end{proof}

Here is the key difference between our setting and the one studied in \cite{MR4273172}. Comparing Equations~\ref{eq:GR3.6} and~\ref{eq: characteristic sum over E_n and M_n,q}, we need to estimate the overhead term in Eq.~(\ref{eq: characteristic sum over E_n and M_n,q}), denoted by
\[
H(u,\lambda)=\prod_{c\in\F_q^\times}\prod_{i=1}^\infty(1-\psi(\lambda c^2)\cdot \frac{u}{q^i})=\prod_{c\in\F_q^\times}\underbrace{\left(1+\sum_{n=1}^\infty\frac{(-\psi(\lambda c^2))^n}{(q^n-1)(q^{n-1}-1)\cdots(q-1)}u^n\right)}_{:=H(u,\lambda,c)}~,
\]
where the product-to-sum conversion is again due to Euler (cf.~\cite[Theorem 2.5 and Eq.~(2.10)]{MR4273172}). By estimating the coefficients of every $H(u,\lambda,c)$, we obtain the following:

\begin{proposition}\label{prop: estimates of overhead serie}
    Suppose $H(u,\lambda)=\sum_{n=0}^\infty h_nu^n$. Then
    \[|h_n|\le \binom{n+q-2}{q-2}(\frac{q}{q-1})^{n}q^{-\frac{n^2}{2(q-1)}-\frac{n}{2}}.
    \]
\end{proposition}

\begin{proof}
Let $\mathbf n_{q-1}$ be the set of weak compositions $(n_c)_{c\in\F_q^\times}$ of $n$ into $q-1$ parts. Then
\[
    h_n=\sum_{\vec n\in\mathbf n_{q-1}}
        \prod_{c\in\F_q^\times}[u^{n_c}]H(u,\lambda,c).
\]
For every $r\geq0$, Euler's product expansion gives
\[
    [u^r]H(u,\lambda,c)
    =\frac{(-\psi(\lambda c^2))^r}{\prod_{j=1}^r(q^j-1)},
\]
with the empty product equal to $1$. Since $|\psi(\lambda c^2)|=1$ and $q^j-1\geq q^{j-1}(q-1)$,
\[
    |[u^r]H(u,\lambda,c)|
    \leq(q-1)^{-r}q^{-r(r-1)/2}
    =\left(\frac q{q-1}\right)^r q^{-(r^2+r)/2}.
\]
Consequently, using the triangle inequality and $\sum_c n_c^2\geq n^2/(q-1)$,
\[
\begin{split}
    |h_n|
    &\leq\sum_{\vec n\in\mathbf n_{q-1}}
        \left(\frac q{q-1}\right)^n q^{-\frac12\sum_c(n_c^2+n_c)}\\
    &\leq\binom{n+q-2}{q-2}
        \left(\frac q{q-1}\right)^n q^{-\frac{n^2}{2(q-1)}-\frac n2},
\end{split}
\]
where $|\mathbf n_{q-1}|=\binom{n+q-2}{q-2}$. This proves the proposition.
\end{proof}

Now we are ready to prove Theorem~\ref{thm: probability of a random eigenvalue-free matrix with trace conditions}.
\medskip

\begin{proof}[Proof of Theorem~\ref{thm: probability of a random eigenvalue-free matrix with trace conditions}]
By Eq.~(\ref{eq: expression of Nnqk}), for every $k\in\F_q$,
\[
    \beta(n,q,k)=\frac1q\sum_{\lambda\in\F_q}
        \sum_{f\in\cE(n,q)}P_{\GL}(f)\chi_\lambda(f)\overline{\psi_\lambda(k)}.
\]
The term with $\lambda=0$ is $q^{-1}v(\GL,n,q)$. Thus
\begin{equation}\label{eq: error decomposition}
    P_{err}(n,k):=\beta(n,q,k)-\frac1qv(\GL,n,q)
    =\frac1q\sum_{\lambda\in\F_q^\times}
        \overline{\psi_\lambda(k)}
        \sum_{f\in\cE(n,q)}P_{\GL}(f)\chi_\lambda(f).
\end{equation}
For fixed $\lambda\neq0$, write
\[
    a_j(\lambda):=\sum_{f\in\cM_{j,q}^{gl}}P_{\GL}(f)\chi_\lambda(f),
    \qquad H(u,\lambda)=\sum_{r\geq0}h_r(\lambda)u^r.
\]
Proposition~\ref{prop: L_GL and L_E} and $|\psi_\lambda(k)|=1$ give
\[
    |P_{err}(n,k)|
    \leq\frac{q-1}{q}\max_{\lambda\in\F_q^\times}
        \left|\sum_{j=0}^n h_{n-j}(\lambda)a_j(\lambda)\right|.
\]
We use the weaker version of Proposition~\ref{prop: estimates of overhead serie} obtained by replacing $-r/2$ in its exponent by $+r/2$. Together with Lemma~\ref{lem: estimation of average character}, this yields
\begin{align*}
    \left|\sum_{j=0}^n h_{n-j}(\lambda)a_j(\lambda)\right|
    &\leq\sum_{j=0}^n
        \binom{n-j+q-2}{q-2}
        \left(\frac q{q-1}\right)^{n-j}
        q^{-\frac{(n-j)^2}{2(q-1)}+\frac{n-j}{2}}
        (j+1)\left(\frac q{q-1}\right)^j q^{-j^2/4}\\
    &\leq(n+1)\binom{n+q-2}{q-2}
        \left(\frac q{q-1}\right)^n
        \sum_{j=0}^n q^{-\frac{(n-j)^2}{2(q-1)}+\frac{n-j}{2}-\frac{j^2}{4}}\\
    &\leq(n+1)^2\binom{n+q-2}{q-2}
        \left(\frac q{q-1}\right)^n
        q^{-\frac{n^2}{2(q+1)}+\frac{(q-1)n}{2(q+1)}+\frac{q-1}{4(q+1)}}.
\end{align*}
Here the estimate for $a_j(\lambda)$ also holds at $j=0$, where $a_0(\lambda)=1$. The second inequality uses $j+1\leq n+1$ and the monotonicity of the binomial coefficient in its top argument. For the last inequality, completing the square gives
\[
\begin{split}
    -\frac{(n-j)^2}{2(q-1)}+\frac{n-j}{2}-\frac{j^2}{4}
    &=-\frac{q+1}{4(q-1)}
        \left(j-\frac{2n-q+1}{q+1}\right)^2\\
    &\quad-\frac{n^2}{2(q+1)}+\frac{(q-1)n}{2(q+1)}+\frac{q-1}{4(q+1)}.
\end{split}
\]
We have therefore proved, uniformly in $k\in\F_q$,
\begin{equation}\label{eq:~bounds on perr}
    |P_{err}(n,k)|\leq(n+1)^2\binom{n+q-2}{q-2}
        \left(\frac q{q-1}\right)^n
        q^{-\frac{n^2}{2(q+1)}+\frac{(q-1)n}{2(q+1)}+\frac{q-1}{4(q+1)}}.
\end{equation}
For fixed $q$, this bound tends to zero as $n\to\infty$. Hence
\[
    \beta(\infty,q,k)=\frac1qv(\GL,\infty,q)=\frac1qc(q)^{q-1}.
\]
Finally,
\[
    |\beta(\infty,q,k)-\beta(n,q,k)|
    \leq\frac1q|v(\GL,\infty,q)-v(\GL,n,q)|+|P_{err}(n,k)|.
\]
For $n>5(q-1)^2$, substituting Proposition~\ref{lem: gen functions of v and u} and Eq.~(\ref{eq:~bounds on perr}) gives $\Gamma(n,q)$ as stated.
\end{proof}

\subsection{Random Gram matrices and orthogonal centralizers}\label{subsec:orth-random-matrices}
We prove Propositions~\ref{prop:orth-gram-uniform} and~\ref{prop:orth-centralizer-moment}, stated in Section~\ref{subsec:orth-average}. Throughout this subsection $q$ is odd and $n\geq2$. The first proof uses the additive character and quadratic character-sum bound already developed in Section~\ref{subsec: random self-dual}; the second follows the primary-decomposition and generating-function method of Section~\ref{subsec:eigen-condition}.

\subsubsection{Proof of Proposition~\ref{prop:orth-gram-uniform}}\label{subsubsec:orth-gram-proof}
\begin{proof}[Proof of Proposition~\ref{prop:orth-gram-uniform}]
Let $\psi$ be the character in Eq.~(\ref{eq: additive character}). We use Lemma~\ref{lem:orthogonal} for character orthogonality. Because $q$ is odd, the pairing
\[
 (T,S)\longmapsto\Tr(TS)=\sum_iT_{ii}S_{ii}+2\sum_{i<j}T_{ij}S_{ij}
\]
is nondegenerate on $\sym(n,q)$. Consequently, the functions $\chi_T(S)=\psi(\Tr(TS))$, indexed by $T\in\sym(n,q)$, are its distinct additive characters. For a probability distribution $\rho$ on this space, put $\widehat\rho(T)=\sum_S\rho(S)\chi_T(S)$.

Write $M=[v_1\ \cdots\ v_m]$ and let $\nu$ be the distribution of $v\trp{v}$ for uniform $v\in\F_q^n$. Since $M\trp{M}=\sum_{j=1}^m v_j\trp{v_j}$ and the columns are independent,
\[
 \widehat\nu(T)=q^{-n}\sum_{v\in\F_q^n}\psi(\trp{v}Tv),
 \qquad
 \widehat\mu_{n,m,q}(T)=\widehat\nu(T)^m.
\]
If $\rk(T)=r$, diagonalize $T$ by congruence and apply Proposition~\ref{prop: estimate of E_1} to each of its $r$ nonzero diagonal coefficients. This uses exactly the quadratic character-sum estimate employed in the proof of Theorem~\ref{thm:random matrix being self-dual}, and gives
\[
 |\widehat\nu(T)|=q^{-r/2},
 \qquad |\widehat\mu_{n,m,q}(T)|^2=q^{-mr}.
\]
For completeness, character orthogonality gives, for every real function $f$ on $\sym(n,q)$,
\[
 \sum_{T\in\sym(n,q)}\left|\sum_S f(S)\chi_T(S)\right|^2
 =q^{n(n+1)/2}\sum_{S\in\sym(n,q)}|f(S)|^2.
\]
Apply this identity to $f=\mu_{n,m,q}-U_{n,q}$; the trivial-character term is zero. Grouping the remaining terms by rank proves Eq.~(\ref{eq:orth-gram-fourier}). We then apply Eq.~(\ref{eq:orth-tv-chi-square}) to bound the total variation distance.

To bound $R_{n,r}$, first choose the $(n-r)$-dimensional radical of the associated symmetric bilinear form. For any such radical there are at most $q^{r(r+1)/2}$ forms on the quotient. The number of choices of the radical is the Gaussian binomial coefficient
\[
 \genfrac{[}{]}{0pt}{}{n}{r}_q
 =q^{r(n-r)}\prod_{i=0}^{r-1}\frac{1-q^{-(n-i)}}{1-q^{-(r-i)}}
 \leq \frac{q^{r(n-r)}}{c(q)},
\]
where the denominator is $c_r(q)\geq c(q)$ by Eq.~(\ref{eq: ratio of GLnq and Mnq}). Thus
\[
 R_{n,r}\leq \genfrac{[}{]}{0pt}{}{n}{r}_q q^{r(r+1)/2}
 \leq c(q)^{-1}q^{nr-r(r-1)/2}.
\]
Substitution proves the first inequality in Eq.~(\ref{eq:orth-gram-tv}). The second follows by dropping $q^{-r(r-1)/2}$ and summing a geometric series. Finally, for odd $q$,
\[
 c(q)\geq1-\sum_{j\geq1}q^{-j}\geq\frac12,
 \qquad 1-q^{-n(n-1)}\geq\frac89
\]
when $n\geq2$. Taking $m=n^2$ gives the asserted bound on $\delta_{n,q}$.
\end{proof}

\subsubsection{Proof of Proposition~\ref{prop:orth-centralizer-moment}}\label{subsubsec:orth-centralizer-proof}
The module decomposition in Section~\ref{subsec:matrix-prelim} also provides the starting point for counting orthogonal conjugacy classes. Here we need its refinement for self-adjoint operators on a symmetric bilinear space.

\begin{lemma}\label{lem:orth-cyclic-types}
Let $q$ be odd and $S\in\sym(n,q)$. There is an invertible matrix $P$ such that
\begin{equation}\label{eq:orth-cyclic-decomposition}
 P^{-1}SP=\operatorname{diag}_{\phi,e,j}C(\phi^e),\qquad
 \trp{P}P=\operatorname{diag}_{\phi,e,j}H_{\phi,e,j},
\end{equation}
where $\phi$ ranges over monic irreducible polynomials, $e\geq1$, $C(\phi^e)$ is the companion matrix of $\phi^e$, and each $H_{\phi,e,j}$ is symmetric and nonsingular with
\[
 \trp{C(\phi^e)}H_{\phi,e,j}=H_{\phi,e,j}C(\phi^e).
\]
For each fixed $(\phi,e)$, the admissible matrices $H$ form exactly two equivalence classes under $H\mapsto\trp{D}HD$, where $D$ is invertible and commutes with $C(\phi^e)$.
\end{lemma}
\begin{proof}
The orthogonal cyclic decomposition for self-adjoint operators~\cite[Theorems 6 and 8]{Ser88}, applied to the pair $(S,\idmat_n)$, refines the primary and cyclic module decomposition of Section~\ref{subsec:matrix-prelim} to give Eq.~(\ref{eq:orth-cyclic-decomposition}). Its cyclic block sizes are the elementary divisors of $S$.

We explain the two choices of form on a fixed cyclic block. Put $C=C(\phi^e)$ and fix one admissible matrix $H_0$. Every matrix commuting with $C$ belongs to $R=\F_q[C]\cong\F_q[t]/(\phi^e)$. For any other admissible $H$, the identity $\trp{C}H_0=H_0C$ implies that $h=H_0^{-1}H$ commutes with $C$. Thus $H=H_0h$ for some $h\in R^\times$. Every element of $R$ is self-adjoint relative to $H_0$, since $C$ is. Hence all $H_0h$, with $h\in R^\times$, are symmetric, and replacing $H$ by $\trp{D}HD$ replaces $h$ by $hD^2$.

The equivalence classes are therefore $R^\times/(R^\times)^2$. Reduction modulo $(\phi)$ maps this quotient isomorphically onto the square-class group of $\F_{q^{\deg\phi}}^\times$: the kernel on unit groups is $1+(\phi)/(\phi^e)$, a finite abelian group of odd order, on which squaring is bijective. The finite-field multiplicative group has exactly two square classes, as in the symmetric-form classification~\cite[Lemma 3.1(a)]{Ser10}. This proves the claim.
\end{proof}

\begin{proof}[Proof of Proposition~\ref{prop:orth-centralizer-moment}]
Let $\kappa_{n,q}$ be the number of orbits for the conjugation action of $\mathrm{O}(n,q)$ on $\sym(n,q)$. The stabilizer of $S$ is $C_{\mathrm O}(S)$. By orbit--stabilizer, the sum of the stabilizer orders over any one orbit is $|\mathrm{O}(n,q)|$, so
\begin{equation}\label{eq:orth-orbit-stabilizer}
 \mathbb{E}_S z_{\mathrm O}(S)=\frac{|\mathrm{O}(n,q)|\kappa_{n,q}}{q^{n(n+1)/2}}.
\end{equation}
We bound $\kappa_{n,q}$ using Lemma~\ref{lem:orth-cyclic-types}. Fix representatives $J_{\phi,e,0},J_{\phi,e,1}$ for its two cyclic-form classes. A block record consists of the multiplicities $m_{\phi,e,\tau}\geq0$, for $\tau\in\{0,1\}$, with
\[
 \sum_{\phi,e,\tau}e\deg(\phi)m_{\phi,e,\tau}=n.
\]
Every $S$ admits such a record. If $S,T$ admit the same record, there are bases $P_S,P_T$ with
\[
 P_S^{-1}SP_S=P_T^{-1}TP_T=C,\qquad
 \trp{P_S}P_S=\trp{P_T}P_T=J,
\]
where $C,J$ are the same ordered direct sums of the fixed block pairs. Then $Q=P_TP_S^{-1}$ is orthogonal and $QSQ^{-1}=T$. Thus each realizable record determines one orbit. Uniqueness of the record is unnecessary for this upper bound.

Let $a_{n,q}$ count all records, including those whose total Gram form is not equivalent to $\idmat_n$. As in the primary-decomposition generating-function method of Section~\ref{subsec:eigen-condition}, each block type contributes a geometric series, giving
\begin{equation}\label{eq:orth-orbit-generating-function}
 \kappa_{n,q}\leq a_{n,q},\qquad
 \sum_{n\geq0}a_{n,q}u^n
 =\prod_{\phi}\prod_{e\geq1}(1-u^{e\deg\phi})^{-2}
 =\prod_{e\geq1}(1-qu^e)^{-2}.
\end{equation}
The last identity follows from unique factorization of monic polynomials:
\[
 \prod_\phi(1-v^{\deg\phi})^{-1}
 =\sum_{d\geq0}q^dv^d=\frac{1}{1-qv}.
\]
Here all monic irreducible polynomials, including $t$, are included. Write
\[
 H_q(u)=\prod_{e\geq2}(1-qu^e)^{-2}=\sum_{j\geq0}h_ju^j,
 \qquad h_j\geq0.
\]
Extracting coefficients in $(1-qu)^{-2}H_q(u)$ yields
\begin{align*}
 q^{-n}a_{n,q}
 &=\sum_{j=0}^n(n-j+1)h_jq^{-j}\leq(n+1)H_q(q^{-1})
 =(n+1)\prod_{e\geq2}(1-q^{1-e})^{-2}
 =\frac{n+1}{c(q)^2}.
\end{align*}
The standard orthogonal-group order formulas~\cite[Chapter 11]{Tay92} imply
\[
 |\mathrm{O}(n,q)|\leq2(1+q^{-1})q^{n(n-1)/2}.
\]
Combine this with Eq.~(\ref{eq:orth-orbit-stabilizer}) and $\kappa_{n,q}\leq a_{n,q}$ to obtain the first inequality in Eq.~(\ref{eq:orth-centralizer-moment}). The second uses $c(q)\geq1/2$ and $q\geq3$.
\end{proof}

\paragraph{AI-use disclosure.} Most of the algorithmic results, including Theorem~\ref{thm:alg-iso} and~\ref{thm:4-tensor-iso}, and most of the random matrix theory results, including the results in Section~\ref{subsec:eigen-condition},~\ref{subsec: random self-dual} and ~\ref{subsec:main-technical}, are obtained by the authors. Most of the results about orthogonal tensor isomorphism over finite fields and its corresponding random matrix theory machinery are also developed by the authors. 

We used GPT-6-Astra to assist with the proof of Lemma~\ref{lem:orth-cyclic-types}, and used GPT-6-Sol to assist with the algorithm design of matrix code conjugacy, and the proofs of Propositions~\ref{prop:hull-dim},~\ref{prop:random-self-dual},~\ref{prop:hull-contraction} and~\ref{prop:matrix-code-projection}. The tool materially affected Sections~\ref{sec:matrix-sp-conj},~\ref{sec:orth-tensor-iso} and~\ref{subsec:orth-random-matrices}. We also used GPT-6-Sol to assist with polishing the whole paper. The authors verified the correctness and originality of all content, including references.

\bibliographystyle{alphaurl}
\bibliography{references}

\end{document}

%% file: preamble.tex
\usepackage[utf8]{inputenc}
\usepackage{latexsym}
\usepackage{amsmath}
\usepackage{amssymb}
\usepackage{amsthm}
\usepackage{cite}
\usepackage{graphicx}
\usepackage{color}
\usepackage{enumitem}
\usepackage{comment}
\usepackage{bm}
\usepackage[left=1in,right=1in,top=1in,bottom=1in]{geometry}
\usepackage{xspace}
\setlist[description]{font=\normalfont\itshape\textbullet\space}
\usepackage[all]{xy}
\usepackage{doi, url}
\usepackage{arydshln}
\usepackage{xcolor} 
\usepackage{framed}
\usepackage{tikz}
\usepackage{mathdots}
\usepackage{yhmath}
\usepackage{cancel}
\usepackage{siunitx}
\usepackage{array}
\usepackage{multirow}
\usepackage{tabularx}
\usepackage{extarrows}
\usepackage{booktabs}
\usepackage{diagbox}
\usepackage{makecell}
\usepackage[linesnumbered,ruled,vlined]{algorithm2e}
\usepackage{float}
\usepackage{hyperref}

\usetikzlibrary{fadings}
\usetikzlibrary{patterns}
\usetikzlibrary{shadows.blur}
\usetikzlibrary{shapes}

\renewcommand{\paragraph}[1]{\vspace{6pt} \noindent \textbf{#1}\xspace}

\theoremstyle{plain}
\newtheorem{theorem}{Theorem}[section]

\newtheorem{corollary}[theorem]{Corollary}
\newtheorem{lemma}[theorem]{Lemma}

\newtheorem{proposition}[theorem]{Proposition}

\theoremstyle{definition}
\newtheorem{problem}[theorem]{Problem}

\newtheorem{definition}[theorem]{Definition}

\allowdisplaybreaks[4]

\newcommand{\GL}{\mathrm{GL}}

\newcommand{\F}{\mathbb{F}}

\newcommand{\C}{\mathbb{C}}
\newcommand{\R}{\mathbb{R}}
\newcommand{\N}{\mathbb{N}}

\newcommand{\idmat}{\mathrm{I}}

\newcommand{\Tr}{\mathrm{Tr}}
\renewcommand{\exp}{\mathrm{exp}}
\newcommand{\trp}[1]{{#1}^{\mathrm{t}}}

\newcommand{\rk}{\mathrm{rk}}

\newcommand{\sym}{\mathrm{Sym}}

\newcommand{\poly}{\mathrm{poly}}

\newcommand{\M}{\mathrm{M}}
\newcommand{\rO}{\mathrm{O}}
\newcommand{\T}{\mathrm{T}}

\newcommand{\conj}{\mathrm{Conj}}

\newcommand{\tuple}[1]{\mathbf{#1}}
\newcommand{\tens}[1]{\mathtt{#1}}
\newcommand{\spa}[1]{\mathcal{#1}}

\newcommand{\cA}{\spa{A}}
\newcommand{\cB}{\spa{B}}
\newcommand{\cC}{\spa{C}}

\newcommand{\cO}{\spa{O}}
\newcommand{\cE}{\mathcal{E}}
\newcommand{\cL}{\mathcal{L}}
\newcommand{\cM}{\mathcal{M}}
\newcommand{\cN}{\mathcal{N}}
\newcommand{\cT}{\mathcal{T}}
\newcommand{\tA}{\tens{A}}
\newcommand{\tB}{\tens{B}}

\newcommand{\tC}{\tens{C}}

\newcommand{\vA}{\tuple{A}}
\newcommand{\vB}{\tuple{B}}

\newcommand{\Char}{\mathrm{Char}}

\newcommand{\stdb}{\vec{\mathrm{e}}}

\newcommand{\linspan}{\mathrm{span}}

\newcommand{\algprobm}[1]{{\sf #1}\xspace}
\newcommand{\GrI}{\algprobm{GraphIso}}
\newcommand{\TsI}{\algprobm{TensorIso}}
\newcommand{\OTI}{\algprobm{O-TensorIso}}
\newcommand{\AlgIso}{\algprobm{AlgIso}}
\newcommand{\MCConj}{\algprobm{MatCodeConj}}

\newcommand{\GpI}{\algprobm{GroupIso}}
\newcommand{\PolyIso}{\algprobm{PolyIso}}

\newcommand{\rin}{\in_{\mathrm{R}}}
\newcommand{\hull}{\mathrm{H}}

\renewcommand{\le}{\leqslant}
\renewcommand{\ge}{\geqslant}

\newcommand{\too}%
{\xrightarrow{\text{\raisebox{-3pt}{$\sim$}}\,}}

\mathchardef\mhyphen="2D

\usepackage[T1]{fontenc}
\def\DJ{{\hbox{D\kern-.8em\raise.15ex\hbox{--}\kern.35em}}}

%% file: references.bib
@misc{christoffersen2025gapssingularvaluessample,
      title={Gaps between Singular Values of Sample Covariance Matrices}, 
      author={Nicholas Christoffersen and Kyle Luh and Sean O'Rourke and Calum Shearer},
      year={2025},
      eprint={2502.15002},
      archivePrefix={arXiv},
      primaryClass={math.PR},
      url={https://arxiv.org/abs/2502.15002}, 
}

@incollection{Karp72,
  author    = {Richard M. Karp},
  title     = {Reducibility among Combinatorial Problems},
  booktitle = {Complexity of Computer Computations},
  publisher = {Plenum Press},
  pages     = {85--103},
  year      = {1972},
  doi       = {10.1007/978-1-4684-2001-2_9}
}

@article{BFF87,
  author  = {B{\'e}la Bollob{\'a}s and Trevor I. Fenner
             and Alan M. Frieze},
  title   = {An algorithm for finding {Hamilton} paths
             and cycles in random graphs},
  journal = {Combinatorica},
  volume  = {7},
  number  = {4},
  pages   = {327--341},
  year    = {1987},
  doi     = {10.1007/BF02579321}
}

@misc{GKW25,
  author        = {David Gamarnik and Eren C. K{\i}z{\i}lda{\u{g}}
                   and Lutz Warnke},
  title         = {Optimal Hardness of Online Algorithms
                   for Large Independent Sets},
  year          = {2025},
  eprint        = {2504.11450},
  archivePrefix = {arXiv},
  primaryClass  = {cs.DS},
  note          = {Version 2, revised January 2026},
  doi           = {10.48550/arXiv.2504.11450},
  url           = {https://arxiv.org/abs/2504.11450v2}
}

@article {MR96677,
    AUTHOR = {Fine, N. J. and Herstein, I. N.},
     TITLE = {The probability that a matrix be nilpotent},
   JOURNAL = {Illinois J. Math.},
  FJOURNAL = {Illinois Journal of Mathematics},
    VOLUME = {2},
      YEAR = {1958},
     PAGES = {499--504},
      ISSN = {0019-2082},
   MRCLASS = {15.00},
  MRNUMBER = {96677},
MRREVIEWER = {B.\ W.\ Jones},
       URL = {http://projecteuclid.org/euclid.ijm/1255454112},
}

@misc{CEMQ26,
      title={Random tensor isomorphism under orthogonal and unitary actions}, 
      author={Jeremy Chizewer and Samuel Everett and Deven Mithal and Youming Qiao},
      year={2026},
      eprint={2603.27128},
      archivePrefix={arXiv},
      primaryClass={cs.CC},
      url={https://arxiv.org/abs/2603.27128}, 
}

@book {MR230728,
    AUTHOR = {Steinberg, Robert},
     TITLE = {Endomorphisms of linear algebraic groups},
    SERIES = {Memoirs of the American Mathematical Society},
    VOLUME = {No. 80},
 PUBLISHER = {American Mathematical Society, Providence, RI},
      YEAR = {1968},
     PAGES = {108},
   MRCLASS = {14.50 (22.00)},
  MRNUMBER = {230728},
MRREVIEWER = {E.\ Abe},
}

@book{10.1093/oso/9780198502951.001.0001,
    author = {Hirschfeld, J. W. P.},
    title = {Projective Geometries over Finite Fields},
    publisher = {Oxford University Press},
    year = {1998},
    isbn = {9780198502951},
    doi = {10.1093/oso/9780198502951.001.0001},
    url = {https://doi.org/10.1093/oso/9780198502951.001.0001},
}

@incollection{bollobas1982distinguishing,
  title={Distinguishing vertices of random graphs},
  author={Bollob{\'a}s, B{\'e}la},
  booktitle={North-Holland Mathematics Studies},
  volume={62},
  pages={33--49},
  year={1982},
  publisher={Elsevier}
}

@book{vzGG13,
  author    = {Joachim von zur Gathen and J{\"u}rgen Gerhard},
  title     = {Modern Computer Algebra},
  edition   = {3},
  publisher = {Cambridge University Press},
  address   = {Cambridge},
  year      = {2013}
}

@article{CantorZassenhaus81,
  author  = {David G. Cantor and Hans Zassenhaus},
  title   = {A New Algorithm for Factoring Polynomials over Finite Fields},
  journal = {Mathematics of Computation},
  volume  = {36},
  number  = {154},
  pages   = {587--592},
  year    = {1981}
}

@article{linial2017rigidity,
  title={On the rigidity of sparse random graphs},
  author={Linial, Nati and Mosheiff, Jonathan},
  journal={Journal of Graph Theory},
  volume={85},
  number={2},
  pages={466--480},
  year={2017},
  publisher={Wiley Online Library}
}

@article{czajka2008improved,
  title={Improved random graph isomorphism},
  author={Czajka, Tomek and Pandurangan, Gopal},
  journal={Journal of Discrete Algorithms},
  volume={6},
  number={1},
  pages={85--92},
  year={2008},
  publisher={Elsevier}
}

@article{10.1214/21-AOP1520,
author = {Kurt Johansson and Gaultier Lambert},
title = {{Multivariate normal approximation for traces of random unitary matrices}},
volume = {49},
journal = {The Annals of Probability},
number = {6},
publisher = {Institute of Mathematical Statistics},
pages = {2961 -- 3010},
year = {2021},
doi = {10.1214/21-AOP1520},
URL = {https://doi.org/10.1214/21-AOP1520}
}

@article{10.1214/20-PS346,
author = {Van H. Vu},
title = {{Recent progress in combinatorial random matrix theory}},
volume = {18},
journal = {Probability Surveys},
number = {none},
publisher = {Institute of Mathematical Statistics and Bernoulli Society},
pages = {179 -- 200},
year = {2021},
doi = {10.1214/20-PS346},
URL = {https://doi.org/10.1214/20-PS346}
}

@Article{Fulman2015,
author={Fulman, Jason
and Goldstein, Larry},
title={Stein's method and the rank distribution of random matrices over finite fields},
journal={The Annals of Probability},
year={2015},
month={May},
day={01},
volume={43},
number={3},
pages={1274-1314},
doi={10.1214/13-AOP889},
url={https://doi.org/10.1214/13-AOP889}
}

@inproceedings{DBLP:conf/eurocrypt/NarayananQT24,
  author       = {Anand Kumar Narayanan and
                  Youming Qiao and
                  Gang Tang},
  title        = {Algorithms for Matrix Code and Alternating Trilinear Form Equivalences
                  via New Isomorphism Invariants},
  booktitle    = {Advances in Cryptology - {EUROCRYPT} 2024 - 43rd Annual International
                  Conference on the Theory and Applications of Cryptographic Techniques, 2024, Proceedings, Part {III}},
  series       = {Lecture Notes in Computer Science},
  volume       = {14653},
  pages        = {160--187},
  publisher    = {Springer},
  year         = {2024},
  url          = {https://doi.org/10.1007/978-3-031-58734-4_6},
  doi          = {10.1007/978-3-031-58734-4_6},
  bibsource    = {dblp computer science bibliography, https://dblp.org}
}

@InProceedings{10.1007/978-3-540-31856-9_1,
author="Agrawal, Manindra
and Saxena, Nitin",
title="Automorphisms of Finite Rings and Applications to Complexity of Problems",
booktitle="Proceedings of the 22nd Annual Symposium on Theoretical Aspects of Computer Science",
year="2005",
pages="1--17",
isbn="978-3-540-31856-9"
}

@Article{Kayal2006,
author={Kayal, Neeraj
and Saxena, Nitin},
title={Complexity of Ring Morphism Problems},
journal={{Computational Complexity}},
year={2006},
month={Dec},
day={01},
volume={15},
number={4},
pages={342-390},
issn={1420-8954},
doi={10.1007/s00037-007-0219-8},
url={https://doi.org/10.1007/s00037-007-0219-8}
}

@article{Sendrierhull,
author = {Sendrier, Nicolas},
title = {On the Dimension of the Hull},
journal = {SIAM Journal on Discrete Mathematics},
volume = {10},
number = {2},
pages = {282-293},
year = {1997},
doi = {10.1137/S0895480195294027}
}

@article{10.1215/ijm/1255629831,
author = {Murray Gerstenhaber},
title = {{On the number of nilpotent matrices with coefficients in a finite field}},
volume = {5},
journal = {Illinois Journal of Mathematics},
number = {2},
publisher = {Duke University Press},
pages = {330 -- 333},
year = {1961},
doi = {10.1215/ijm/1255629831},
URL = {https://doi.org/10.1215/ijm/1255629831}
}

@article{10.1215/ijm/1255629830,
author = {Irving Reiner},
title = {{On the number of matrices with given characteristic polynomial}},
volume = {5},
journal = {Illinois Journal of Mathematics},
number = {2},
publisher = {Duke University Press},
pages = {324 -- 329},
year = {1961},
doi = {10.1215/ijm/1255629830},
URL = {https://doi.org/10.1215/ijm/1255629830}
}

@book{Lidl_Niederreiter_1996, 
place={Cambridge}, 
edition={2}, 
series={Encyclopedia of Mathematics and its Applications}, 
title={{Finite Fields}}, 
publisher={Cambridge University Press}, 
author={Lidl, Rudolf and Niederreiter, Harald}, 
year={1996}, 
collection={Encyclopedia of Mathematics and its Applications}
}

@article{GK24,
  title={Equidistribution of high traces of random matrices over finite fields and cancellation in character sums of high conductor},
  author={Gorodetsky, Ofir and Kovaleva, Valeriya},
  journal={Bulletin of the London Mathematical Society},
  volume={56},
  number={7},
  pages={2315--2337},
  year={2024},
  publisher={Wiley Online Library}
}

@article{DS94,
  title={On the eigenvalues of random matrices},
  author={Diaconis, Persi and Shahshahani, Mehrdad},
  journal={Journal of Applied Probability},
  volume={31},
  number={A},
  pages={49--62},
  year={1994},
  publisher={Cambridge University Press}
}

@article {MR4273172,
    AUTHOR = {Gorodetsky, Ofir and Rodgers, Brad},
     TITLE = {Traces of powers of matrices over finite fields},
   JOURNAL = {Trans. Amer. Math. Soc.},
  FJOURNAL = {Transactions of the American Mathematical Society},
    VOLUME = {374},
      YEAR = {2021},
    NUMBER = {7},
     PAGES = {4579--4638},
      ISSN = {0002-9947,1088-6850},
   MRCLASS = {60B15 (11T55 15B52 20G40 60B20)},
  MRNUMBER = {4273172},
MRREVIEWER = {Florent\ Benaych-Georges},
       DOI = {10.1090/tran/8337},
       URL = {https://doi.org/10.1090/tran/8337},
}

@article{NEUMANN_PRAEGER_1998, 
title={{Derangements and Eigenvalue-free Elements in Finite Classical Groups}},
volume={58}, 
DOI={10.1112/S0024610798006772}, 
number={3}, 
journal={Journal of the London Mathematical Society}, 
author={Neumann, Peter M. and Praeger, Cheryl E.}, 
year={1998}, 
pages={564–586}}

@article{NP95,
  title={Cyclic matrices over finite fields},
  author={Neumann, Peter M. and Praeger, Cheryl E.},
  journal={Journal of the London Mathematical Society},
  volume={52},
  number={2},
  pages={263--284},
  year={1995},
  publisher={Oxford University Press}
}

@article{Sen02,
  title={Finding the permutation between equivalent linear codes: The support splitting algorithm},
  author={Sendrier, Nicolas},
  journal={IEEE Transactions on Information Theory},
  volume={46},
  number={4},
  pages={1193--1203},
  year={2000},
  publisher={IEEE}
}

@inproceedings{CL25,
  author       = {Alain Couvreur and
                  Christophe Levrat},
  title        = {Highway to Hull: An Algorithm for Solving the General Matrix Code
                  Equivalence Problem},
  booktitle    = {Advances in Cryptology - {CRYPTO} 2025 - 45th Annual International
                  Cryptology Conference, 2025,
                  Proceedings, Part {I}},
  series       = {Lecture Notes in Computer Science},
  volume       = {16000},
  pages        = {253--283},
  publisher    = {Springer},
  year         = {2025},
  url          = {https://doi.org/10.1007/978-3-032-01855-7_9},
  doi          = {10.1007/978-3-032-01855-7_9},
  bibsource    = {dblp computer science bibliography, https://dblp.org}
}

@inproceedings{Beu23,
	author       = {Ward Beullens},
	title        = {Graph-Theoretic Algorithms for the Alternating Trilinear Form Equivalence
	Problem},
	booktitle    = {Advances in Cryptology - {CRYPTO} 2023 - 43rd Annual International Cryptology Conference, Proceedings, Part {III}},
	series       = {Lecture Notes in Computer Science},
	volume       = {14083},
	pages        = {101--126},
	publisher    = {Springer},
	year         = {2023},
	doi          = {10.1007/978-3-031-38548-3_4},
	bibsource    = {dblp computer science bibliography, https://dblp.org}
}

@inproceedings{AFMP20,
  author    = {Navid Alamati and
               Luca De Feo and
               Hart Montgomery and
               Sikhar Patranabis},
  title     = {Cryptographic Group Actions and Applications},
  booktitle = {Advances in Cryptology - {ASIACRYPT} 2020 - 26th International 
  Conference
               on the Theory and Application of Cryptology and Information 
               Security,
               Proceedings, Part {II}},
  series    = {Lecture Notes in Computer Science},
  volume    = {12492},
  pages     = {411--439},
  publisher = {Springer},
  year      = {2020},
  url       = {https://doi.org/10.1007/978-3-030-64834-3_14},
  doi       = {10.1007/978-3-030-64834-3_14},
  bibsource = {dblp computer science bibliography, https://dblp.org}
}

@article{MEDSspecs,
	title={{MEDS}: Matrix Equivalence Digital Signature (2023)},
	author={Chou, Tung and Niederhagen, Ruben and Persichetti, Edoardo and Ran, Lars and Randrianarisoa, Tovohery Hajatiana and Reijnders, Krijn and Samardjiska, Simona and Trimoska, Monika},
	journal={Submission to the NIST Digital Signature Scheme standardization process}, 
	year = {2023},
	url = {https://www.meds-pqc.org/}
}

@article{ALTEQspecs,
	title={The {ALTEQ} Signature Scheme: Algorithm Specifications and Supporting Documentation},
	author={Bl{\"a}ser, Markus and Duong, Dung Hoang and Narayanan, Anand Kumar and Plantard, Thomas and Qiao, Youming and Sipasseuth, Arnaud and Tang, Gang},
	journal={Submission to the NIST Digital Signature Scheme standardization process},
	year={2023},
	url = {https://pqcalteq.github.io/}
}

@inproceedings{RS24,
  author       = {Lars Ran and
                  Simona Samardjiska},
  title        = {Rare Structures in Tensor Graphs - Bermuda Triangles for Cryptosystems
                  Based on the Tensor Isomorphism Problem},
  booktitle    = {Advances in Cryptology - {ASIACRYPT} 2024 - 30th International Conference
                  on the Theory and Application of Cryptology and Information Security,
                 2024, Proceedings, Part {VIII}},
  series       = {Lecture Notes in Computer Science},
  volume       = {15491},
  pages        = {66--96},
  publisher    = {Springer},
  year         = {2024},
  url          = {https://doi.org/10.1007/978-981-96-0944-4_3},
  doi          = {10.1007/978-981-96-0944-4_3},
  bibsource    = {dblp computer science bibliography, https://dblp.org}
}

@article{PhysRevLett.104.020504,
  title = {Local Unitary Equivalence of Multipartite Pure States},
  author = {Kraus, B.},
  journal = {Phys. Rev. Lett.},
  volume = {104},
  issue = {2},
  pages = {020504},
  numpages = {4},
  year = {2010},
  month = {Jan},
  publisher = {American Physical Society},
  doi = {10.1103/PhysRevLett.104.020504},
  url = {https://link.aps.org/doi/10.1103/PhysRevLett.104.020504}
}

@article{PhysRevA.82.032121,
  title = {Local unitary equivalence and entanglement of multipartite pure states},
  author = {Kraus, B.},
  journal = {Phys. Rev. A},
  volume = {82},
  issue = {3},
  pages = {032121},
  numpages = {14},
  year = {2010},
  month = {Sep},
  publisher = {American Physical Society},
  doi = {10.1103/PhysRevA.82.032121},
  url = {https://link.aps.org/doi/10.1103/PhysRevA.82.032121}
}

@inproceedings{Sun23,
	author       = {Xiaorui Sun},
	title        = {Faster Isomorphism for $p$-Groups of Class 2 and Exponent
	$p$},
	booktitle    = {Proceedings of the 55th Annual {ACM} Symposium on Theory of Computing,
	{STOC} 2023},
	pages        = {433--440},
	publisher    = {{ACM}},
	year         = {2023},
	url          = {https://doi.org/10.1145/3564246.3585250},
	doi          = {10.1145/3564246.3585250},
	bibsource    = {dblp computer science bibliography, https://dblp.org}
}

@inproceedings{IMQSZ24,
	author       = {Gábor Ivanyos and Euan Jacob Mendoza and Youming Qiao and Xiaorui Sun and Chuanqi Zhang},
	title        = {Faster Isomorphism Testing of $p$-Groups of {Frattini} Class-2},
	booktitle    = {65th {IEEE} Annual Symposium on Foundations of Computer Science, 
	{FOCS} 2024},
	pages        = {1408--1424},
	publisher    = {{IEEE}},
	year         = {2024}
}

@article{GMW91,
	Author = {Oded Goldreich and Silvio Micali and Avi Wigderson},
	Bibsource = {dblp computer science bibliography, https://dblp.org},
	Doi = {10.1145/116825.116852},
	Journal = {J. {ACM}},
	Number = {3},
	Pages = {691--729},
	Title = {Proofs that Yield Nothing But Their Validity for All Languages in 
	{NP} Have Zero-Knowledge Proof Systems},
	Volume = {38},
	Year = {1991}}

@inproceedings{FS86,
	Author = {Amos Fiat and Adi Shamir},
	Booktitle = {Advances in Cryptology -- {CRYPTO} 1986},
	Pages = {186--194},
	Title = {How to Prove Yourself: Practical Solutions to Identification and 
	Signature Problems},
	Year = {1986}}

@article{BL08,
	Author = {Peter A. Brooksbank and Eugene M. Luks},
	Doi = {10.1016/j.jalgebra.2008.07.014},
	Issn = {0021-8693},
	Journal = {Journal of Algebra},
	Number = {11},
	Pages = {4020 - 4029},
	Title = {Testing isomorphism of modules},
	Volume = {320},
	Year = {2008}}

@inproceedings{MEDS-paper,
	author       = {Tung Chou and
	Ruben Niederhagen and
	Edoardo Persichetti and
	Tovohery Hajatiana Randrianarisoa and
	Krijn Reijnders and
	Simona Samardjiska and
	Monika Trimoska},
	title        = {Take Your {MEDS:} Digital Signatures from Matrix Code Equivalence},
	booktitle    = {Progress in Cryptology - {AFRICACRYPT} 2023 - 14th International Conference
	on Cryptology in Africa, 2023, Proceedings},
	series       = {Lecture Notes in Computer Science},
	volume       = {14064},
	pages        = {28--52},
	publisher    = {Springer},
	year         = {2023},
	url          = {https://doi.org/10.1007/978-3-031-37679-5_2},
	doi          = {10.1007/978-3-031-37679-5_2},
	bibsource    = {dblp computer science bibliography, https://dblp.org}
}

@inproceedings{Pat96,
	author    = {Jacques Patarin},
	title     = {Hidden Fields Equations {(HFE)} and Isomorphisms of Polynomials 
	{(IP):}
	Two New Families of Asymmetric Algorithms},
	booktitle = {Advances in Cryptology - {EUROCRYPT} '96, International Conference
	on the Theory and Application of Cryptographic Techniques, 
	1996, Proceeding},
	pages     = {33--48},
	year      = {1996},
	doi       = {10.1007/3-540-68339-9_4},
}

@article{Kung81,
  title={The cycle structure of a linear transformation over a finite field},
  author={Kung, Joseph PS},
  journal={Linear Algebra and its Applications},
  volume={36},
  pages={141--155},
  year={1981},
  publisher={Elsevier}
}

@article{Sto88,
  title={Some asymptotic results on finite vector spaces},
  author={Stong, Richard},
  journal={Advances in Applied Mathematics},
  volume={9},
  number={2},
  pages={167--199},
  year={1988},
  publisher={Elsevier}
}

@article{Ful02,
  title={Random matrix theory over finite fields},
  author={Fulman, Jason},
  journal={Bulletin of the American Mathematical Society},
  volume={39},
  number={1},
  pages={51--85},
  year={2002}
}

@inproceedings{AS06,
    author = {Manindra Agrawal and Nitin Saxena},
    title = {Equivalence of {$\mathbb{F}$} -- Algebras and Cubic Forms},
    booktitle = {Proceedings of the 23rd Annual Symposium on Theoretical Aspects of Computer Science},
    pages     = {115--126},
    year      = {2006},
}

@article{GQ17,
	TITLE = {Algorithms for group isomorphism via group extensions and cohomology},
	AUTHOR = {Grochow, Joshua A. and Qiao, Youming},
	JOURNAL = {SIAM Journal on Computing},
	YEAR = {2017},
	VOLUME = {46},
	NUMBER = {4},
	PAGES = {1153--1216},
	NOTE = {Preliminary version in IEEE Conference on Computational Complexity (CCC) 2014.},
	DOI = {10.1137/15M1009767},
}

@article{TI1,
       AUTHOR = {Grochow, Joshua A. and Qiao, Youming},
        TITLE = {On the Complexity of Isomorphism Problems for Tensors, Groups, 
                 and Polynomials {I}: {Tensor} {Isomorphism}-Completeness},
      JOURNAL = {SIAM J. Comput.},
      FJOUNAL = {SIAM Journal on Computing},
       VOLUME = {52},
        ISSUE = {2},
        PAGES = {568--617},
         YEAR = {2023},
          DOI = {10.1137/21M1441110},
         NOTE = {Part of the preprint \cite{GQarxiv}. 
                 Preliminary version appeared at ITCS '21.},
}

@article{TI2,
	author       = {Joshua A. Grochow and
	Youming Qiao},
	title        = {On \emph{p}-Group Isomorphism: Search-to-Decision, Counting-to-Decision,
	and Nilpotency Class Reductions via Tensors},
	journal      = {{ACM} Trans. Comput. Theory},
	volume       = {16},
	number       = {1},
	pages        = {2:1--2:39},
	year         = {2024},
	url          = {https://doi.org/10.1145/3625308},
	doi          = {10.1145/3625308},
	bibsource    = {dblp computer science bibliography, https://dblp.org}
}

@misc{GQarxiv,
 author = {Grochow, Joshua A. and Qiao, Youming},
 title = {Isomorphism problems for tensors, groups, and cubic forms: completeness and reductions},
 year = {2019},
 howpublished={arXiv:\href{https://arxiv.org/abs/1907.00309}{1907.00309} [cs.CC]},
}

@article{BFP15,
	Author = {J{\'{e}}r{\'{e}}my Berthomieu and Jean{-}Charles Faug{\`{e}}re and Ludovic Perret},
	Journal = {Journal of Complexity},
	Number = {4},
	Pages = {590--616},
	Title = {Polynomial-time algorithms for quadratic isomorphism of polynomials: The regular case},
	Volume = {31},
	Year = {2015}}

@inproceedings{MPG13,
	Author = {Gilles Macario{-}Rat and J{\'{e}}r{\^{o}}me Pl{\^{u}}t and Henri Gilbert},
	Booktitle = {Advances in Cryptology - {ASIACRYPT} 2013},
	Pages = {117--133},
	Title = {New Insight into the Isomorphism of Polynomial Problem {IP1S} and Its Use in Cryptography},
	Year = {2013}}

@inproceedings{BFV13,
	Author = {Charles Bouillaguet and Pierre{-}Alain Fouque and Amandine V{\'{e}}ber},
	Bibsource = {dblp computer science bibliography, http://dblp.org},
	Booktitle = {Advances in Cryptology - {EUROCRYPT} 2013},
	Pages = {211--227},
	Title = {Graph-Theoretic Algorithms for the ``{Isomorphism of Polynomials}'' Problem},
	Year = {2013}}

@inproceedings{ALTEQ-paper,
	author    = {Gang Tang and
	Dung Hoang Duong and
	Antoine Joux and
	Thomas Plantard and
	Youming Qiao and
	Willy Susilo},
	title     = {Practical Post-Quantum Signature Schemes from Isomorphism Problems
	of Trilinear Forms},
	booktitle = {Advances in Cryptology - {EUROCRYPT} 2022 - 41st Annual International
	Conference on the Theory and Applications of Cryptographic Techniques, 2022, Proceedings, Part {III}},
	series    = {Lecture Notes in Computer Science},
	volume    = {13277},
	pages     = {582--612},
	publisher = {Springer},
	year      = {2022},
	url       = {https://doi.org/10.1007/978-3-031-07082-2_21},
	doi       = {10.1007/978-3-031-07082-2_21},
	bibsource = {dblp computer science bibliography, https://dblp.org}
}

@inproceedings{JQSY19,
	author    = {Zhengfeng Ji and
	Youming Qiao and
	Fang Song and
	Aaram Yun},
	title     = {General Linear Group Action on Tensors: {A} Candidate for 
	Post-quantum
	Cryptography},
	booktitle = {Theory of Cryptography - 17th International Conference, {TCC} 2019, Proceedings, Part {I}},
	pages     = {251--281},
	year      = {2019},
	url       = {https://doi.org/10.1007/978-3-030-36030-6_11},
	doi       = {10.1007/978-3-030-36030-6_11},
	bibsource = {dblp computer science bibliography, https://dblp.org}
}

@article{IKS10,
  title={Deterministic polynomial time algorithms for matrix completion problems},
  author={Ivanyos, G{\'a}bor and Karpinski, Marek and Saxena, Nitin},
  journal={SIAM journal on computing},
  volume={39},
  number={8},
  pages={3736--3751},
  year={2010},
  publisher={SIAM}
}

@article{BO08,
  title={Constructing the group preserving a system of forms},
  author={Brooksbank, Peter A and O'Brien, Eamonn A.},
  journal={International Journal of Algebra and Computation},
  volume={18},
  number={02},
  pages={227--241},
  year={2008},
  publisher={World Scientific}
}

@article{Bae38,
	title={Groups with abelian central quotient group},
	author={Baer, Reinhold},
	journal={Transactions of the American Mathematical Society},
	volume={44},
	number={3},
	pages={357--386},
	year={1938},
	publisher={JSTOR}
}

@inproceedings{BLQW20,
	author    = {Peter A. Brooksbank and
	Yinan Li and
	Youming Qiao and
	James B. Wilson},
	title     = {Improved Algorithms for Alternating Matrix Space Isometry: From Theory
	to Practice},
	booktitle = {28th Annual European Symposium on Algorithms, {ESA} 2020},
	series    = {LIPIcs},
	volume    = {173},
	pages     = {26:1--26:15},
	publisher = {Schloss Dagstuhl - Leibniz-Zentrum f{\"{u}}r Informatik},
	year      = {2020},
	url       = {https://doi.org/10.4230/LIPIcs.ESA.2020.26},
	doi       = {10.4230/LIPIcs.ESA.2020.26},
	bibsource = {dblp computer science bibliography, https://dblp.org}
}

@article{IQ19,
	author = {Ivanyos, G\'abor and Qiao, Youming},
	title = {Algorithms Based on $*$-Algebras, and Their Applications to Isomorphism of 
	Polynomials with One Secret, Group Isomorphism, and Polynomial Identity Testing},
	journal = {{SIAM Journal on Computing}},
	volume = {48},
	number = {3},
	pages = {926-963},
	year = {2019},
	doi = {10.1137/18M1165682},
}

@inproceedings{LQ17,
	author    = {Yinan Li and
	Youming Qiao},
	editor    = {Chris Umans},
	title     = {Linear Algebraic Analogues of the Graph Isomorphism Problem and the
	{Erd{\H{o}}s}--{R{\'{e}}nyi} Model},
	booktitle = {58th {IEEE} Annual Symposium on Foundations of Computer Science, 
	{FOCS}
	2017},
	pages     = {463--474},
	publisher = {{IEEE} Computer Society},
	year      = {2017},
	url       = {https://doi.org/10.1109/FOCS.2017.49},
	doi       = {10.1109/FOCS.2017.49},
	bibsource = {dblp computer science bibliography, https://dblp.org}
}

@inproceedings{TI5,
	author       = {Joshua A. Grochow and
	Youming Qiao and
	Katherine E. Stange and
	Xiaorui Sun},
	title        = {On the Complexity of Isomorphism Problems for Tensors, Groups, and
	Polynomials {V:} Over Commutative Rings},
	booktitle    = {Proceedings of the 57th Annual ACM Symposium on Theory of Computing},
	pages        = {777--784},
	publisher    = {{ACM}},
	year         = {2025},
	url          = {https://doi.org/10.1145/3717823.3718286},
	doi          = {10.1145/3717823.3718286},
	bibsource    = {dblp computer science bibliography, https://dblp.org}
}

@inproceedings{TI3,
	author       = {Zhili Chen and
	Joshua A. Grochow and
	Youming Qiao and
	Gang Tang and
	Chuanqi Zhang},
	title        = {On the Complexity of Isomorphism Problems for Tensors, Groups, and	Polynomials {III:} Actions by Classical Groups},
	booktitle    = {15th Innovations in Theoretical Computer Science Conference, {ITCS} 2024},
	series       = {LIPIcs},
	volume       = {287},
	pages        = {31:1--31:23},
	publisher    = {Schloss Dagstuhl - Leibniz-Zentrum f{\"{u}}r Informatik},
	year         = {2024},
	url          = {https://doi.org/10.4230/LIPIcs.ITCS.2024.31},
	doi          = {10.4230/LIPICS.ITCS.2024.31},
	bibsource    = {dblp computer science bibliography, https://dblp.org}
}

@inproceedings{TI4,
author = {Grochow, Joshua A. and Qiao, Youming},
title = {On the Complexity of Isomorphism Problems for Tensors, Groups, and Polynomials {IV}: {Linear}-Length Reductions and Their Applications},
year = {2025},
isbn = {9798400715105},
publisher = {ACM},
url = {https://doi.org/10.1145/3717823.3718282},
doi = {10.1145/3717823.3718282},
booktitle = {Proceedings of the 57th Annual ACM Symposium on Theory of Computing},
pages = {766 -- 776},
numpages = {11},
location = {Prague, Czechia},
series = {STOC 2025}
}

@inproceedings{Bab16,
	author    = {L{\'{a}}szl{\'{o}} Babai},
	title     = {Graph isomorphism in quasipolynomial time [extended abstract]},
	booktitle = {Proceedings of the 48th Annual {ACM} {SIGACT} Symposium on Theory
	of Computing, {STOC} 2016},
	pages     = {684--697},
	publisher = {{ACM}},
	year      = {2016},
	url       = {https://doi.org/10.1145/2897518.2897542},
	doi       = {10.1145/2897518.2897542},
	bibsource = {dblp computer science bibliography, https://dblp.org}
}

@inproceedings{AKM25,
  author       = {Michael Anastos and
                  Matthew Kwan and
                  Benjamin Moore},
  title        = {Smoothed Analysis for Graph Isomorphism},
  booktitle    = {Proceedings of the 57th Annual {ACM} Symposium on Theory of Computing,
                  {STOC} 2025},
  pages        = {2098--2106},
  publisher    = {{ACM}},
  year         = {2025},
  url          = {https://doi.org/10.1145/3717823.3718173},
  doi          = {10.1145/3717823.3718173},
  bibsource    = {dblp computer science bibliography, https://dblp.org}
}

@inproceedings{BK79,
	author       = {L{\'{a}}szl{\'{o}} Babai and
	Ludek Ku{\v c}era},
	title        = {Canonical Labelling of Graphs in Linear Average Time},
	booktitle    = {20th Annual Symposium on Foundations of Computer Science, 1979},
	pages        = {39--46},
	publisher    = {{IEEE} Computer Society},
	year         = {1979},
	doi          = {10.1109/SFCS.1979.8},
	bibsource    = {dblp computer science bibliography, https://dblp.org}
}

@article{BES80,
  author    = {L{\'{a}}szl{\'{o}} Babai and
               Paul Erd{\H{o}}s and
               Stanley M. Selkow},
  title     = {Random Graph Isomorphism},
  journal   = {{SIAM} J. Comput.},
  volume    = {9},
  number    = {3},
  pages     = {628--635},
  year      = {1980},
  doi       = {10.1137/0209047},
  bibsource = {dblp computer science bibliography, http://dblp.org}
}

@inproceedings{BL83,
	author       = {L{\'{a}}szl{\'{o}} Babai and
	Eugene M. Luks},
	title        = {Canonical Labeling of Graphs},
	booktitle    = {Proceedings of the 15th Annual {ACM} Symposium on Theory of Computing,
	25-27 April, 1983},
	pages        = {171--183},
	publisher    = {{ACM}},
	year         = {1983},
	url          = {https://doi.org/10.1145/800061.808746},
	doi          = {10.1145/800061.808746},
	bibsource    = {dblp computer science bibliography, https://dblp.org}
}

@article{MP14,
title = "Practical graph isomorphism, {II} ",
journal = "Journal of Symbolic Computation ",
volume = "60",
pages = "94 - 112",
year = "2014",
note = "",
issn = "0747-7171",
doi = "10.1016/j.jsc.2013.09.003",
author = "Brendan D. McKay and Adolfo Piperno"
}

@article{McK80,
  title={Practical graph isomorphism},
  author={McKay, Brendan D.},
  volume = {30},
  year={1981},
  pages={45--87},
  journal={Congressus Numerantium}
}

@inproceedings{BC86,
  author       = {Gilles Brassard and
                  Claude Cr{\'{e}}peau},
  title        = {Non-Transitive Transfer of Confidence: {A} Perfect Zero-Knowledge
                  Interactive Protocol for {SAT} and Beyond},
  booktitle    = {27th Annual Symposium on Foundations of Computer Science, 1986},
  pages        = {188--195},
  publisher    = {{IEEE} Computer Society},
  year         = {1986},
  url          = {https://doi.org/10.1109/SFCS.1986.33},
  doi          = {10.1109/SFCS.1986.33},
  bibsource    = {dblp computer science bibliography, https://dblp.org}
}

@inproceedings{BY90,
  author       = {Gilles Brassard and
                  Moti Yung},
  title        = {One-Way Group Actions},
  booktitle    = {Advances in Cryptology - {CRYPTO} '90, 10th Annual International Cryptology
                  Conference, 1990, Proceedings},
  series       = {Lecture Notes in Computer Science},
  volume       = {537},
  pages        = {94--107},
  publisher    = {Springer},
  year         = {1990},
  url          = {https://doi.org/10.1007/3-540-38424-3_7},
  doi          = {10.1007/3-540-38424-3_7},
  bibsource    = {dblp computer science bibliography, https://dblp.org}
}

@article {GQT,
    AUTHOR = {Grochow, Joshua A. and Qiao, Youming and Tang, Gang},
     TITLE = {Average-case algorithms for testing isomorphism of
              polynomials, algebras, and multilinear forms},
   JOURNAL = {Journal of Groups, Complexity, Cryptology},
    VOLUME = {14},
      YEAR = {2022},
    NUMBER = {1},
     PAGES = {21},
      ISSN = {1867-1144},
   MRCLASS = {11Y16 (11E76 11T06 15A69 15B52 68W20)},
  MRNUMBER = {4468830},
  DOI = {10.46298/jgcc.2022.14.1.9431},
  NOTE = {Extended abstract appeared in STACS 2021},
}

@inproceedings{GrochowLie,
    AUTHOR = {Grochow, Joshua A.},
     TITLE = {Matrix {Lie} algebra isomorphism},
 BOOKTITLE = {IEEE Conference on Computational Complexity (CCC)},
      YEAR = {2012},
     PAGES = {203--213},
       DOI = {10.1109/CCC.2012.34},
}

@article{BW12,
  author  = {Peter A. Brooksbank and James B. Wilson},
  title   = {Computing isometry groups of {Hermitian} maps},
  journal = {Transactions of the American Mathematical Society},
  volume  = {364},
  number  = {4},
  pages   = {1975--1996},
  year    = {2012},
  doi     = {10.1090/S0002-9947-2011-05388-2}
}

@article{Ser88,
  author  = {Vladimir V. Sergeichuk},
  title   = {Classification problems for systems of forms and linear mappings},
  journal = {Mathematics of the USSR-Izvestiya},
  volume  = {31},
  number  = {3},
  pages   = {481--501},
  year    = {1988},
  doi     = {10.1070/IM1988v031n03ABEH001086},
  url     = {https://arxiv.org/abs/0801.0823},
  note    = {Author's corrected version available as arXiv:0801.0823}
}

@misc{Ser10,
  author        = {Vladimir V. Sergeichuk},
  title         = {Canonical matrices of forms and pairs of forms over finite and $p$-adic fields},
  year          = {2010},
  eprint        = {1011.3142},
  archivePrefix = {arXiv},
  primaryClass  = {math.RT},
  url           = {https://arxiv.org/abs/1011.3142}
}

@book{Tay92,
  author    = {Donald E. Taylor},
  title     = {The Geometry of the Classical Groups},
  series    = {Sigma Series in Pure Mathematics},
  volume    = {9},
  publisher = {Heldermann Verlag},
  address   = {Berlin},
  year      = {1992},
  isbn      = {3-88538-009-9},
  url       = {https://mail.maths.usyd.edu.au/u/don/papers/gcg.pdf}
}
